\documentclass[12pt]{article}
\usepackage{times}

\usepackage[utf8]{inputenc} 
\usepackage[T1]{fontenc}    
\usepackage{url}            
\usepackage{booktabs}       
\usepackage{amsfonts}       
\usepackage{amsmath}
\usepackage{amssymb}
\usepackage{amsthm}
\usepackage{bbm}
\usepackage{nicefrac}       
\usepackage{microtype}      
\usepackage[usenames,dvipsnames]{xcolor}         
\usepackage[toc,page, title, titletoc]{appendix}
\usepackage{algorithm}
\usepackage{algpseudocode}
\usepackage{caption}
\usepackage{subcaption}
\usepackage{graphicx}
\usepackage{tikz}
\usepackage{tkz-tab}
\usetikzlibrary{shapes.multipart}
\usepackage{bigints}
\usepackage{centernot}
\usepackage{wrapfig}
\usepackage[inline, shortlabels]{enumitem}
\usepackage{mdframed}
\usepackage{array}
\newlength\mylength
\newcolumntype{C}[1]{>{\centering\arraybackslash}p{#1}}
\newcolumntype{P}[1]{>{\centering\arraybackslash}p{#1}}
\usepackage{longtable}
\usepackage{tabularray}
\allowdisplaybreaks

\usepackage{etoolbox}
\newcommand{\zerodisplayskips}{%
  \setlength{\abovedisplayskip}{6pt}%
  \setlength{\belowdisplayskip}{6pt}%
  \setlength{\abovedisplayshortskip}{3pt}%
  \setlength{\belowdisplayshortskip}{3pt}}
\appto{\normalsize}{\zerodisplayskips}
\appto{\small}{\zerodisplayskips}
\appto{\footnotesize}{\zerodisplayskips}

\usepackage{xr-hyper}
\definecolor{darkblue}{rgb}{0.0, 0.0, 0.55}	
\definecolor{cobalt}{rgb}{0.0, 0.28, 0.67}
\usepackage[pagebackref=true]{hyperref}
\renewcommand*{\backrefalt}[4]{%
    \ifcase #1 \footnotesize{(Not cited.)}%
    \or        \footnotesize{(Cited on page~#2)}%
    \else      \footnotesize{(Cited on pages~#2)}%
    \fi}
\hypersetup{
    colorlinks,
    linkcolor={darkblue},
    filecolor={ao(english)},
    citecolor={darkblue},
    urlcolor={cobalt}
}

\newcommand{\eqname}[1]

\usepackage{booktabs} 
\usepackage[shortlabels]{enumitem}
\usepackage{cite}
\usepackage{geometry}
\usepackage{appendix}
\usepackage{relsize}%
\usepackage{subcaption}

\usepackage{amsmath,amssymb,amsfonts,amsthm}
\usepackage{bbm}

\usepackage{algpseudocode}
\usepackage{algorithm}
\usepackage{custom_macros}
\usepackage{graphicx}
\usepackage{mdframed}

\usepackage{multicol}
\usepackage{xcolor}
\usepackage{tikz}
 \usepackage{tikz-3dplot}
 \usetikzlibrary{shapes.multipart}
 \usetikzlibrary{shapes.symbols}
\usepackage{pgfplots}

\usepackage{standalone}

\usepackage{subcaption}
\usepackage{comment}

\usepackage{etoolbox}
\usepackage{float}
\usepackage[belowskip=-15pt,aboveskip=0pt]{caption}

\usepackage{fancyhdr}
\newlength\FHoffset
\fancyheadoffset{\FHoffset}

\newlength\FHleft
\newlength\FHright

\newbox\FHline
\setbox\FHline=\hbox{\hsize=\paperwidth%
  \hspace*{\FHleft}%
  \rule{\dimexpr\headwidth-\FHleft-\FHright\relax}{\headrulewidth}\hspace*{\FHright}%
}

\allowdisplaybreaks

\title{\textbf{\Large Minimax Optimal Quantization of Linear Models: Information-Theoretic Limits and Algorithms}}

\author{Rajarshi Saha, Mert Pilanci, and Andrea J. Goldsmith
\thanks{Rajarshi Saha and Mert Pilanci are with the Department of Electrical Engineering, Stanford University, Stanford, CA 94305, USA (Emails: rajsaha@stanford.edu, pilanci@stanford.edu).
Andrea J. Goldsmith is with the Department of Electrical and Computer Engineering, Princeton University, Princeton, NJ 08544, USA. (Email: goldsmith@princeton.edu).
}
\thanks{
This work was partially supported by the Office of Naval Research under grant ONR N00014-18-1-2191, the AFRL/AFOSR University Center of Excellence (COE), National Science
Foundation under grants DMS-2134248 and ECCS-2037304, Army Research Office Early Career Award W911NF-21-1-0242, Intel Research, Huawei Technologies, Facebook Research, Adobe Research, and Stanford SystemX Alliance.}}
\date{\today}

\begin{document}
\maketitle

\begin{abstract}
\vspace{4mm}

High-dimensional models often have a large memory footprint and must be quantized after training before being deployed on resource-constrained edge devices for inference tasks.
In this work, we develop an information-theoretic framework for the problem of quantizing a linear regressor learned from training data $(\Xv, \yv)$, for some underlying statistical relationship $\yv = \Xv\thetav + \vv$.
The learned model, which is an estimate of the latent parameter $\thetav \in \Real^d$, is constrained to be representable using only $Bd$ bits, where $B \in (0, \infty)$ is a pre-specified budget and $d$ is the dimension.
We derive an information-theoretic lower bound for the minimax risk under this setting and propose a matching upper bound using randomized embedding-based algorithms which is tight up to constant factors. 
The lower and upper bounds together characterize the minimum threshold bit-budget required to achieve a performance risk comparable to the unquantized setting.
We also propose randomized Hadamard embeddings that are computationally efficient and are optimal up to a mild logarithmic factor of the lower bound. 
Our model quantization strategy can be generalized and we show its efficacy by extending the method and upper-bounds to two-layer ReLU neural networks for non-linear regression.
Numerical simulations show the improved performance of our proposed scheme as well as its closeness to the lower bound.

\end{abstract}

\clearpage

\tableofcontents

\clearpage

\section{Introduction}
\label{sec:introduction}

The sizes of the trained models in high dimensional learning problems have witnessed a tremendous increase, easily consisting of at least $10^6$ parameters.
Deploying them on severely resource-constrained edge devices for inference tasks is becoming a significant challenge.
To mitigate this, quantization approaches aim to represent model parameters using low-precision (LP) formats, instead of the usual $32$ or $64$-bit precision floating point.
Apart from reducing the memory footprint, LP representations also facilitate low latency for real-time inference and low energy consumption \cite{gholami_2021}.
Even if memory limits on the edge device is not a constraint, in a different situation, models might be trained on one device, and the trained models need to be transmitted to a remote device.
For example, the Kepler space telescope collects flux data from stars and does onboard estimation of their light curves \cite{Jenkins_2010}.
The estimated parameters need to be transmitted to earth for further studies.
Here, the parameters need to be quantized owing to the bandwidth constraints imposed by the communication link.

In this work, we adopt an information-theoretic perspective for the problem of \textit{learning a quantized linear model} from \textit{training data}, $(\Xv, \yv) \in \Real^{n \times d} \times \Real^n$.
The rows of the \textit{feature matrix} $\Xv$, i.e. $\{\xv_i\}_{i=1}^{n} \in \Real^d$, constitute the \textit{input features}, and the entries of $\yv \in \Real^n$, i.e. $\{y_i\}_{i=1}^{n} \in \Real$ are the corresponding \textit{observations}.
In its most general form, our problem setup is presented in Fig. \ref{fig:learning_and_quantizing}.
We assume the existence of an \textit{oracle} that can provide us with $y_i \in \Real$ for $\xv_i \in \Real^d$, and the \textit{learner} observes a noisy version of this response.
This framework captures several problems like system identification \cite{system_identification}, data-driven control \cite{brunton_kutz_2019, HOU20133}, etc., where it becomes crucial to quantize the learned parameters to store them on (possibly remote) hardware for subsequent processing.
$\Xv$ may or may not be designed by the \textit{learner}.
We study methods to approximately solve the following \textbf{quantized least squares} (\textbf{QLS}) optimization problem subject to a finite cardinality constraint,

\begin{equation}
\label{eq:bit_constrained_least_squares}
\tag{\sc QLS}
    \thetatv := \argminimize_{\zetav \in \Svcal}\,\lVert \yv - \Xv\zetav \rVert_2^2,
\end{equation}
\begin{wrapfigure}{r}{0.55\linewidth}
    \centering
    \resizebox{1\linewidth}{!}{
        \begin{tikzpicture}[
    		node distance= 6em and 4em,
    		sloped]
            \tikzset{
    			box/.style = {
    				fill=blue!15,
    				shape=rectangle, 
    				rounded corners,
    				draw=blue!40, 
    				align=center,
    				text = black,
    				font=\fontsize{5}{5}\selectfont},
    			dummybox/.style = {
    				shape=circle,  
    				align=center, 
    				minimum size={width("rrrrrrrrrrr")+2pt}},
    			arrow/.style={
    				color=black,
    				draw=blue,
    				-latex,
    				font=\fontsize{4}{4}\selectfont},
    			state/.style={
    				rectangle split,
    				rectangle split parts=1,
    				rectangle split part fill={cyan!10},
    				rounded corners,
    				draw=blue, thick,
    				minimum height=4cm,
    				text width=2.4cm,
    				inner sep=4pt,
    				text centered,
    			}
	        }
	        
	        \node  (Oracle) [state] at (2.0,2.0) {\textbf{Oracle}};
	        \node  (Learner) [state] at (2.0,0.0) {\textbf{Learner}};
	        \draw[->, very thick] (0.67,0) node[] {} to[out=180,in=180] (0.67,2.0) node[]{};
	        \node (plus) [draw, circle] at (4.5,1.0) {$\mathbf{+}$};
	        \draw[->, very thick] (3.35,2.0) node[] {} to[out=0,in=90] (4.5,1.4) node[]{};
	        \draw[->, very thick] (4.5,0.6) node[] {} to[out=270,in=0] (3.35,0) node[]{};
	        \node (Noise) at (6.7,1.0) {$\textbf{Noise } \mathbf{\in \Real}$};
	        \draw[->, very thick] (5.7,1.0) node[] {} to[out=180,in=0] (4.9,1.0) node[]{};
	        \node (action) at (-1.45,1.0) {$\textbf{Action / Probe /}$};
	        \node (feature) at (-1.35,0.6) {$\textbf{ Features } \mathbf{\in \Real^d}$};
	        \node (Response) at (5.8,1.9) {$\textbf{Response } \mathbf{\in \Real}$};
	        \node (Observation) at (5.9,0) {$\textbf{Observation } \mathbf{\in \Real}$};
	        \node  (Edge_Device) [state] at (10.0,-0.85) {\textbf{Edge Device}};
	        \draw[-, very thick] (2.0,-0.29) node[] {} to[out=270,in=270] (2.0,-0.8) node[]{};
	        \draw[-, very thick] (2.0,-0.85) node[] {} to[out=0,in=180] (2.5,-0.85) node[]{};
	        \node (Deploy) at (5.3,-0.85) {$\textbf{Deploy } (\text{constrained to } dB \text{ bits})$};
	        \draw[->, very thick] (8.05,-0.85) node[] {} to[out=0,in=180] (8.66,-0.85) node[]{};
	        \node (Encoder) at (1.1, -0.65) {$(\text{Encoder})$};
	        \node (Decoder) at (10.0, -0.25) {$(\text{Decoder})$};
        \end{tikzpicture}
    }
    \vspace{0.5mm}
    \caption{Learning a Quantized Model}
    \vspace{2mm}
    \label{fig:learning_and_quantizing}
\end{wrapfigure}
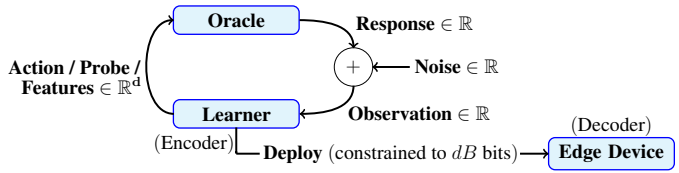

where $\Svcal \subseteq \Real^d$ is a finite set that satisfies $\log_2\lvert \Svcal \rvert \leq dB$ for some pre-specified $B \in (0, \infty)$.
The constraint $\zetav \in \Svcal$ denotes that we search for quantized models that can be represented using only $B$-bits per dimension.
Solving \eqref{eq:bit_constrained_least_squares} requires us to have a finite cardinality set $\Svcal$ (referred to as the \textbf{codebook}, or the \textbf{quantization lattice}), and subsequently find the best lattice point $\thetatv \in \Svcal$ that serves as our quantized model.
Although \eqref{eq:bit_constrained_least_squares} can be solved without any prior knowledge of the statistical relationship between $\Xv$ and $\yv$, for the purposes of deriving the tight lower and upper bounds, we assume that $\yv \in \Real^n$ is a linear observation obtained according to the \textit{noisy planted model} (\textbf{NPM}) as follows,
\begin{equation}
\label{eq:noisy_planted_model}
\tag{\sc NPM}
    \yv = \Xv\thetav + \vv; \hspace{2mm} \thetav \in \Thetav \subseteq \Real^d, \hspace{5mm} \vv \sim \Ncal(\mathbf{0}, \sigma^2\Iv_n),
\end{equation}

for some \textit{true model} (or \textit{latent parameter}) $\thetav \in \Real^d$ that we want to learn.
\eqref{eq:noisy_planted_model} is a reasonable model when $\yv$ is corrupted by several noise sources, whose superposition is approximated by a Gaussian noise due to the Central Limit Theorem.
To begin with, we derive an information-theoretic lower bound on the minimax $\ell_2$-risk of any estimate $\thetatv$ of $\thetav$ in \S \ref{sec:lower_bound_to_the_minimax_risk}, followed by several model quantization algorithms and their comparison in terms of $\ell_2$-error and computational complexity in \S \ref{sec:algorithms_for_model_quantization}.
Numerical evaluations are presented in \S \ref{sec:numerical_simulations}.
We also extend our model quantization scheme for general noise settings in App. \ref{app:general_noise_models}, as well as for obtaining an upper bound on the output error of a $2$-layer densely connected ReLU neural network for non-linear regression in App. \ref{app:NN_analysis}.

\subsection{Significance and Related Work}
\label{sec:related_work}

Model compression techniques employ lossless and lossy source coding schemes to quantize the model parameters subject to bit-budget constraints in order to make them deployable on memory-constrained devices \cite{transform_quant_young_2021, han2016deep}.
We develop an information-theoretic framework for the problem of quantizing linear regressors.
Related to our work, \cite{gao_Rate_distortion_model_compression} adopts a rate-distortion theoretic approach to answer the question: \textit{What is the minimum number of bits required to compress a model to achieve a target distortion?}
Their result for linear models \cite{gao_Rate_distortion_model_compression}[\S 4.1] assumes the model parameters follow a Gaussian distribution and the different input features are uncorrelated.
Similarly, \cite{isik_2022_rate_distortion} provides empirical evidence to approximate the distribution of model parameters as an i.i.d. Laplacian, and consequently proposes a weight pruning scheme to compress a model without deteriorating performance.
However, modeling parameters as an i.i.d. sequence also ignores correlations between them, leading to sub-optimal performance in worst-case scenarios.

In contrast to these works, we answer the question: \textit{What is the minimum achievable risk when a model is quantized according to a \textbf{pre-specified} bit budget of $B$ bits per dimension?}
Moreover, we adopt a minimax approach, and we do not assume any prior distributional knowledge about the latent parameter $\thetav$ or the feature matrix $\Xv$.
Consequently, our proposed model quantization strategies are minimax optimal, i.e. optimal in a worst-case sense.
Our work can be seen as a generalization of Gaussian sequence models \cite{zhu_neurips_2014} to general least-squares regression models.

We also propose randomized subspace embedding-based model quantization algorithms that have a near optimal performance (up to constant factors) with respect to our minimax lower bound.
We note that our lower bound expression is comprised of two terms that can be interpreted as the sum of the \textit{unavoidable learning risk} that is persistent even in the absence of any bit-budget constraint, and an \textit{excess quantization risk} that depends on the bit-budget.
This suggests that we can have tight upper bounds that also have this separation property, by first learning the model parameters agnostic to any bit-budget constraints, and then subsequently quantizing them. 
We use the idea of democratic (or Kashin) representations \cite{lyubarskii, studer2012, safaryan2021, DSC_early_access, chen_2020_breaking_trilemma} to show that this hypothesis is indeed true.

\section{Model Quantization: An Information-Theoretic Perspective}
\label{sec:model_quantization_framework}

As introduced in \S \ref{sec:introduction}, we consider the model \ref{eq:noisy_planted_model} specified by $\yv = \Xv\thetav + \vv$, where $\thetav \in \Real^d$ is the underlying (\textit{latent}) model parameter to be estimated.
We do not make any distributional assumptions on the entries of $\Xv$ or $\thetav$.
$\Xv$ can be completely arbitrary and each $\theta_i \in \Real$ of the fixed parameter $\thetav \equiv [\theta_1, \ldots, \theta_d]^\top$ can take arbitrary values.
Given \textit{training data} $(\Xv, \yv)$, our goal is to learn and quantize an estimate of $\thetav$ that can be represented using $B$-bits per dimension.
We consider the model parameter space to be a Euclidean ball, i.e. $\thetav \in \Thetav = \{\thetav : (1/d)\lVert \thetav \rVert_2^2 \leq c^2\}$ for some known constant $c \in (0,\infty)$.
Let $\thetatv \equiv \thetatv(\Xv, \yv) \in \Thetav$ be a quantized estimator of $\thetav$.
For any $\thetatv$, we define the expected $\ell_2$-risk be,
\begin{equation}
    \label{eq:l2_risk_definition}
    R(\thetatv, \thetav) = \mathbb{E}_{\yv}\left[\frac{1}{d}\lVert \thetatv - \thetav \rVert_2^2\right],
\end{equation}
where the expectation $\mathbb{E}[\cdot]$ is over the randomness in $\yv$ due to noise $\vv$.
Since $B \in (0, \infty)$, we consider both, when $B < 1$ (sub-linear regime) and $B > 1$, when at least $1$ bit is available per dimension.

\subsection{An Encoder-Decoder Framework}
\label{subsec:encoder_decoder_framework}

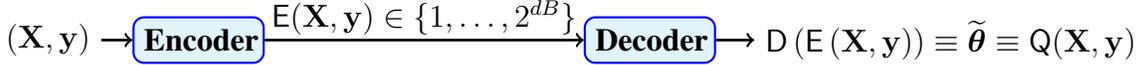
\begin{figure*}[t!]
	\centering
	{
		\begin{tikzpicture}[
		node distance= 6em and 4em,
		sloped]
		\tikzset{
			box/.style = {
				fill=blue!15,
				shape=rectangle, 
				rounded corners,
				draw=blue!40, 
				align=center,
				text = black,
				font=\fontsize{10}{10}\selectfont},
			dummybox/.style = {
				shape=circle,  
				align=center, 
				minimum size={width("rrrrrrrrrrr")+2pt}},
			arrow/.style={
				color=black,
				draw=blue,
				-latex,
				font=\fontsize{8}{8}\selectfont},
			state/.style={
				rectangle split,
				rectangle split parts=1,
				rectangle split part fill={cyan!10},
				rounded corners,
				draw=blue, thick,
				minimum height=4cm,
				text width=1.45cm,
				inner sep=4pt,
				text centered,
			}
		}
	
    	\node  (Encoder) [state] at (2.0,0) {\textbf{Encoder}};
    	\node (input) at (0,0) {{$(\Xv, \yv)$}};
    	\draw[->,line width=1pt] (input) -- (Encoder);
    	\node (index) at (5.0,0.3) {{$\mathsf{E}(\Xv,\yv) \in \{1, \ldots, 2^{dB}\}$}};
    	\node  (Decoder) [state] at (8.0,0) {\textbf{Decoder}};
    	\draw[->,line width=1pt] (Encoder) -- (Decoder);
    	\node (Learned_model) at (12.0,0) {{$\mathsf{D}\left(\mathsf{E}\left(\Xv,\yv\right)\right) \equiv \thetatv \equiv \Qsf(\Xv, \yv)$}};
    	\draw[->,line width=1pt] (Decoder) -- (Learned_model);
		\end{tikzpicture}}
		\vspace{2mm}
\caption{An $(n,d,B)$-\textbf{learning code} represented as a composition of \textbf{Encoder} and \textbf{Decoder}}
\label{fig:encoder_decoder}
\end{figure*}

To view model quantization from an information-theoretic lens, we propose an \textbf{encoder-decoder} framework (Fig. \ref{fig:encoder_decoder}) to describe the class of algorithms that jointly learn and quantize the latent parameter $\thetav$ from training data $(\Xv, \yv)$.
Its efficacy becomes apparent when we use it in \S \ref{sec:lower_bound_to_the_minimax_risk} to derive the minimax lower bound over the class of algorithms described by it.
This general framework is applicable irrespective of whether $(\Xv, \yv)$ is generated according to \eqref{eq:noisy_planted_model}, or any other model.

In the absence of any quantization, any \textbf{\textit{learning algorithm}} is a mapping $\Real^{n \times d} \times \Real^n \to \Thetav$.
With quantization, we study algorithms that \textit{jointly learn and quantize} $\thetav$, and we refer to such algorithms as \textbf{\textit{learning codes}} (formally defined in Def. \ref{def:ndB_learning_code} below).
We use the term ``code" because, any algorithm that learns a quantized estimate of $\thetav$ can be perceived as a composition of an \textbf{encoder} and a \textbf{decoder} function.
In Fig. \ref{fig:encoder_decoder}, the \textit{encoder} $\Esf:\Real^{n \times d} \times \Real^n \to [2^{dB}] \triangleq \{1, \ldots, 2^{dB}\}$ maps the data $(\Xv, \yv)$ to an index $i \in [2^{dB}]$.
The \textit{decoder} $\Dsf:[2^{dB}] \to \Svcal \subseteq \Thetav$ maps the index $i \in [2^{dB}]$ to a quantized parameter in $\Svcal \subseteq \Thetav$, $|\Svcal| \leq 2^{dB}$.
The encoder is implemented at the learner, and the decoding happens on the edge device (cf. Fig. \ref{fig:learning_and_quantizing}).
The codebook $\Svcal$ forms a \textit{quantization lattice} in $\Real^d$.
It is designed beforehand, and is known (or can be constructed) at both the encoder and the decoder.
The goal of a \textit{learning code} is to find the lattice point in $\Svcal$ that best estimates $\thetav$.
\begin{definition}
\label{def:ndB_learning_code}
An $(n,d,B)$-\textbf{learning code} $\Qsf:\Real^{n \times d}\times \Real^n  \to \Thetav$ is defined to be the composition of \textit{encoder} and \textit{decoder} mappings $\Esf(\cdot, \cdot):\Real^{n \times d} \times \Real^n \to[2^{dB}]$ and $\Dsf:[2^{dB}] \to \Thetav$, such that for any given training data $(\Xv, \yv) \in \Real^{n \times d} \times \Real^n$, we have $\Qsf(\Xv,\yv) \equiv \Dsf(\Esf(\Xv, \yv)) \in \Thetav$.
\end{definition}
%

\subsection{A Minimax Risk Formulation for Comparing Learning Codes}
\label{subsec:minimax_framework}

Given training data $(\Xv, \yv)$, let $\thetatv \equiv \Qsf(\Xv, \yv)$ be the output of an $(n,d,B)$-learning code $\Qsf$.
The performance of any such learning code can be measured by its worst-case risk for all $\thetav$ over the parameter space $\Thetav$ as $\Rcal_{\Qsf} \triangleq \sup_{\thetav \in \Thetav}R(\Qsf(\Xv, \yv), \thetav)$, where $R(\cdot, \cdot)$ is the $\ell_2$-risk as defined in eq. \eqref{eq:l2_risk_definition}.
Let $\Qcal_{n,d,B}$ be the class of all $(n,d,B)$-learning codes.
In a minimax formulation, our goal is to find a learning code that achieves the best possible worst-case risk and, hence, an information-theoretic limit corresponds to a lower bound on
\begin{equation}
\label{eq:asymptotic_minimax_risk}
    \Rcal_{B, c} \triangleq \inf_{\Qsf \in \Qcal_{n,d,B}} \Rcal_{\Qsf} \equiv \inf_{\Qsf \in \Qcal_{n,d,B}}\sup_{\thetav : \lVert \thetav \rVert_2 \leq c\sqrt{d}} R(\Qsf(\Xv, \yv), \thetav).
\end{equation}
In \S \ref{sec:lower_bound_to_the_minimax_risk}, we derive a lower bound on \eqref{eq:asymptotic_minimax_risk} giving us a performance limit of the best possible learning code (in a global minimax sense) that we can hope to come up with under this setting.
Subsequently, in \S \ref{sec:algorithms_for_model_quantization}, we propose several model quantization schemes that achieve a performance matching the derived lower bound, establishing their optimality as well as the tightness of the lower bound.

\section{Lower Bound for the Minimax Risk}
\label{sec:lower_bound_to_the_minimax_risk}

We now present a lower bound for the minimax risk \eqref{eq:asymptotic_minimax_risk} in estimating $\thetav$ from the linear model \eqref{eq:noisy_planted_model} under quantization constraints.
The proof sketch in this section is outlined via Lemmas \ref{lem:mutual_information_bit_budget_inequality}, \ref{lem:inequality_supremum_bayes_error}, and \ref{lem:lower_bound_MI_Gaussian_prior}, building up to our lower bound in Thm. \ref{thm:minimax_error_lower_bound}.
Inspired from \cite{zhu_neurips_2014}, our proof technique relies on the following two key intuitive observations:
\begin{enumerate}[leftmargin=*]
    \item Given $\Xv$, for any learning code $\Qsf$, the mutual information between $\yv$ and $\Qsf(\Xv, \yv)$ is a reasonable proxy for the number of bits $B$ used to design the corresponding encoder-decoder pair $(\Esf, \Dsf)$; in other words, the number of bits required for representing the estimate of $\thetav$.
    The marginal distribution of $\yv$ is determined by $\Xv$, $\thetav$, and $\vv$ according to \eqref{eq:noisy_planted_model}, and, for a given $\Xv$, the distribution of the output of the learning code $\thetatv \equiv \Qsf(\Xv, \yv)$ is determined solely by the marginal of $\yv$ and the conditional distribution of $\thetatv$ given $\yv$.
    This idea follows from the \textit{information-rate distortion function} in rate distortion theory literature \cite{info_theory_book}[Ch. 10], and is formalized in Lemma \ref{lem:mutual_information_bit_budget_inequality}.
    \item The worst-case risk for any learning code is the Bayes risk for the least favorable choice of prior distribution over $\thetav$, i.e. given $\Qsf$, $\Rcal_{\Qsf}= \sup_{\thetav \in \Thetav}R(\Qsf(\cdot, \cdot), \thetav)$ is equal to the Bayes risk when the prior is chosen adversarially.
    In other words, $\Rcal_{\Qsf} $ is lower bounded by the Bayes risk \textit{for any prior}.
\end{enumerate}
We combine the two observations above as follows.
Since the worst-case risk is lower bounded by the Bayes risk, for any distribution $\thetav \sim P(\thetav)$ and any $(n,d,B)$-learning code $\Qsf$, we have,
\begin{equation}
\label{eq:worst_case_and_Bayes_error_inequality}
    \Rcal_{\Qsf} \equiv \sup_{\thetav \in \Thetav}R(\Qsf(\Xv,\yv), \thetav) \geq \int_{\thetav \in \Thetav}R(\Qsf(\Xv, \yv), \thetav)dP(\thetav)
\end{equation}
We then choose the prior $P(\thetav)$ in \eqref{eq:worst_case_and_Bayes_error_inequality} appropriately, so that the Bayes risk can be evaluated (or at least lower bounded) by an expression that demonstrates its dependence on various problem parameters.
We consider a lower bound on the asymptotic setting of $d \to \infty$.
However, we discuss methods for extending our lower bound to non-asymptotic settings\footnote{When $d$ is finite, the prior is chosen to be $\Ncal(\mathbf{0}, \sigma^2\delta^2\Iv_d)$ for some $\delta \in (0,1)$.
The optimal choice of $\delta$ in Lemmas \ref{lem:inequality_supremum_bayes_error} and \ref{lem:lower_bound_MI_Gaussian_prior} needs to be chosen numerically, which does not yield neat closed form expressions.} in App. \ref{app:finite_dimension_lower_bound}.
For our setting, the appropriate choice turns out to be the Gaussian prior $\thetav \sim \Ncal(\mathbf{0}, c^2\Iv_d)$, stated in Lemma \ref{lem:lower_bound_MI_Gaussian_prior}.
Lemma \ref{lem:inequality_supremum_bayes_error} formally states a stronger version of \eqref{eq:worst_case_and_Bayes_error_inequality}, in which the inequality between the worst-case risk and the Bayes risk holds true even when the latter is evaluated over the whole of $\Real^d$ instead of just $\Thetav$, due to concentration properties of the Gaussian prior.
We formalize the above sketch in Lemmas \ref{lem:mutual_information_bit_budget_inequality}, \ref{lem:inequality_supremum_bayes_error}, and \ref{lem:lower_bound_MI_Gaussian_prior} below.
Their detailed proofs are delegated to App. \ref{app:proof_mutual_information_bit_budget_inequality}, \ref{app:proof_inequality_supremum_bayes_error}, and \ref{app:proof_lower_bound_MI_Gaussian_prior} in that order.
\begin{lemma}
\label{lem:mutual_information_bit_budget_inequality}
Let the true model be drawn from some distribution $\thetav \sim P(\thetav)$, and given $\Xv$, let the output $\yv \in \Real^n$ be generated from $\thetav$ according to some conditional distribution $p(\yv \vert \thetav)$ specified by $\yv = \Xv \thetav + \vv$, $\vv \sim \Ncal(\mathbf{0}, \sigma^2\Iv_n)$.
Then for any $\Qsf \in \Qcal_{n,d,B}$, we have $I(\yv; \Qsf(\Xv, \yv)) \leq dB$.
\end{lemma}

\begin{lemma}
\label{lem:inequality_supremum_bayes_error}
For $\Thetav = \{\thetav : (1/d)\lVert \thetav \rVert_2^2 \leq c^2\}$ and the Gaussian prior $F \equiv \Ncal(\mathbf{0}, c^2\delta^2\Iv_d)$ for some $\delta \in (0,1)$, in the asymptotic setting as $d \to \infty$, any $\Qsf \in \Qcal_{n,d,B}$ satisfies,
\begin{equation}
    \label{eq:inequality_supremum_bayes_error}
    \sup_{\thetav \in \Thetav}R(\Qsf(\Xv,\yv), \thetav) \geq \int_{\thetav \in \Real^d}R(\Qsf(\Xv, \yv), \thetav)dF(\thetav)
\end{equation}
\end{lemma}

\begin{lemma}
\label{lem:lower_bound_MI_Gaussian_prior}
Suppose $\thetav \sim \Ncal(\mathbf{0}, c^2\delta^2\Iv_d)$ for some $\delta \in (0,1)$, $\yv = \Xv\thetav + \vv$, and $\vv \sim \Ncal(\mathbf{0}, \sigma^2\Iv_n)$.
Then given $\Xv$, for any $\zetav$ satisfying $p(\zetav \vert \yv, \thetav) = p(\zetav \vert \yv)$, we have,
\begin{equation*}
    \label{eq:lower_bound_MI_Gaussian_prior_expression}
     I(\yv;\zetav) \geq \frac{d}{2}\log\left(\frac{c^4\delta^4\sigma_{\mathrm{min}}^2}{\sigma^2 + c^2\delta^2\sigma_{\mathrm{min}}^2}\Paren{\int  R(\zetav, \thetav)dF(\thetav)
    - \frac{c^2\delta^2\sigma^2}{\sigma^2 + c^2\delta^2\sigma_{\mathrm{max}}^2}}^{-1}\right),
\end{equation*}
where $p(\cdot)$ denotes any distribution, $F(\thetav)$ is the Gaussian prior over $\thetav$, and $\sigma_{\mathrm{max}}$ and $\sigma_{\mathrm{min}}$ are the maximum and minimum singular values of $\Xv$, respectively.
\end{lemma}
\vspace{-3.5mm}
For given $\Xv$, the requirement $p(\zetav \vert \yv, \thetav) = p(\zetav \vert \yv)$ ensures that the learned and quantized model $\thetatv$, that is an estimate of $\thetav$, depends only on the responses $\yv$ and no prior information about $\thetav$ is utilized.
\begin{theorem}
\label{thm:minimax_error_lower_bound}
For $B, \sigma$, $c$, the asymptotic minimax risk \eqref{eq:asymptotic_minimax_risk} for \eqref{eq:noisy_planted_model} can be lower bounded as\footnote{Since $\liminf_{d \to \infty}$, we can consider \textit{any} sequence of matrices of non-decreasing dimensions, $\{X_i\}_{i=1}^{\infty} \to \Xv$, where max. and min. singular values of $\Xv$ are upper and lower bounded by $\sigma_{\mathrm{max}}$ and $\sigma_{\mathrm{min}}$ respectively.},
\begin{equation}
    \Rcal_{B,c}^{(\infty)} \triangleq \liminf_{d \to \infty} \Rcal_{B, c} \geq \frac{c^2\sigma^2}{\sigma^2 + c^2\sigma_{\mathrm{max}}^2} + \frac{c^4\sigma_{\mathrm{min}}^2}{\sigma^2 + c^2\sigma_{\mathrm{min}}^2}\cdot 2^{-2B}.
\end{equation}
\end{theorem}

We prove Thm. \ref{thm:minimax_error_lower_bound} in App. \ref{app:proof_minimax_error_lower_bound}
Interestingly, the lower bound in Thm. \ref{thm:minimax_error_lower_bound} consists of two terms.
The first term is the unavoidable \textit{learning risk} (independent of $B$), and the second term is the excess \textit{quantization risk} that arises due to the bit-budget constraint.
This suggests the notion of a minimum threshold bit-budget ($B_{thr}$) that ensures that the quantization risk is of the same order of magnitude as the learning risk.
To get an expression for $B_{thr}$, we require,
\begin{align}
    \label{eq:threshold_bit_budget}
    \frac{c^2\sigma^2}{\sigma^2 + c^2\sigma_{\mathrm{max}}^2} \gtrsim \frac{c^4\sigma_{\mathrm{min}}^2}{\sigma^2 + c^2\sigma_{\mathrm{min}}^2}\cdot 2^{-2B}
    \implies B_{thr} = \Omega\Paren{\log\Paren{\omega^2\sigma_{\mathrm{min}}^2\frac{1 + \omega^2\sigma_{\mathrm{max}}^2}{1 + \omega^2\sigma_{\mathrm{min}}^2}}},
\end{align}
where $\omega^2 \triangleq c^2/\sigma^2$.
For well-conditioned matrices, $\sigma_{\mathrm{max}} = \Theta(\sigma_{\mathrm{min}})$ and $B_{thr} = \Omega\Paren{\log(\omega^2\sigma_{\mathrm{min}}^2)}$.
Note that $\omega^2$ can increase either by increasing $c^2$ or decreasing $\sigma^2$.
Increasing $c^2$ means $\Thetav$ gets larger implying the learning problem becomes more difficult.
Decreasing $\sigma^2$ means with lower noise variance, the unquantized learning risk is lower, and consequently, more bits are required for model quantization in order to ensure a lower excess quantization risk.
In \S \ref{sec:algorithms_for_model_quantization}, we propose several learning codes and compare the minimum budget required by each of them with the threshold budget in eq. \eqref{eq:threshold_bit_budget} to evaluate their performance.
Thm. \ref{thm:minimax_error_lower_bound} simplifies to the result of \cite{zhu_neurips_2014} when $\Xv = \Iv$, and that of \cite{nussbaum_1999} when $\Xv = \Iv$ with no budget constraint, i.e. $B \to \infty$, and hence can be seen as their generalization.

\vspace{-1mm}
\section{Algorithms for Learning and Quantizing Linear Models}
\label{sec:algorithms_for_model_quantization}
\vspace{-2mm}

We now propose several learning codes that achieve $\ell_2$-risks asymptotically matching the lower bound to the minimax risk, while requiring $B + o_d(1)$ bits per dimension, where $o_d(1) \to 0$ as $d \to \infty$.
These algorithms provide upper bounds on the minimax risk that together with Thm. \ref{thm:minimax_error_lower_bound}, prove the tightness of the lower bound (up to constant factors), completely characterizing the threshold budget (eq. \eqref{eq:threshold_bit_budget}) required to learn a model $\thetav$ from linear observations under bit-budget constraints.

We first discuss the $\ell_2$-risk of a straightforward strategy in \S \ref{subsec:naive_learning_and_quantization} that initially ignores any bit-budget constraints, estimates unconstrained $\thetav$ from training data $(\Xv, \yv)$, and subsequently does a uniform scalar quantization of the learned model.
Such a na\"ive strategy is highly sub-optimal as its $\ell_2$-risk grows linearly with dimension.
Consequently, we propose a randomized projection based learning code called \textbf{Random Correlation Maximization (RCM)} in \S \ref{subsec:random_correlation_maximization}, where the encoder and decoder agree on a random codebook whose elements are distributed uniformly at random on the unit sphere.
Although optimal, this scheme is not practically feasible and simply exhibits the \textit{existence of a good learning code}, much like Shannon's achievability proof for channel coding theorem \cite{info_theory_book}[Ch. 7].
We also propose: \textbf{Democratic Quantization (DQ)} and \textbf{Near-Democratic Quantization (NDQ)} in \S \ref{subsec:democratic_quantization} and \S \ref{subsec:near_democratic_quantization}.
Compared to \textbf{RCM}, \textbf{DQ} and \textbf{NDQ} introduce a slightly more structure into the codebooks generated, making them practically feasible.
\textbf{DQ} achieves a dimension-independent worst-case $\ell_2$-risk with a computational complexity of $O(d^2)$ multiplications.
On the other hand, \textbf{NDQ} has a complexity of $O(d \log d)$ additions, significantly less than \textbf{DQ}, while paying only a very modest price for its $\ell_2$-risk in terms of a weak logarithmic dependence on the dimension.
We summarize the performance and computational complexity of different learning codes in Table \ref{tab:learning_codes_comparison}.

\begin{remark}
In the absence of any bit-budget constraints, the least squares problem \eqref{eq:bit_constrained_least_squares} can be solved in closed form.
We analyze our proposed learning codes by assuming that the pseudoinverse $\Xinv = (\Xv^{\top}\Xv)^{-1}\Xv^{\top} \in \Real^{d \times n}$, is available and do not take into account the complexity of computing it.
However, in practice, unconstrained least-squares problems can be solved via one of several optimization algorithms such as gradient descent, conjugate gradient, etc, and all of them eventually converge to the estimate given by the closed form solution.
\end{remark}

\begin{table*}[t!]
\begin{center}
\begin{small}
\begin{sc}
\renewcommand{\arraystretch}{1.1}
\begin{tabular}{lcccr}
\toprule
Code & Performance guarantee (w.h.p.) & Threshold budget ($B_*$) & Complexity \\

\midrule
\textbf{Na\"ive} (\S \ref{subsec:naive_learning_and_quantization}) & $R(\thetatv, \thetav) \leq R_i + \frac{2\textcolor{red}{d}c^4\sigma_{\mathrm{max}}^2}{\sigma^2 + c^2\sigma_{\mathrm{max}}^2}2^{-2B}$ & $\Omega\Paren{\log\Paren{\textcolor{red}{d}\omega^2\sigma_{\mathrm{max}}^2}}$ & $O(d)$ \\

\textbf{RCM} (\S \ref{subsec:random_correlation_maximization}) & $R(\thetatv, \thetav) \leq R_i+ \frac{c^4\sigma_{\mathrm{max}}^2}{\sigma^2 + c^2\sigma_{\mathrm{max}}^2}2^{-2B}$ & $\Omega\Paren{\log\Paren{\omega^2\sigma_{\mathrm{max}}^2}}$ & $O(2^{dB})$ \\

\textbf{DQ} (\S \ref{subsec:democratic_quantization}) & $R(\thetatv, \thetav) \leq R_i + \frac{16K_u^2c^4\sigma_{\mathrm{max}}^2}{\sigma^2 + c^2\sigma_{\mathrm{max}}^2}2^{-\frac{2B}{\lambda}}$ & $\Omega\Paren{\log\Paren{K_u^2\omega^2\sigma_{\mathrm{max}}^2}}$ & $O(d^2)$\\

\textbf{NDQ} (\S \ref{subsec:near_democratic_quantization}) & $R(\thetatv,\thetav) \leq R_i + \frac{64c^4\sigma_{\mathrm{max}}^2\textcolor{red}{\log(2\lambda d)}}{\sigma^2 + c^2\sigma_{\mathrm{max}}^2}2^{-\frac{2B}{\lambda}}$ & $\Omega\Paren{\log\Paren{\omega^2\sigma_{\mathrm{max}}^2\textcolor{red}{\log(2\lambda d)}}}$ & $O(d \log d)$ \\

\midrule
\textbf{LB} (\S \ref{sec:lower_bound_to_the_minimax_risk}) & $R(\thetatv,\thetav) \geq R_i + \frac{c^4\sigma_{\mathrm{min}}^2}{\sigma^2 + c^2\sigma_{\mathrm{min}}^2}2^{-2B}$ & $B_{thr} = \Omega\Paren{\log\Paren{\omega^2\sigma_{\mathrm{min}}^2}}$ & $--$ \\
\bottomrule
\end{tabular}
\end{sc}
\end{small}
\end{center}
\vspace{-2mm}
\caption{\textsc{Comparison of various learning codes and Lower Bound (LB) of Thm. \ref{thm:minimax_error_lower_bound}}\\ ($\textcolor{blue}{d}$: dimension, $\textcolor{blue}{c} : \norm{\thetav}_2 \leq c\sqrt{d}$, $\textcolor{blue}{\sigma^2}$: noise, $\textcolor{blue}{\omega^2} \triangleq c^2/\sigma^2$, \textcolor{blue}{B}: bit-budget, $\textcolor{blue}{\sigma_{\mathrm{max}}}, \textcolor{blue}{\sigma_{\mathrm{min}}}$: max. \& min. singular values of $\Xv$, $\textcolor{blue}{K_u}, \textcolor{blue}{\lambda}$: constants, $\textcolor{blue}{R_i} = O(c^2\sigma^2/(\sigma^2 + c^2\sigma_{\mathrm{min}}^2)$: inherent risk)}
\vspace{2mm}
\label{tab:learning_codes_comparison}
\end{table*}

\vspace{-2mm}
\subsection{Naive Learning and Quantization}
\label{subsec:naive_learning_and_quantization}

We first look at a na\"ive strategy for learning and quantizing $\thetav$ from $(\Xv, \yv)$.
We do this in two steps: learning and encoding
\begin{enumerate*}[(i)]
    \item the magnitude, $\lVert \thetav \rVert_2$, and \item the direction, $\thetav/\lVert \thetav \rVert_2$.
\end{enumerate*}

\textbf{Learning and Encoding the Magnitude}.
For this, we make use of the fact that $(1/d)\lVert \thetav \rVert_2^2 \leq c^2$.
Let $\sigma_1 \geq \sigma_2 \geq \ldots \geq \sigma_d$ denote the singular values of $\Xv$.
Then an estimate of $\lVert \thetav \rVert_2^2$ is given by, $\bhat^2 \triangleq \frac{1}{d}\lVert\Xinv\yv\rVert_2^2 - \frac{\sigma^2}{d}\sum_{i=1}^{d}\sigma_i^{-2}$.
This choice of $\bhat^2$ becomes apparent in the proof of Lemma \ref{lem:error_encoded_magnitude}.
Now, we generate a codebook $\Bcal = \{1/\sqrt{d}, 2/\sqrt{d}, \ldots, \lceil c^2\sqrt{d}\rceil/\sqrt{d}\}$ for encoding the magnitude, and encode $\bhat^2$ as, $\varphi_{\Bcal}(\bhat^2) = \argminimize_{b}\{\lvert b - \bhat^2 \rvert: b \in \Bcal\}$.
$\varphi_{\Bcal}(\bhat^2)$ returns an index from $\{1, \ldots, \lceil c^2\sqrt{d} \rceil\}$.
To decode, let $\bt^2 = \psi_{\Bcal}(\varphi_{\Bcal}(\bhat^2))$, where $\psi_{\Bcal}(i)$ returns the $i^{th}$ element in the codebook $\Bcal$.
The strategy for learning and encoding $\lVert \thetav \rVert_2$ is the same for other learning codes proposed subsequently in \S \ref{subsec:random_correlation_maximization}, \ref{subsec:democratic_quantization}, and \ref{subsec:near_democratic_quantization}.
For this reason, we state Lemma \ref{lem:error_encoded_magnitude} in App. \ref{sec:error_encoded_magnitude} that gives a tail probability bound on the deviation of the estimate $\bhat^2$ from $\lVert \thetav \rVert_2^2$.
We show that as $d \to \infty$, $\bhat^2 \to \lVert \thetav \rVert_2^2$ with high probability (w.h.p.) and we use it in our analysis.

\textbf{Uniform scalar quantization to encode the direction}.
To do this, we uniformly quantize the \textit{shape} $\sv \triangleq \Xinv\yv/\lVert \Xinv\yv \rVert_2$ by allocating $B$-bits to each coordinate.
Let $\textbf{B}_{\infty}^d(1)$ denote the $\ell_{\infty}$-ball of radius $1$ in $\Real^d$.
Since $\lVert \sv \rVert_{\infty} \leq 1$, a \textbf{uniform quantizer}, denoted as $\text{Q}_{u,B}(\cdot)$, constructs the quantization lattice by choosing $M = 2^B$ points $\{u_i\}_{i=1}^{M}$ along every dimension, given by $u_i = -1 + (2i-1)\Delta/2$ for $i = 1, \ldots, M$, where $\Delta = 2/M$ is the \textbf{resolution}.
For $\zv \equiv [z_1 \ldots z_d] \in \textbf{B}_{\infty}^d(1)$, $\text{Q}_{u,B}$ is defined as,
\begin{equation}
    \text{Q}_{u,B}(\zv) \triangleq [z_1', \ldots, z_d']^\top; \hspace{1mm} z_j' \triangleq \argminimize_{u \in \{u_1, \ldots, u_M\}} \lvert u - z_j \rvert
\end{equation}
The worst-case quantization error of a uniform quantizer is $\Delta\sqrt{d}/2$.
We get the uniform scalar quantized direction as $\stv \triangleq \text{Q}_{u,B}(\Xinv\yv/\lVert \Xinv\yv \rVert_2) \in \Real^d$.

\textbf{Getting the quantized model}.
Using the quantized magnitude $\bt^2 \in \Real$ and the quantized direction $\stv \triangleq \text{Q}_{u,B}(\Xinv\yv/\lVert \Xinv\yv \rVert_2) \in \Real^d$, the $\thetatv$ is obtained by multiplying and appropriately scaling as,
\begin{equation}
\label{eq:final_quantized_model_naive}
\thetatv \triangleq \sqrt{\frac{d\bt^4}{\bt^2 + \sigma^2\xi/d}}\cdot\stv = \sqrt{\frac{d\bt^4}{\bt^2 + \sigma^2\xi/d}}\cdot \text{Q}_{u,B}\Paren{\frac{\Xinv\yv}{\lVert \Xinv\yv\rVert_2}},
\end{equation}
where $\xi = \sum_{i=1}^{d}\sigma_i^{-2}$.
Thm. \ref{thm:naive_quantizer_guarantee} below gives a guarantee for this na\"ive strategy while showing that it can be highly suboptimal for large $d$ due to the $O(d)$ scaling of second term of the risk.

\begin{theorem}
\label{thm:naive_quantizer_guarantee}
For $\yv = \Xv\thetav + \vv$, where $\lVert \thetav \rVert_2^2 \leq dc^2$, the output $\thetatv$ of the na\"ive strategy satisfies
\begin{equation*}
\mathbb{P}\Paren{R(\thetatv, \thetav) > \frac{2c^2\sigma^2}{\sigma^2 + c^2\sigma_{\mathrm{min}}^2} + \frac{2\textcolor{red}{d}c^4\sigma_{\mathrm{max}}^22^{-2B}}{\sigma^2 + c^2\sigma_{\mathrm{max}}^2} + \frac{C_n}{\sqrt{d}}} \to 0    
\end{equation*}
as $d \to \infty$ for some constant $C_n$ independent of $d$.
\end{theorem}
Its proof is given in App. \ref{app:naive_quantizer_guarantee}.
This (and all subsequent strategies) requires $\log c^2 + (1/2)\log d$ bits to quantize the magnitude and $dB$ bits to quantize the direction.
The number of bits required per dimension is then $B + (\log c^2)/d + (\log d)/(2d) \to B$ as $d \to \infty$.
To ensure that the \textit{excess risk due to quantization is of the same order of magnitude as the unavoidable learning risk}, we require,
\begin{equation}
\label{eq:threshold_bit_budget_naive}
    \left(\frac{2c^2\sigma^2}{\sigma^2 + c^2\sigma_{\mathrm{min}}^2} - \frac{c^2\sigma^2}{\sigma^2 + c^2\sigma_{\mathrm{max}}^2}\right) + \frac{2dc^4\sigma_{\mathrm{max}}^22^{-2B}}{\sigma^2 + c^2\sigma_{\mathrm{max}}^2} \lesssim \frac{2K_r c^2\sigma^2}{\sigma^2 + c^2\sigma_{\mathrm{max}}^2},
\end{equation}
for some constant $K_r$.
This implies a minimum bit-budget requirement for the \textbf{na\"ive strategy} to be $B_* = \Omega\Paren{\log\Paren{d\omega^2\sigma_{\mathrm{max}}^2}}$, that scales logarithmically with $d$ compared $B_{thr} = \Omega(1)$ in eq. \eqref{eq:threshold_bit_budget}.

\subsection{Random Correlation Maximization (RCM)}
\label{subsec:random_correlation_maximization}

We now propose and analyze \textbf{RCM}, a random source coding scheme that achieves the lower bound (with an additive $\sqrt{(\log d)/d}$ term) of Thm. \ref{thm:minimax_error_lower_bound} for well-conditioned matrices $\Xv$.
\textbf{RCM} estimates $\lVert \thetav \rVert_2$ exactly as in \S \ref{subsec:naive_learning_and_quantization}.
The difference lies in the learning strategy of the direction of $\thetav$.

\textbf{Maximize random correlations to learn the direction}.
We do this by first generating a codebook $\Ycal$ that consists of $2^{dB}$ i.i.d. vectors drawn uniformly at random on the $d$-dimensional unit sphere $\Sbb^{d-1}$.
We learn and encode the direction of $\thetav \in \Real^d$ by $\varphi_{\Ycal}\Paren{\Xv,\yv} = \argmaximize_{\yvo \in \Ycal} \inprod{\Xinv\yv, \yvo}$, and decode it as $\stv = \psi_{\Ycal}\Paren{\varphi_{\Ycal}\Paren{\Xv,\yv}} \in \Real^d$, where $\psi_{\Ycal}\Paren{i}$ returns the $i^{th}$ element in the codebook $\Ycal$.

\textbf{Getting the quantized model}.
Similar to \S \ref{subsec:naive_learning_and_quantization}, the model $\thetatv$ is obtained (after appropriate scaling) as,
\begin{equation}
    \thetatv = \sqrt{\frac{d\bt^4(1 - 2^{-2B})}{\bt^2 + \sigma^2\xi/d}}\cdot\stv.
\end{equation}
The following Thm. \ref{thm:RCM_guarantee} shows that as $d \to \infty$, with high probability over the construction of the random codebook $\Ycal$, the $\ell_2$-risk of \textbf{RCM} is very close to the lower bound.
\begin{theorem}
\label{thm:RCM_guarantee}
For $\yv = \Xv\thetav + \vv$, where $\lVert \thetav \rVert_2^2 \leq dc^2$, the output $\thetatv$ of \textbf{RCM} satisfies,
\begin{equation*}
    \mathbb{P}\left(R(\thetatv, \thetav)  > \frac{c^2\sigma^2}{\Paren{\sigma^2 + c^2\sigma_{\mathrm{min}}^2}}+ \frac{c^4\sigma_{\mathrm{max}}^22^{-2B}}{\Paren{\sigma^2 + c^2\sigma_{\mathrm{max}}^2}} + C\sqrt{\frac{\log d}{d}} \right) \to 0
\end{equation*}
as $d \to \infty$ for some constant $C$ independent of $d$.
\end{theorem}
Its proof is delegated to App. \ref{app:proof_of_thm_RCM_guarantee}.
As in \eqref{eq:threshold_bit_budget_naive}, we can calculate $B_*$ for \textbf{RCM}, given in Table \ref{tab:learning_codes_comparison}.
Thm. \ref{thm:RCM_guarantee} guarantees the \textit{existence of a good quantizer} with performance close to the lower bound of Thm. \ref{thm:minimax_error_lower_bound}.
However, implementing it requires exponential complexity since the codebook $\Ycal$ comprises of $2^{dB}$ random directions, and solving $\argmaximize_{\yvo \in \Ycal} \inprod{\Xinv\yv, \yvo}$ requires us to search over all of them.

\vspace{-1mm}
\subsection{Democratic Quantization (DQ)}
\label{subsec:democratic_quantization}
\vspace{-1mm}

Since \textbf{RCM} is practically infeasible, we now propose \textbf{DQ} that gives us a close to optimal polynomial-time scheme.
In \textbf{DQ}, the scheme for estimating and quantizing the magnitude remains the same as in \S \ref{subsec:naive_learning_and_quantization}.
But instead of uniformly quantizing $\sv = \Xinv\yv/\lVert \Xinv\yv \rVert_2$, we quantize the solution of:
\begin{equation}
\label{eq:l_inf_minimization_problem}
    \sv_d \triangleq \argminimize_{\xv \in \Real^D} \norm{\xv}_{\infty} \hspace{2mm} \text{subject to} \hspace{2mm} \sv = \Sv\xv.
\end{equation}
Here, $D \geq d$ and $\Sv \in \Real^{d \times D}$ is a \textit{wide matrix}.
The solution of \eqref{eq:l_inf_minimization_problem} is called the \textbf{democratic embedding} of $\sv$ w.r.t. $\Sv$.
It can be shown that w.h.p., $\lVert \sv_d \rVert_{\infty} \leq K_u/\sqrt{d}$ for some constant $K_u$ that depends on the choice of $\Sv$.
Compared to the na\"ive strategy which uniformly quantized $\sv$ to get $\stv$, \textbf{DQ} uniformly quantizes the random embedding $\sv_d$ instead.
Previously, $\lVert \sv \rVert_{\infty} \leq 1$ required the \textit{na\"ive} quantizer to have a dynamic range of $[-1,+1]$, but $\lVert \sv_d \rVert_{\infty} \leq K_u/\sqrt{d}$ lets us choose the dynamic range of the uniform quantizer as $[-K_u/\sqrt{d}, +K_u/\sqrt{d}]$ instead.
This reduced dynamic range significantly improves the resolution.
From the quantized embedding, an estimate of $\sv$ can be obtained by taking the inverse embedding.
We denote this by $\text{Q}_{d,B}(\cdot)$ and the decoded direction is,
\begin{equation}
\label{eq:democratic_quantization_shape}
    \stv = \text{Q}_{d,B}(\Xinv\yv/\lVert \Xinv\yv\rVert_2) \triangleq \Sv\cdot \text{Q}_{u,B}(\sv_d).
\end{equation}
We discuss democratic embeddings in detail and relevant choices of $\Sv$ in App. \ref{app:democratic_embeddings}.
In particular, we show that the worst case quantization error is $\lVert \stv - \text{Q}_{d,B}(\sv) \rVert_2 \leq O(1)$, a constant that does not scale with dimension, unlike $O(\sqrt{d})$ for na\"ive uniform quantizers.

\textbf{Getting the quantized model}.
As before in \S \ref{subsec:naive_learning_and_quantization}, \textbf{DQ} obtains the final quantized model according to eq. \eqref{eq:final_quantized_model_naive}, where $\stv$ is given by \eqref{eq:democratic_quantization_shape}.
Thm. \ref{thm:DQ_guarantee} below shows that the $\ell_2$-risk of this strategy is close to the lower bound of Thm. \ref{thm:minimax_error_lower_bound} up to constant factors.
\begin{theorem}
\label{thm:DQ_guarantee}
For $\yv = \Xv\thetav + \vv$, where $\lVert \thetav \rVert_2^2 \leq dc^2$, the output $\thetatv$ of the \textbf{DQ} satisfies,
\begin{equation*}
    \mathbb{P}\Paren{R(\thetatv, \thetav) > \frac{2c^2\sigma^2}{\sigma^2 + c^2\sigma_{\mathrm{min}}^2} + \frac{16K_u^2c^4\sigma_{\mathrm{max}}^22^{-2B/\lambda}}{\sigma^2 + c^2\sigma_{\mathrm{max}}^2} + \frac{C_{dq}}{\sqrt{d}}} \to 0,
\end{equation*}
as $d \to \infty$ for some constant $C_{dq}$ independent of $d$.
\end{theorem}
We prove this in App. \ref{app:DQ_guarantee_proof}.
Here, $K_u$ is a constant (precisely, the upper Kashin constant of $\Sv$ as defined in App. \ref{app:democratic_embeddings}) and $\lambda = D/d$, that depends on the choice of user, and is usually chosen very close to $1$.
%

\subsection{Near-Democratic Quantization (NDQ)}
\label{subsec:near_democratic_quantization}

Although \textbf{DQ} requires $O(d^2)$ multiplications which is much better than \textbf{RCM}, it can still be computationally prohibitive for large $d$, especially when compared to the $O(d)$ complexity of the na\"ive strategy.
To get a \textbf{near-democratic embedding}, an $\ell_2$ relaxed version of \eqref{eq:l_inf_minimization_problem} is solved as follows,
\begin{equation}
\label{eq:l_2_minimization_problem}
    \sv_{nd} \triangleq \argminimize_{\xv \in \Real^D} \norm{\xv}_2 \hspace{2mm} \text{subject to} \hspace{2mm} \sv = \Sv\xv.
\end{equation}
The solution of \eqref{eq:l_2_minimization_problem} can be found in a closed form (App. \ref{app:near_democratic_embeddings}) as $\sv_{nd} = \Sv^\top \sv$.
Although $\Sv$ can be one of several random matrices as discussed in App. \ref{app:democratic_embeddings}, we specifically consider the construction $\Sv = \Pv\Dv\Hv \in \Real^{d \times D}$.
Here, $\Pv \in \Real^{d \times D}$ is the sampling matrix obtained by randomly selecting $d$ rows of $\Iv_D$, $\Dv \in \Real^{D \times D}$ is a diagonal matrix whose entries are randomly chosen to be $\pm 1$ with equal probability, and $\Hv$ is the $D \times D$ Hadamard matrix with normalized entries, i.e. $H_{ij} = \pm 1/\sqrt{D}$.
We let $D = 2^{\lceil \log d \rceil}$ so that $\Hv$ can be constructed.
Using the fast Walsh-Hadamard transform, $\sv_{nd}$ can be computed with just $O(d \log d)$ additions and sign flips, which is near-linear complexity.

We show in App. \ref{app:near_democratic_embeddings} that $\lVert \sv_{nd} \rVert_{\infty} \leq 2\sqrt{(\log D) / D}$, and we again reduce the dynamic range.
\textbf{NDQ} uniformly quantizes $\sv_{nd}$ using $B$ bits per dimension, and takes an inverse embedding.
Precisely,
\begin{equation}
\label{eq:near_democratic_quantization_shape}
    \stv = \text{Q}_{nd,B}\left(\Xinv\yv/\lVert \Xinv \yv \rVert_2\right) \triangleq \Sv\cdot \text{Q}_{u,B}(\sv_{nd}).
\end{equation}

The worst-case quantization error can be shown to be $\lVert \stv - \text{Q}_{nd,B}(\sv) \rVert_2 \leq O(\sqrt{\log d})$, where $\text{Q}_{nd,B}$ denotes \textbf{NDQ}.
Compared to \textbf{DQ}, the logarithmic factor is a mild price to pay considering the significant computational savings.
The final quantized model is obtained in the same way as before according to eq. \eqref{eq:final_quantized_model_naive}, where $\stv$ is given by \eqref{eq:near_democratic_quantization_shape}.
Thm. \ref{thm:NDQ_guarantee} below gives an upper bound to the risk.
\begin{theorem}
\label{thm:NDQ_guarantee}
For data $\yv = \Xv\thetav + \vv$, where $\lVert \thetav \rVert_2^2 \leq dc^2$, the output $\thetatv$ of \textbf{NDQ} satisfies,
\begin{equation*}
\mathbb{P}\Paren{R(\thetatv, \thetav) > \frac{2c^2\sigma^2}{\sigma^2 + c^2\sigma_{\mathrm{min}}^2} + \frac{64c^4\sigma_{\mathrm{max}}^2\textcolor{red}{\log(2\lambda d)}2^{-\frac{2B}{\lambda}}}{\sigma^2 + c^2\sigma_{\mathrm{max}}^2} + \frac{C_{nd}}{\sqrt{d}}} \to 0     
\end{equation*}
as $d \to \infty$ for some constant $C_{nd}$ independent of $d$.
\end{theorem}
The proof of \ref{thm:NDQ_guarantee} is provided in App. \ref{app:NDQ_guarantee_proof}.
It also implies that the minimum budget requirement for \textbf{NDQ} is $B_* = \Omega\left(\log\log d\right)$, very close to the lower bound of $B_{thr} = \Omega(1)$ for practical values of $d$.

\begin{remark}
We can assume that both the encoder and decoder know $\Sv$.
If quantization constraints are imposed by a communication channel \cite{Jenkins_2010}, $\Sv$ need not be communicated between them; it can be generated with shared randomness using initially agreed upon seeds.
Hence, it does not count towards the total bit-requirement of the quantized model.
Moreover, computing $\Pv\Dv\Hv(\cdot)$ requires Hadamard transform and sign-flips which can be done efficiently on hardware.
Although specifying $\Pv$ and $\Dv$ for sampling \& sign-flips requires $d\log(D)$ bits, it can be done using in-place computation or circuits and does not require memory hardware like RAM.
\end{remark}

\vspace{-2mm}
\section{Numerical Simulations}
\label{sec:numerical_simulations}
\vspace{-2mm}

To numerically validate our results, we compare learning codes for several different settings.

\textbf{Linear Regression}.
In Fig. \ref{fig:synthetic_simulations}, we consider learning and subsequently quantizing an estimate $\thetatv$ of $\thetav$ from the training data $(\Xv, \yv)$, where the response $\yv$ is generated according to $\yv = \Xv\thetav + \vv$.
We consider $n = d = 128, c = 1, \sigma = 1$, and all the plots have been averaged over $10$ realizations.
In Figs. \ref{fig:X_Identity_Theta_Student-t}, \ref{fig:X_PerturbedOrthonormal_Theta_Student-t}, and \ref{fig:X_PerturbedOrthonormal_Theta_Gaussian3}, $\Xv$ is either an identity matrix, or a perturbed orthonormal matrix.
When $\Xv = \Iv_n$, $\sigma_{\mathrm{max}} = \sigma_{\mathrm{min}}$, and for perturbed orthonormal, we have $\sigma_{\mathrm{max}} \approx \sigma_{\mathrm{min}}$.
The coordinates of the latent parameter $\thetav$ are drawn i.i.d., either from Student-t, or from $\Ncal(0,1)$ and then cubed, i.e. $\Ncal(0,1)^3$.

We plot the results of na\"ive uniform quantization and near-democratic model quantization.
For \textbf{NDQ}, we distinguish between the cases when $\Sv$ used to obtain the near-democratic embedding in \eqref{eq:l_2_minimization_problem} is
\begin{enumerate*}[(i)]
    \item Random orthonormal, and
    \item Randomized Hadamard ($\Pv\Dv\Hv$).
\end{enumerate*}
A random Haar orthonormal matrix is constructed by first obtaining an $n \times d$ matrix with i.i.d. $\Ncal(0,1)$ entries, computing its singular value decomposition, and then multiplying the matrices of its left and right singular vectors.
For comparison, we also plot the \textbf{lower bound} of Thm. \ref{thm:minimax_error_lower_bound}.
There is a noticeable difference in performance between the \textit{na\"ive} and \textit{randomized embedding-based} strategies, especially at low bit-budgets $B$.
When $B < 1$, i.e. less than $1$-bit is available to quantize each coordinate, some coordinates are randomly set to $0$, so that the average bit requirement of the learning code is still $B$.

\begin{remark}
We consider a minimax framework for our lower bound.
In our setup, for some small $\varepsilon > 0$, consider {$\thetav = [\Paren{\varepsilon/(d-1)}^{1/2}, \ldots, (c^2d - \varepsilon)^{1/2}, \ldots, \Paren{\varepsilon/(d-1)}^{1/2}]$} where just one arbitrary entry is large and $\norm{\thetav}_{\infty} \approx c\sqrt{d}$.
The worst-case $\thetav$ satisfies $\norm{\thetav}_2 \approx \norm{\thetav}_{\infty} c\sqrt{d}$, as na\"ive quantization yields the worst case error $\norm{Q(\thetav) - \thetav}_2 \approx 2^{(1-B)}c\sqrt{d}$.
\textit{Student-t} and $\Ncal(0,1)^3$ are heavy-tailed distributions.
When $\thetav$ is sampled from either of them, its coordinate magnitudes are distributed very non-uniformly, i.e. close to worst-case, and we compare this case with the minimax lower bound.
\end{remark}

\textbf{Neural-Network}.
Our proposed randomized embedding-based quantization schemes can be extended beyond linear regression models, to other general settings where no prior distributional information about the model weights is known.
To show this, in Fig. \ref{fig:TwoLayerNN_FuelEfficiency}, we train a $2$-layer densely connected ReLU neural network ($32$ neurons) with scalar output for non-linear regression to predict the fuel efficiency of cars from parameters like \textit{horsepower, weight}, etc for the \textit{auto\_mpg} dataset from UCI Machine Learning repository \cite{Dua:2019}.
We do layer-wise quantization and treat the weights and biases of each layer separately.
To quantize the model parameters, we first vectorize the weight matrix of any specific layer, and then use randomized Hadamard-based \textbf{NDQ}.
Since the total number of parameters is not necessarily a power of $2$, in order to be able to construct Hadamard matrices, we cluster the weight vector into sub-vectors, each of which is a largest possible power of $2$.
\textbf{NDQ} can be seen to perform better than the na\"ive uniform scalar quantization strategy.
The upper bound for this neural network setting is analytically obtained in App. \ref{app:NN_analysis}.

\textbf{Comparison between \textit{Na\"ive} vs. \textit{Hadamard} quantizers}.
While \textbf{NDQ} requires $O(d \log d)$ complexity compared to $O(d)$ of the na\"ive strategy, the additional $\log(\cdot)$ factor in the runtime for \textbf{NDQ} is not prohibitively large, especially when there are other slower steps in a larger task pipeline, eg. when models need to be transmitted from one device to another over bandwidth-constrained communication links.
Here, quantization error is more of a concern than encoding/decoding complexity.
In Figs. \ref{fig:error_vs_dim} and \ref{fig:time_vs_dim}, we compare the quantization error and elapsed system time, respectively, with increasing dimension.
We use \textit{Python SymPy} \cite{sympy_citation} to compute the fast Walsh-Hadamard transform.

\begin{figure}[t!]
    \centering
    \begin{minipage}{.33\textwidth}
        \centering
        \begin{subfigure}[b]{\textwidth}
           \includegraphics[width=1\linewidth]{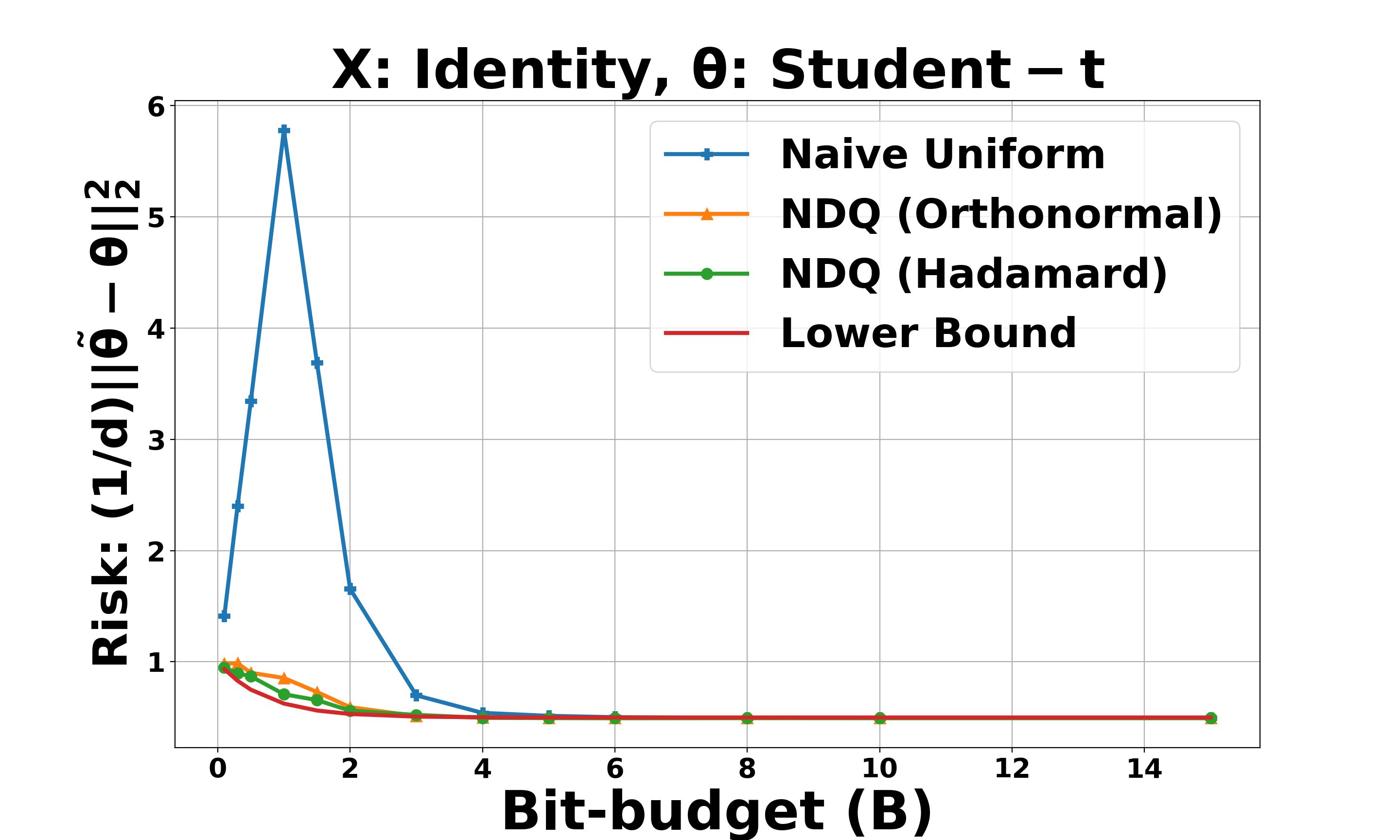}
           \caption{}
           \label{fig:X_Identity_Theta_Student-t} 
        \end{subfigure}
    \end{minipage}%
    \begin{minipage}{0.33\textwidth}
        \centering
        \begin{subfigure}[b]{\textwidth}
           \includegraphics[width=1\linewidth]{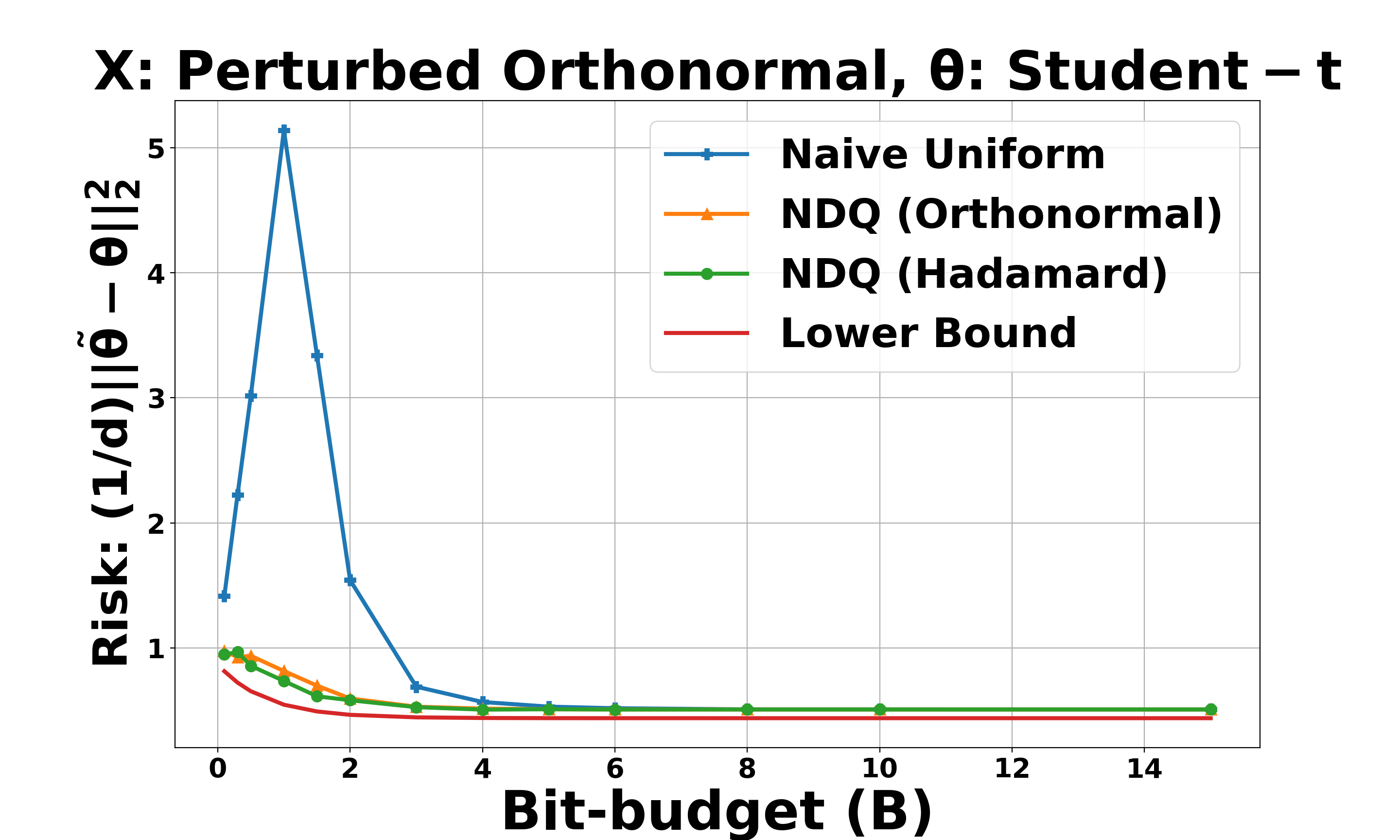}
           \caption{}
           \label{fig:X_PerturbedOrthonormal_Theta_Student-t} 
        \end{subfigure}
    \end{minipage}
    \begin{minipage}{0.33\textwidth}
        \centering
        \begin{subfigure}[b]{\textwidth}
           \includegraphics[width=1\linewidth]{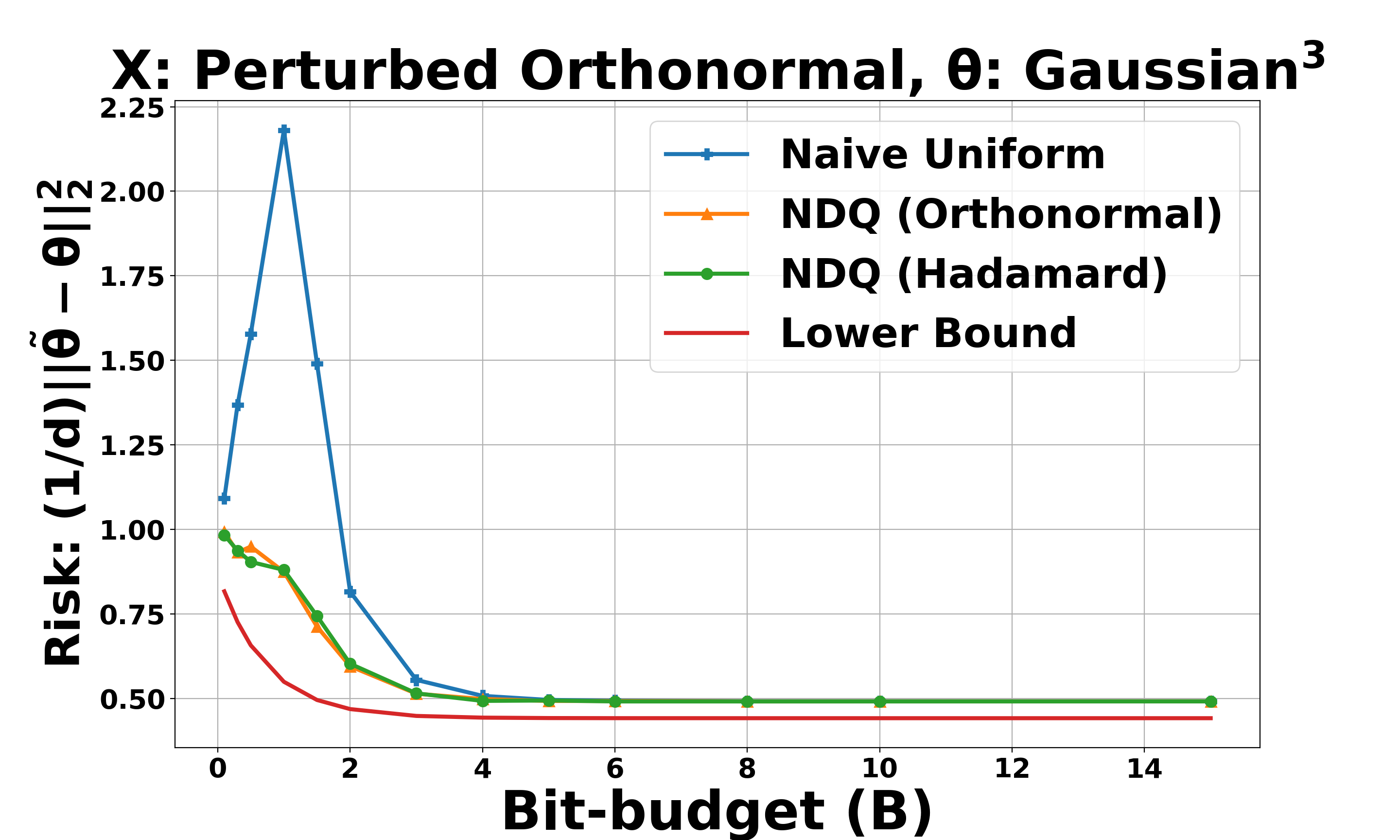}
           \caption{}
           \label{fig:X_PerturbedOrthonormal_Theta_Gaussian3} 
        \end{subfigure}
    \end{minipage}
    \vspace{5mm}
    \caption{Learning and Quantizing Linear Regressor $\thetav$ from $(\Xv, \yv)$, related by $\yv = \Xv\thetav + \vv$.}
    \vspace{2mm}
    \label{fig:synthetic_simulations}
\end{figure}
\begin{figure}[t!]
    \centering
    \begin{minipage}{.33\textwidth}
        \centering
        \begin{subfigure}[b]{\textwidth}
           \includegraphics[width=1\linewidth]{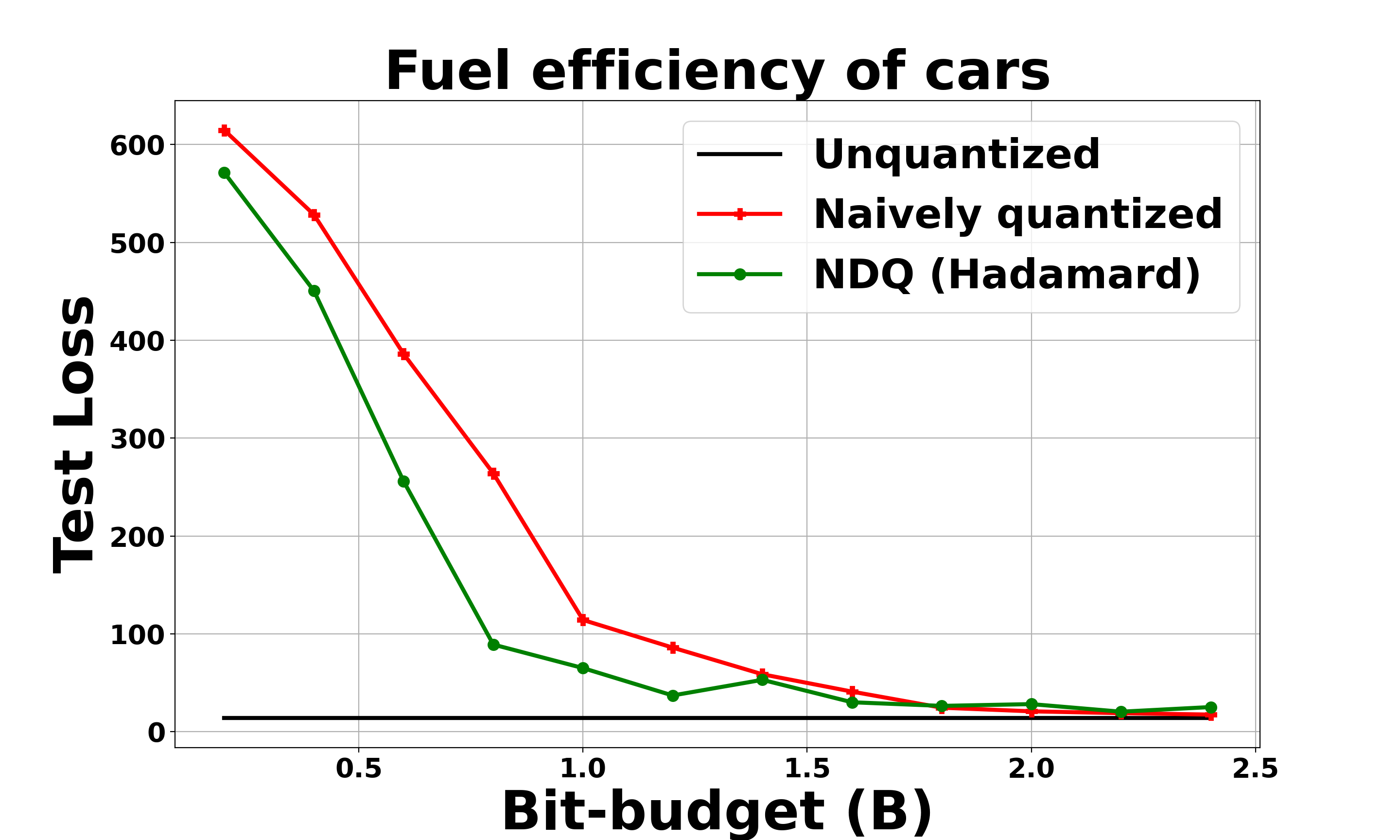}
           \caption{Non-Linear Regression with NN}
           \label{fig:TwoLayerNN_FuelEfficiency} 
        \end{subfigure}
    \end{minipage}%
    \begin{minipage}{0.33\textwidth}
        \centering
        \begin{subfigure}[b]{\textwidth}
           \includegraphics[width=1\linewidth]{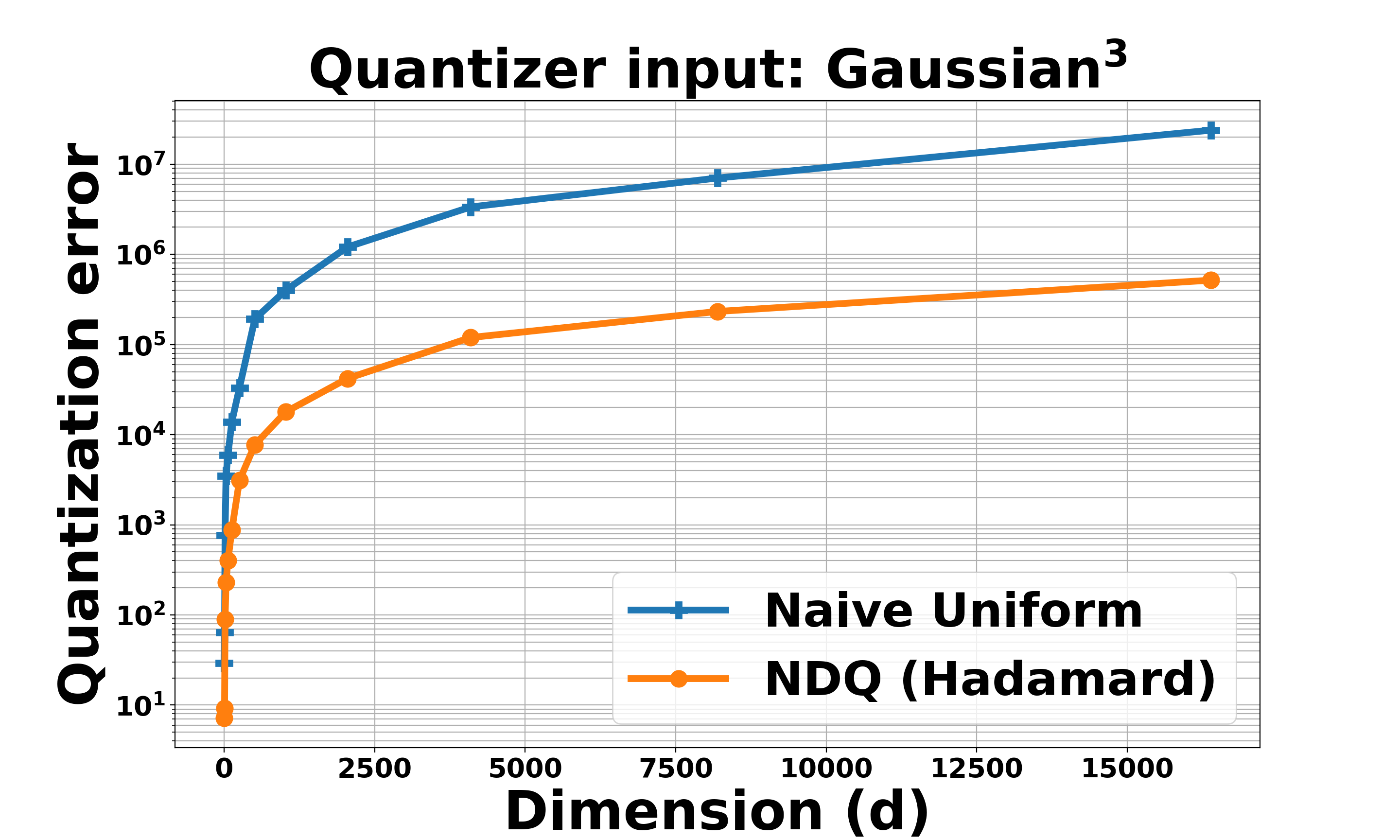}
           \caption{Quantization error vs. dim.}
           \label{fig:error_vs_dim} 
        \end{subfigure}
    \end{minipage}
    \begin{minipage}{0.33\textwidth}
        \centering
        \begin{subfigure}[b]{\textwidth}
           \includegraphics[width=1\linewidth]{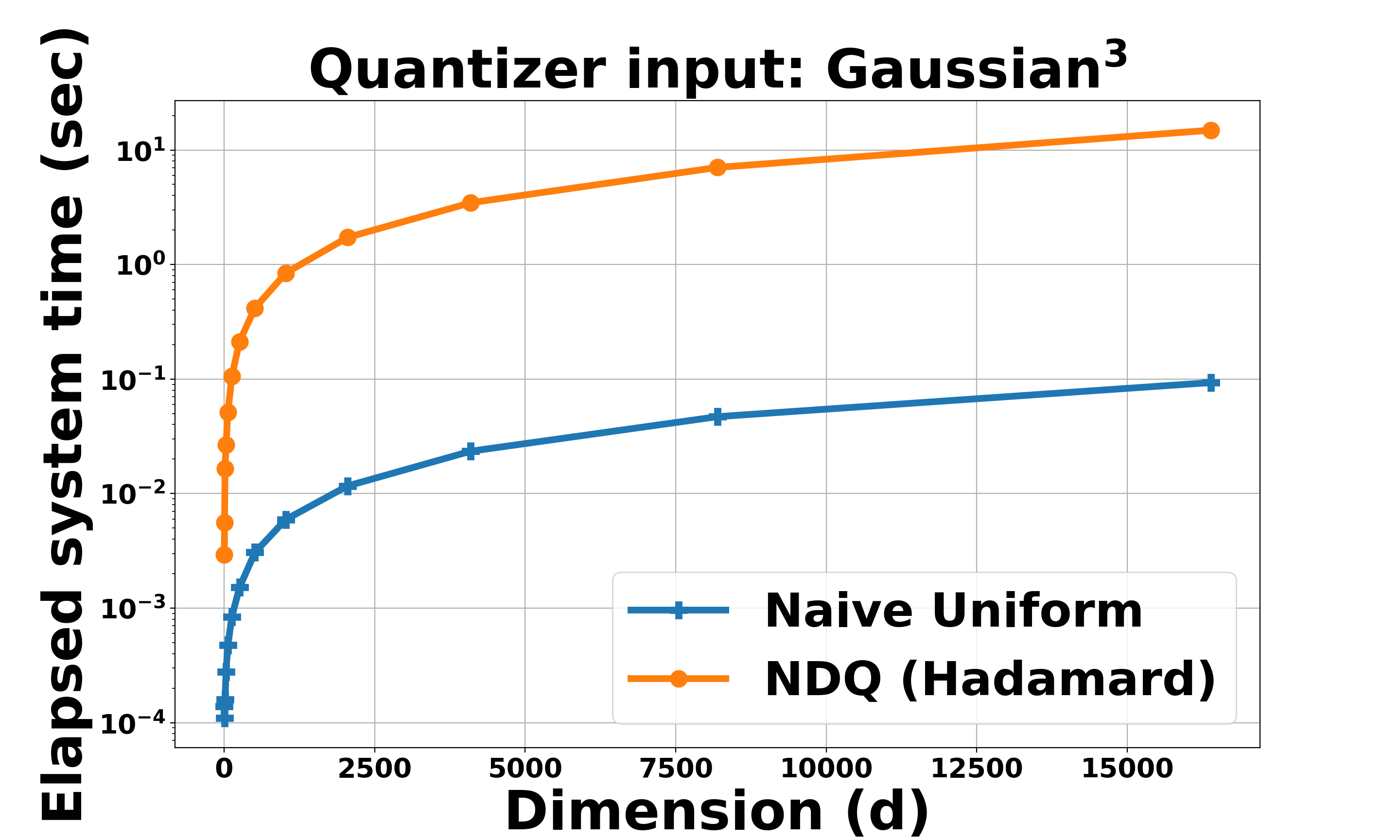}
           \caption{Elapsed time vs. dim.}
           \label{fig:time_vs_dim} 
        \end{subfigure}
    \end{minipage}
    \vspace{5mm}
    \caption{}
\end{figure}

\vspace{-2mm}
\section{Conclusions}
\label{sec:conclusions}
\vspace{-2mm}

In this work, we considered model quantization for least-squares regression, where the solution is constrained to be representable using a pre-specified budget of $B$-bits per dimension.
For the setting when the observations are noisy linear measurements of some true regressor, we obtained an information-theoretic lower bound for the minimax $\ell_2$-risk of learning under quantization constraints.
We saw that the $\ell_2$-risk can be separated as the sum of two terms: the \textit{unavoidable learning risk} and the \textit{excess quantization risk}.
We subsequently proposed several algorithms for model quantization and compared their performance in terms of the worst-case $\ell_2$-risk and computational requirements.
Our results showed the tightness of our lower bound as well as the minimax optimality of our proposed algorithms.
Although na\"ive quantization is computationally simple, its error is sub-optimal for large $d$.
Contrastingly, \textbf{RCM} and \textbf{DQ} yield small quantization errors, but require exponential and quadratic computations, respectively.
We proposed \textit{randomized Hadamard}-based \textbf{NDQ} as a practically feasible strategy with near-optimal error.
Finally, we showed that our method and upper bounds can be extended for non-linear regression using $2$-layer densely connected ReLU neural networks.

\bibliographystyle{IEEEtran}
\bibliography{refs}

\clearpage

\appendices

\section{Summary of abbreviations and notations}
\label{sec:summary_notations_and_acronyms}

\begin{longtblr}[
  caption = {\textsc{\textcolor{darkblue}{Summary of General Abbreviations}}},
  label = {tab:summary_of_acronyms},
]{
  colspec = {|p{0.14\linewidth} | p{0.47\linewidth} | p{0.28\linewidth}|},
  rowhead = 1,
  hlines,
  vlines,
  row{even} = {gray9},
  row{1} = {olive9},
}
\centering \textsc{\textbf{Acronym}} & \centering \textsc{\textbf{Meaning}} & \centering \textsc{\textbf{Remarks}}\\
\centering QLS & Quantized Least Squares & Def. in \eqref{eq:bit_constrained_least_squares}\\
\centering NPM & Noisy Planted Model & Def. in \eqref{eq:noisy_planted_model}\\
\centering RCM & Random Correlation Maximization & \S \ref{subsec:random_correlation_maximization}\\
\centering DQ & Democratic Quantization & \S \ref{subsec:democratic_quantization}\\
\centering NDQ & Near-Democratic Quantization & \S \ref{subsec:near_democratic_quantization}\\
\end{longtblr}

\begin{longtblr}[
  caption = {\textsc{\textcolor{darkblue}{Summary of General Notations}}},
  label = {tab:summary_of_general_notations},
]{
  colspec = {|p{0.14\linewidth} | p{0.47\linewidth} | p{0.28\linewidth}|},
  rowhead = 1,
  hlines,
  vlines,
  row{even} = {gray9},
  row{1} = {olive9},
}
\centering \textsc{\textbf{Notation}} & \centering \textsc{\textbf{Meaning}} & \textsc{\textbf{Remarks}} \\
\centering $d$ & Dimension of model & --- \\
\centering $B$ & (Pre-specified) Bit budget & $B \in (0, \infty)$\\
\centering $n$ & Number of linear measurements & $n \geq d$\\
\centering $\thetav$ & Latent (true) parameter & \centering $\thetav \in \Thetav \subseteq \Real^d$ \\
\centering $\Thetav$ & Parameter space & $\Thetav \triangleq \{\zv : \frac{1}{d}\lVert \zv \rVert_2^2 \leq c^2\}$\\
\centering $c$ & Radius of the (Euclidean ball) parameter space & ---\\
\centering $\yv$ & Observations / Response & $\yv \in \Real^n$\\
\centering $\Xv$ & Feature matrix & $\Xv \in \Real^{n \times d}$\\
\centering $\vv$ & Noise & $\vv \sim \Ncal(\mathbf{0}, \sigma^2\Iv_n)$\\
\centering $\sigma$ & Gaussian noise variance & ---\\
\centering $\thetatv$ & Estimate of $\thetav$ representable using $dB$ bits & (Approximate) solution of \eqref{eq:bit_constrained_least_squares} in $\Svcal \subseteq \Thetav$\\
\centering $\Svcal$ & Codebook (of finite cardinality) & $\thetatv \in \Svcal$, $\lvert \Svcal \rvert \leq 2^{dB}$\\
\centering $\Esf(\cdot, \cdot)$ & Encoder & $\Esf:\Real^{n \times d} \times \Real^n \to \{1, \ldots, 2^{dB}\}$\\
\centering $\Dsf(\cdot)$ & Decoder & $\Dsf:\{1, \ldots, 2^{dB}\} \to \Thetav$\\
\centering $\Qsf(\cdot, \cdot)$ & $(n,d,B)$-Learning code & $\Qsf:\Real^{n \times d} \times \Real^n \to \Thetav$, $\Qsf(\cdot, \cdot) \triangleq \Dsf(\Esf(\cdot, \cdot))$\\
\centering $\Qcal_{n,d,B}$ & Class of all $(n,d,B)$-learning codes & $\Qsf(\cdot, \cdot) \in \Qcal_{n,d,B}$\\
\centering $R(\thetatv, \thetav)$ & Expected $\ell_2$-risk & Def. in eq. \eqref{eq:l2_risk_definition}\\
\centering $\Rcal_{\Qsf}$ & Worst-case risk over $\Thetav$ & $\sup_{\thetav \in \Thetav} R(\Qsf(\cdot,\cdot), \thetav)$\\
\centering $\Rcal_{B,c}$ & Minimax risk over $\Qcal_{n,d,B}$ & Def. in eq. \eqref{eq:asymptotic_minimax_risk}\\
\centering $I(\cdot, \cdot)$ & Mutual information between two random variables & Ref. to \cite[Ch. 2]{info_theory_book}\\
\centering $\sigma_{\rm max}, \sigma_{\rm min}$ & Upper \& Lower bound to max. and min. singular values of $\Xv$ & ---\\
\centering $\Rcal_{B,c}^{(\infty)}$ & Asymptotic minimax risk & $\liminf_{d \to \infty} \Rcal_{B,c}$\\
\centering $\omega^2$ & Problem Parameter-to-Noise ratio & $\omega^2 \triangleq c^2/\sigma^2$ (Def. in eq. \eqref{eq:threshold_bit_budget})\\
\centering $\Xinv$ & Moore-Penrose pseudoinverse & $\Xinv \in \Real^{d \times n}$\\
\end{longtblr}

\begin{longtblr}[
  caption = {\textsc{\textcolor{darkblue}{Summary of Notations for Derivation of Lower Bound}}},
  label = {tab:summary_of_lower_bound},
]{
  colspec = {|p{0.14\linewidth} | p{0.47\linewidth} | p{0.28\linewidth}|},
  rowhead = 1,
  hlines,
  vlines,
  row{even} = {gray9},
  row{1} = {olive9},
}
\centering\textsc{\textbf{Notation}} & \centering\textsc{\textbf{Meaning}} & \centering\textsc{\textbf{Remarks}} \\
\centering$H(\cdot)$ & Entropy of a discrete random variable & Ref. to App. \ref{app:proof_mutual_information_bit_budget_inequality}\\
\centering$F_{\delta}$ & Gaussian distribution parameterized by $\delta \in (0,1)$ & $F_{\delta} \equiv \Ncal(\mathbf{0}, c^2\delta^2\Iv_d)$\\
\centering$h(\cdot)$ & Differential entropy of a continuous random variable & Ref. of eq. \eqref{eq:splitting_bayes error}\\
\centering$\muov, \Sigmaov$ & Conditional expectation and covariance of $\thetav$ given $y$ in \eqref{eq:noisy_planted_model} & Ref. to \eqref{eq:conditional_mean_and_variance}\\
\centering$\Tr{\cdot}$ & Trace of a matrix & ---\\
\centering $\Cov{\cdot}$ & Covariance matrix of a random vector & ---\\
\centering $\thetahv$ & Maximum likelihood estimator of $\thetav$ for a Gaussian prior & $\thetahv = \muov$\\
\centering $\Uv, \Sigmav, \Vv$ & SVD matrices of $\Xv$ & $\Xv = \Uv \Sigmav \Vv^{\top}$\\
\centering $\Sigmav_1$ & Eigenvalue matrix of $\Xv^{\top}\Xv \in \Real^{d \times d}$ & $\Sigmav_1 \triangleq \Sigmav^{\top}\Sigmav$ \\
\centering $\lvert \cdot \rvert$ & Determinant of a matrix & ---\\
\end{longtblr}

\begin{longtblr}[
  caption = {\textsc{\textcolor{darkblue}{Summary of Notations for Naive Learning and Quantization}}},
  label = {tab:summary_of_naive_quantization_notations},
]{
  colspec = {|p{0.14\linewidth} | p{0.47\linewidth} | p{0.28\linewidth}|},
  rowhead = 1,
  hlines,
  vlines,
  row{even} = {gray9},
  row{1} = {olive9},
}
\centering\textsc{\textbf{Notation}} & \centering\textsc{\textbf{Meaning}} & \centering\textsc{\textbf{Remarks}} \\
\centering $b$ & (True) magnitude of latent parameter & $b^2 \triangleq \frac{1}{d}\lVert \thetav \rVert_2^2$, Def. in App. \ref{app:naive_quantizer_guarantee}\\
\centering $\bhat$ & Unquantized estimate of $b$ & Ref. to \S \ref{subsec:naive_learning_and_quantization}\\
\centering $\Bcal$ & Codebook for quantizing magnitude & Ref. to \S \ref{subsec:naive_learning_and_quantization}\\
\centering $\varphi_{\Bcal}(\cdot, \cdot)$ & Encoder function for $\bhat^2$ & Ref. to \S \ref{subsec:naive_learning_and_quantization}\\
\centering $\psi_{\Bcal}(\cdot)$ & Decoder function for $\bhat^2$ & Ref. to \S \ref{subsec:naive_learning_and_quantization}\\
\centering $\bt$ & Quantized estimate of the magnitude & Ref. to \S \ref{subsec:naive_learning_and_quantization}\\
\centering $\{\sigma\}_{i=1}^{d}$ & Singular values of $\Xv \in \Real^{n \times d}$, $n \geq d$ & ---\\
\centering $\sv$ & Shape (Direction) of the quantizer input & $\sv \in \Real^d$, Ref. to \S \ref{subsec:naive_learning_and_quantization}\\
\centering $\mathbf{B}_{\infty}^d(1)$ & $\ell_{\infty}$-ball of radius $1$ & ---\\
\centering $Q_{u,B}(\cdot)$ & Uniform quantizer with $B$-bits per dimension & Quantizer input in $\Real^d$\\
\centering $\stv$ & Quantized direction / shape & ---\\
\centering $M$ & No. of quantization points along each dimension & $M = 2^B$\\
\centering $\Delta$ & Quantization resolution per dimension & $\Delta = 2/M$\\
\centering $\xi$ & $\xi = \sum_{i=1}^{d}\sigma_i^{-2}$ & ---\\
\centering $\gammat$ & Scaled and quantized estimate of magnitude & Ref. to eq. \eqref{eq:gamma_dash_scaling}\\
\centering $\gammah$ & Auxiliary variable used for splitting the $\ell_2$-risk expression & Ref. to eq. \eqref{eq:risk_decomposition_naive}\\
\centering $\Av$ & Orthonormal transform that aligns $\thetav$ with first canonical basis vector in $\Real^d$ & Def. in App. \ref{app:naive_quantizer_guarantee}. $\Av\thetav = \tauv \triangleq [\sqrt{d}b, 0, \ldots, 0]$\\
\centering $\ytv$ & Auxiliary variable & $\ytv \triangleq \Av\Xinv\vv$\\
\end{longtblr}

\begin{longtblr}[
  caption = {\textsc{\textcolor{darkblue}{Summary of Notations for Random Correlation Maximization}}},
  label = {tab:summary_of_random_correlation_maximization},
]{
  colspec = {|p{0.14\linewidth} | p{0.47\linewidth} | p{0.28\linewidth}|},
  rowhead = 1,
  hlines,
  vlines,
  row{even} = {gray9},
  row{1} = {olive9},
}
\centering\textsc{\textbf{Notation}} & \centering\textsc{\textbf{Meaning}} & \centering\textsc{\textbf{Remarks}} \\
\centering $\Ycal$ & Codebook for quantizing direction & Def. in \S \ref{subsec:random_correlation_maximization}\\
\centering $\varphi_{\Ycal}(\cdot, \cdot)$ & Encoder function for RCM & Ref. to \S \ref{subsec:random_correlation_maximization}\\
\centering $\psi_{\Ycal}(\cdot)$ & Decoder function for RCM & Ref. to \S \ref{subsec:random_correlation_maximization}\\
\centering $\thetaov$ & Auxiliary variable & Def. in eq. \eqref{eq:auxiliary_variable_1_RCM}\\
\centering $\gamma'$ & Projection of $\thetav$ on to $\Xinv\yv$ & $\gamma' \triangleq \frac{\inprod{\thetav, \Xinv\yv}}{\lVert \Xinv\yv \rVert_2}$, cf. App. \ref{app:proof_lemma_RCM_fixed_model_magnitude}\\
\centering $\Xinv\yv^{\perp}$ & Orthogonal subspace of $\Xinv\yv$ & Ref. to App. \ref{app:proof_lemma_RCM_fixed_model_magnitude}\\
\end{longtblr}

\begin{longtblr}[
  caption = {\textsc{\textcolor{darkblue}{Summary of Notations for Democratic and Near-Democratic Quantization}}},
  label = {tab:summary_of_democratic_and_near_democratic},
]{
  colspec = {|p{0.14\linewidth} | p{0.47\linewidth} | p{0.28\linewidth}|},
  rowhead = 1,
  hlines,
  vlines,
  row{even} = {gray9},
  row{1} = {olive9},
}
\centering\textsc{\textbf{Notation}} & \centering\textsc{\textbf{Meaning}} & \centering\textsc{\textbf{Remarks}} \\
\centering $D$ & Embedding dimension & $D \geq d$\\
\centering $\Sv$ & Frame that satisfies Uncertainty Principle (UP) & $\Sv \in \Real^{d \times D}$\\
\centering $\lambda$ & Aspect ratio of the frame $\Sv$ & $\lambda \triangleq D/d$\\
\centering $\sv_d$ & Democratic embedding & Ref. to eq. \eqref{eq:democratic_embedding_definition}\\
\centering $K_u$ & Upper Kashin constant of frame $\Sv$ & Ref. to App. \ref{app:democratic_embeddings}\\
\centering $(\eta, \delta)$ & Uncertainty principle parameters & Ref. to Def. \ref{def:uncertainty_principle}\\
\centering $Q_{d,B}(\cdot)$ & Democratic quantizer subject to $B$ bits per dimension & Ref. to eq. \eqref{eq:democratic_quantization_shape}\\
\centering $\sv_{nd}$ & Near-democratic embedding & Ref. to eq. \eqref{eq:near_democratic_embedding_definition}\\
\centering $\Pv$ & Sampling matrix obtained by randomly selecting $d$ rows from $\Iv_D$ & $\Pv \in \Real^{d \times D}$, Ref. to \S \ref{subsec:near_democratic_quantization}\\
\centering $\Dv$ & Diagonal matrix with Rademacher entries & $\Dv \in \Real^{D \times D}$\\
\centering $\Hv$ & Hadamard matrix with normalized entries & $H_{ij} = \pm \frac{1}{\sqrt{D}}$\\
\end{longtblr}

\begin{longtblr}[
  caption = {\textsc{\textcolor{darkblue}{Summary of Notations for Extension to Neural Networks (App. \ref{app:NN_analysis})}}},
  label = {tab:summary_of_neural_networks},
]{
  colspec = {|p{0.14\linewidth} | p{0.47\linewidth} | p{0.28\linewidth}|},
  rowhead = 1,
  hlines,
  vlines,
  row{even} = {gray9},
  row{1} = {olive9},
}
\centering\textsc{\textbf{Notation}} & \centering\textsc{\textbf{Meaning}} & \centering\textsc{\textbf{Remarks}} \\
\centering $d$ & Dimension of the NN input & ---\\
\centering $\xv$ & Input to neural network & $\xv \in \Real^d$\\
\centering $f(\cdot)$ & NN output with unquantized weights & $f(\cdot) \in \Real$\\
\centering $\ft(\cdot)$ & NN output with quantized weights & $\ft(\cdot) \in \Real$\\
\centering $m$ & No. of (hidden) neurons in second layer & ---\\
\centering $\Wv$ & NN weights for first layer & $\Wv \in \Real^{m \times d}$\\
\centering $\wv$ & NN weights for the second layer & $\wv \in \Real^m$\\
\centering $\Whv$ & First \& Second layer weights multiplied & Ref. to eq. \eqref{eq:scalar_output_NN_reparametrized}\\
\centering $\Wtv$ & Quantized NN weights & $\Whv \triangleq Q(\Whv)$\\
\centering $B^*(\epsilon)$ & Bit-budget required to attain $\epsilon$-error in NN output & ---\\
\end{longtblr}

\section{Additional simulations}
\label{sec:additional_simulations}

\subsection{Linear regression on \texorpdfstring{$\mathrm{ash331}$}{ash331}}
\label{subsec:linear_regression}

We do additional simulations for quantizing the solution of the least squares regression problem
\begin{equation}
    \zetav^* \triangleq \argminimize_{\zetav \in \Real^d}\frac{1}{2}\lVert \yv - \Xv\zetav \rVert_2^2,
\end{equation}
for a real dataset.
Here, $\yv \in \Real^n$ and $\Xv \in \Real^{n \times d}$.
We consider our feature matrix $\Xv$ to be the sparse matrix \textit{ash331} from \cite{suitesparse}.
Curated by Askenazi as part of Scotland's survey, \textit{ash331} is a $331 \times 104$ full rank binary matrix comprising of $0/1$ entries mostly concentrated around the diagonal, and has $\sigma_{\rm max} = 4.15$ and $\sigma_{\rm min} = 1.34$.
Such sparse graphs often depicts the network connectivity and are useful in analyzing local neighborhood interactions \cite{gsp_survey}.
The \textit{ash331} matrix is visualized as a network in Fig. \ref{fig:ash331_network} using the network visualization tool of \cite{nr}.
The network has $331$ nodes.
In its most general form, the adjacency matrix of a graph is a square matrix.
$\Xv \in \Real^{331 \times 104}$ implies that the columns indexed from $105 \ldots 331$ in the corresponding square adjacency matrix, are zero-valued.

\begin{wrapfigure}{r}{0.54\linewidth}
    \centering
    \includegraphics[width=1\linewidth]{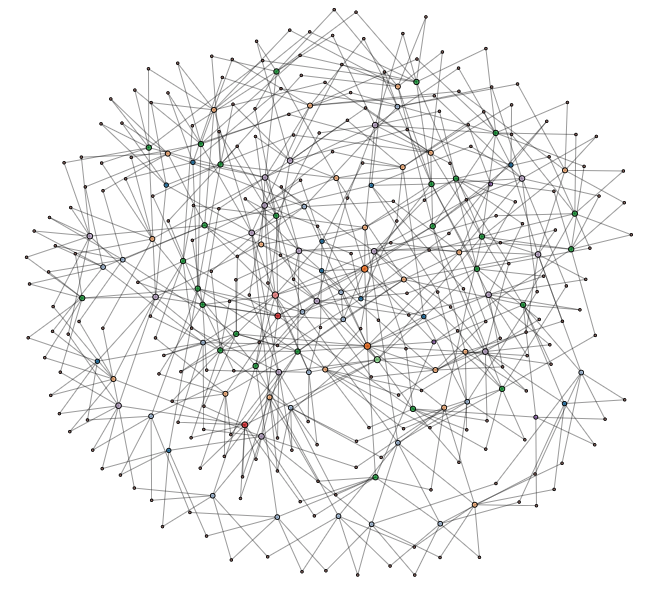}
    \caption{$\mathrm{ash331}$ visualized as a network of nodes}
    \label{fig:ash331_network}
\end{wrapfigure}

Suppose the latent parameter $\thetav \in \Real^{104}$ corresponds to some initial scalar values at nodes labelled $1 \ldots 104$, and suppose the nodes $105 \ldots 331$ have an initial value of zero.
Then, the transformation $\Xv\thetav \in \Real^{331}$ gives the new values after each node has computed a local consensus of its neighbors.
Furthermore, the values after $\Xv\thetav + \vv \in \Real^{331}$ correspond to a noisy consensus.
The goal here is to recover a quantized estimate of the initial values $\thetav$ from the (observed) noisy consensus values $\yv = \Xv\thetav + \vv$. 
We consider two heavy-tailed distributions for the latent parameter $\thetav \stackrel{iid}{\sim} \mathrm{Student-t}$ and $\thetav \stackrel{iid}{\sim} \mathrm{Gaussian}^3$, where $\mathrm{Gaussian}^3$ denotes the i.i.d. entries drawn from $\Ncal(0,1)$ and subsequently cubed.
These settings correspond to Figs. \ref{fig:ash331_student-t} and \ref{fig:ash331_gaussian3} respectively.
Here, $n = 330$ and $d = 104$.
The plots are averaged over $10$ realizations, and $\thetav$ is sampled from either $\mathrm{Student-t}$ or $\mathrm{Gaussian}^3$ independently for each realization.
Moreover, the corresponding observations $\yv \in \Real^n$ are generated according to the model \eqref{eq:noisy_planted_model}.
We plot the variation of the mean squared risk, $R(\thetatv, \thetav) \triangleq \frac{1}{d}\lVert \thetatv - \thetav\rVert_2^2$ with respect to the bit-budget $B$ used for quantizing the learned model.
As was the trend before, \textbf{NDQ} with either Hadamard or random orthonormal frame, outperforms the na\"ive quantization strategy.

\begin{figure}[h!]
    \centering
    \begin{minipage}{0.5\textwidth}
        \centering
        \begin{subfigure}[b]{\textwidth}
           \includegraphics[width=1\linewidth]{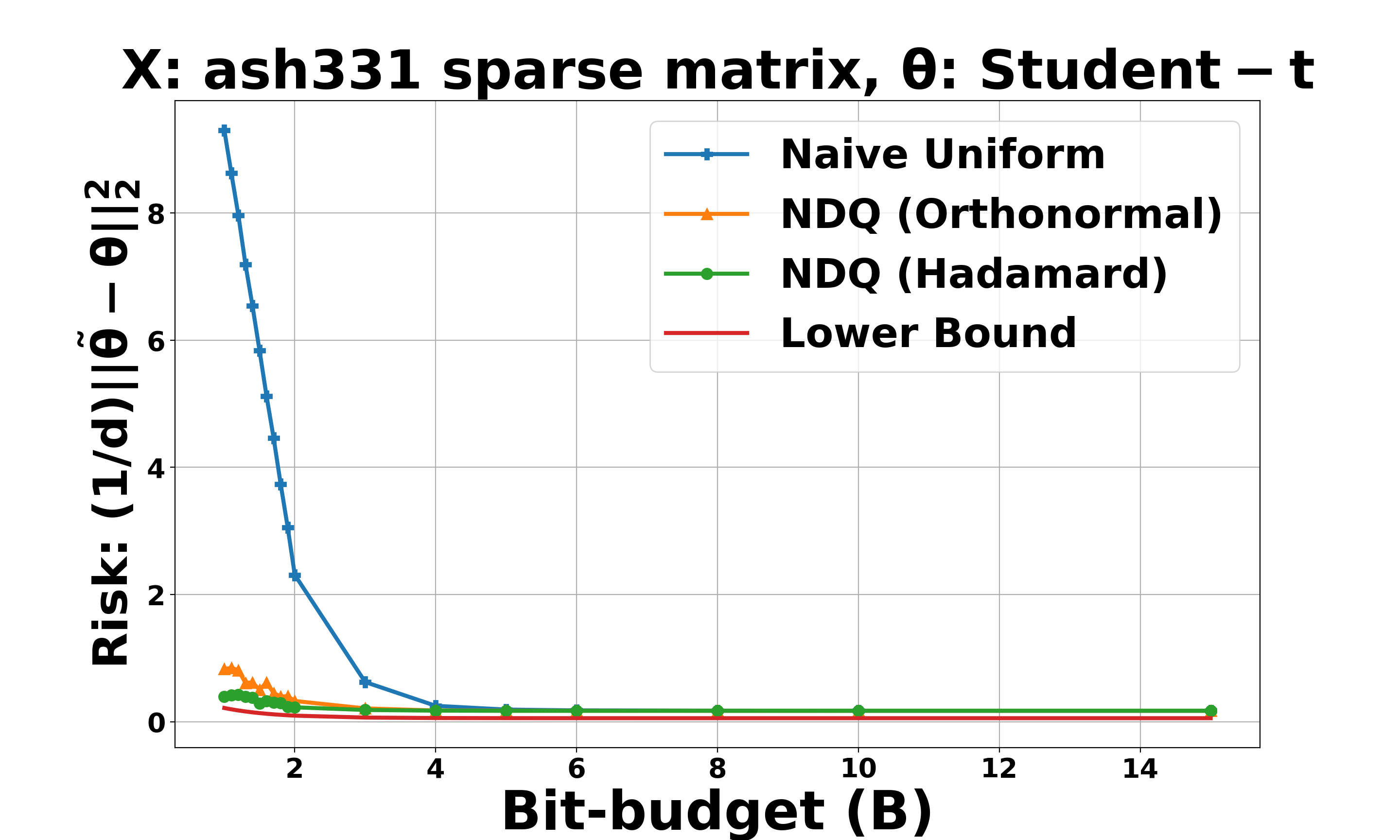}
           \caption{}
           \label{fig:ash331_student-t} 
        \end{subfigure}
    \end{minipage}%
    \begin{minipage}{0.5\textwidth}
        \centering
        \begin{subfigure}[b]{\textwidth}
           \includegraphics[width=1\linewidth]{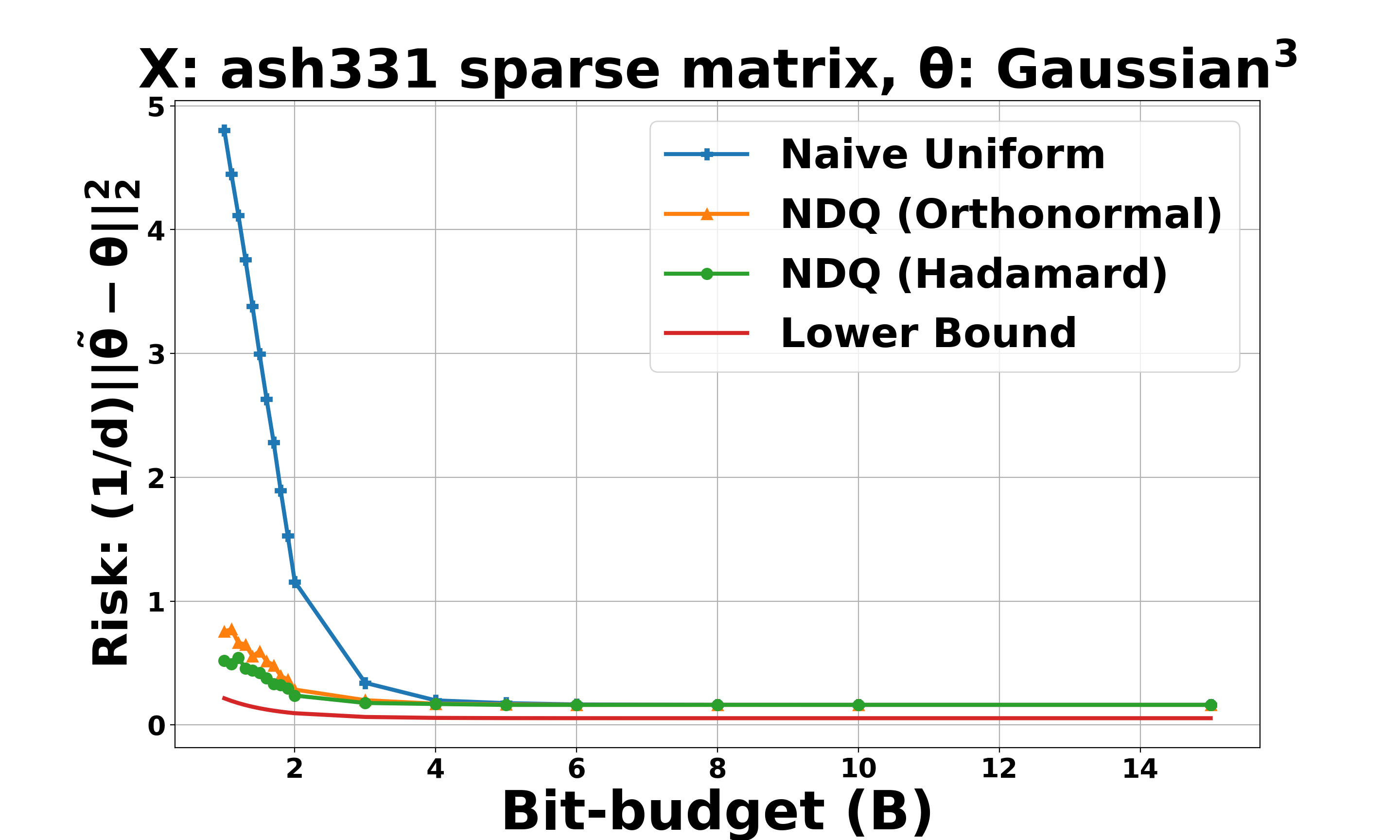}
           \caption{}
           \label{fig:ash331_gaussian3} 
        \end{subfigure}
    \end{minipage}
    \vspace{6mm}
    \caption{MSE risk vs. bit-budget plots for $\mathrm{ash331}$ dataset}
\end{figure}

\subsection{Non-linear regression on yacht hydrodynamics dataset}
\label{subsec:nonlinear_regression}

Although the theoretical analysis in our main paper is primarily for learning and quantization linear regression models, we also perform experiments on non-linear regression models using neural networks.
We consider the yacht hydrodynamics dataset from \cite{yacht_hydrodynamics_dataset}, \cite{yacht_NN} from the UCI machine learning repository \cite{Dua:2019}.
This dataset predicts the residuary resistance of sailing yachts from input features representing basic hull dimensions and the boat velocity, such as such as beam-draught ratio, length-displacement ratio, prismatic coefficient, etc.
More information about this dataset can be found at the \href{http://archive.ics.uci.edu/ml/datasets/Yacht+Hydrodynamics}{UCI ML repository link}.

We use a $4$-layer neural network from \cite{hse} for this non-linear regression task.
The first, second and third hidden layers consist of, respectively, $128$, $32$, and $8$ neurons with ReLU activation function.
The output layer has a single output neuron with linear activation.
The total number of neural network weights (parameters) being quantized is $5297$.
We do a layer-wise quantization and treat the weights and biases of each layer separately.
To quantize the weights, we first vectorize the weight matrix of any specific layer, and then use randomized Hadamard-based \textbf{NDQ}.
Since the total number of parameters is not necessarily a power of $2$, in order to be able to construct Hadamard matrices, we cluster the weight vector into sub-vectors, each of which is a largest possible power of $2$.

We compare the performance of different quantization schemes using \textit{coefficient of determination}, denoted as $r^2$ as our evaluation metric.
The $r^2$-score is the proportion of variation in the response that is predictable from the features.
It is defined as follows.
For a given feature matrix $\Xv \in \Real^{n \times d}$ with rows $\{\xv_1, \ldots, \xv_n\}$, suppose a given model $f$ predicts values $\{f_1, \ldots, f_n\}$, where $f_i \triangleq f(\xv_i)$.
Given ground-truth responses $\{y_1, \ldots, y_n\}$, the \textit{residual sum of squares} is defined as $\mathrm{SS}_{\rm res} = \sum_{i = 1}^{n}(y_i - f_i)^2$, and the \textit{total sum of squares} is defined as, $\mathrm{SS}_{\rm tot} = \sum_{i=1}^{n}(y_i - \yo)^2$, where $\yo = \frac{1}{n}\sum_{i=1}^{n}y_i$ is the mean response.
The $r^2$-score is then defined as,

\begin{equation}
    r^2 \triangleq 1 - \frac{\rm SS_{res}}{\rm SS_{tot}}
\end{equation}

Note that in the best case, the response as predicted by the model $f$ may exactly match the ground-truth responses, in which case, $\rm SS_{res} = 0$, and $r^2 = 1$.
On the other hand, a trivial baseline model, which always predicts $\yo$, will have $r^2 = 0$.
Regression models that have worse predictions than this baseline will have negative $r^2$-score.

In Figs. \ref{fig:r2_training} and \ref{fig:r2_validation} below, we plot the the $r^2$-score versus the bit-budget used for quantizing the weights of the neural network on \textit{training} and \textit{validation} sets, respectively.
We note that our proposed randomized Hadamard-based \textbf{NDQ} strategy outperforms the na\"ive quantization scheme.

\begin{figure}[h!]
    \centering
    \begin{minipage}{0.5\textwidth}
        \centering
        \begin{subfigure}[b]{\textwidth}
           \includegraphics[width=1\linewidth]{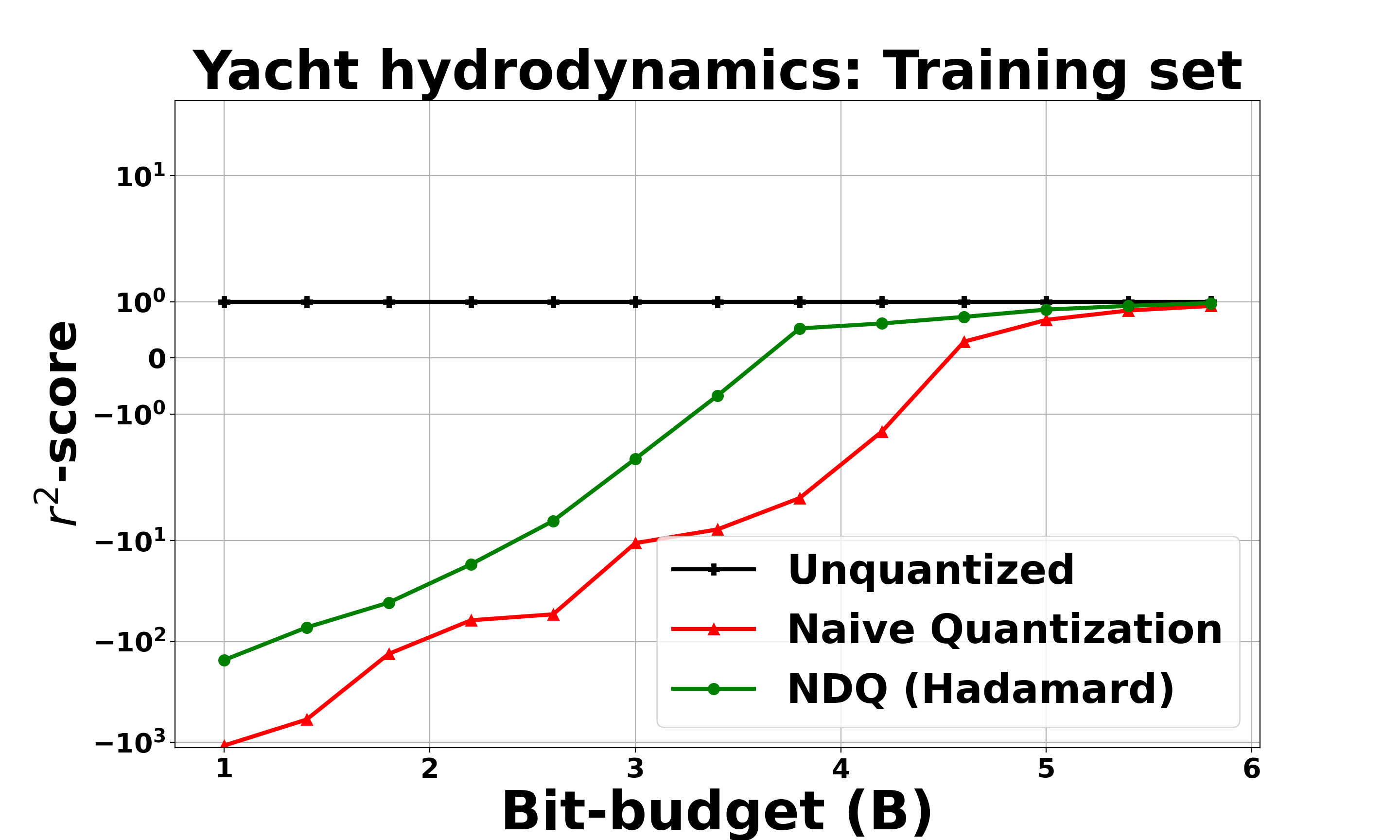}
           \caption{}
           \label{fig:r2_training} 
        \end{subfigure}
    \end{minipage}%
    \begin{minipage}{0.5\textwidth}
        \centering
        \begin{subfigure}[b]{\textwidth}
           \includegraphics[width=1\linewidth]{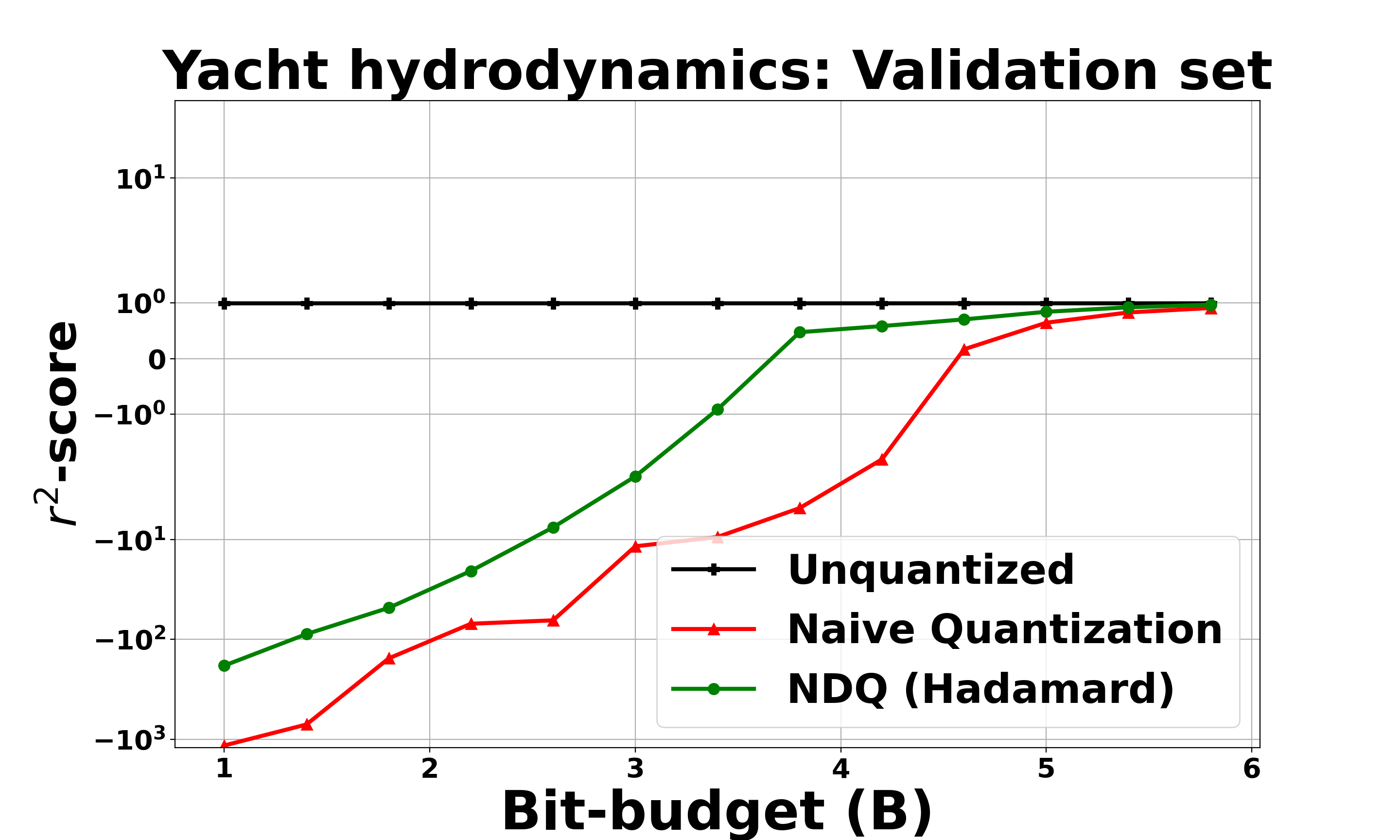}
           \caption{}
           \label{fig:r2_validation} 
        \end{subfigure}
    \end{minipage}
    \vspace{6mm}
    \caption{$r^2$-score vs. bit-budget plots for yacht hydrodynamics dataset}
    \vspace{4mm}
\end{figure}

We also visualize the effect of model quantization by obtaining the scatter plots of the predicted versus ground-truth values in Figs. \ref{fig:hydrodynamics_scatter_plots_training} and \ref{fig:hydrodynamics_scatter_plots_validation}.
The Y-axis is the \textit{predicted response}, while the X-axis is the \textit{ground-truth response}.
Ideally, for a good model, the scatter plot should be concentrated on the line $Y = X$.
We obtain these plots for three different bit budgets $B = 4$, $B = 4.5$, and $B = 5$.
In contrast to an apparent prediction trend for \textbf{NDQ}, the points in the scatter plot for na\"ive quantization are all over the plot.
As we increase $B$, we can see that, \textbf{NDQ}-Hadamard concentrates more prominently about $Y=X$ than the na\"ive strategy.

\begin{figure}[h!]
\vspace{2mm}
\begin{tabular}{ccc}
    \includegraphics[width=0.315\linewidth]{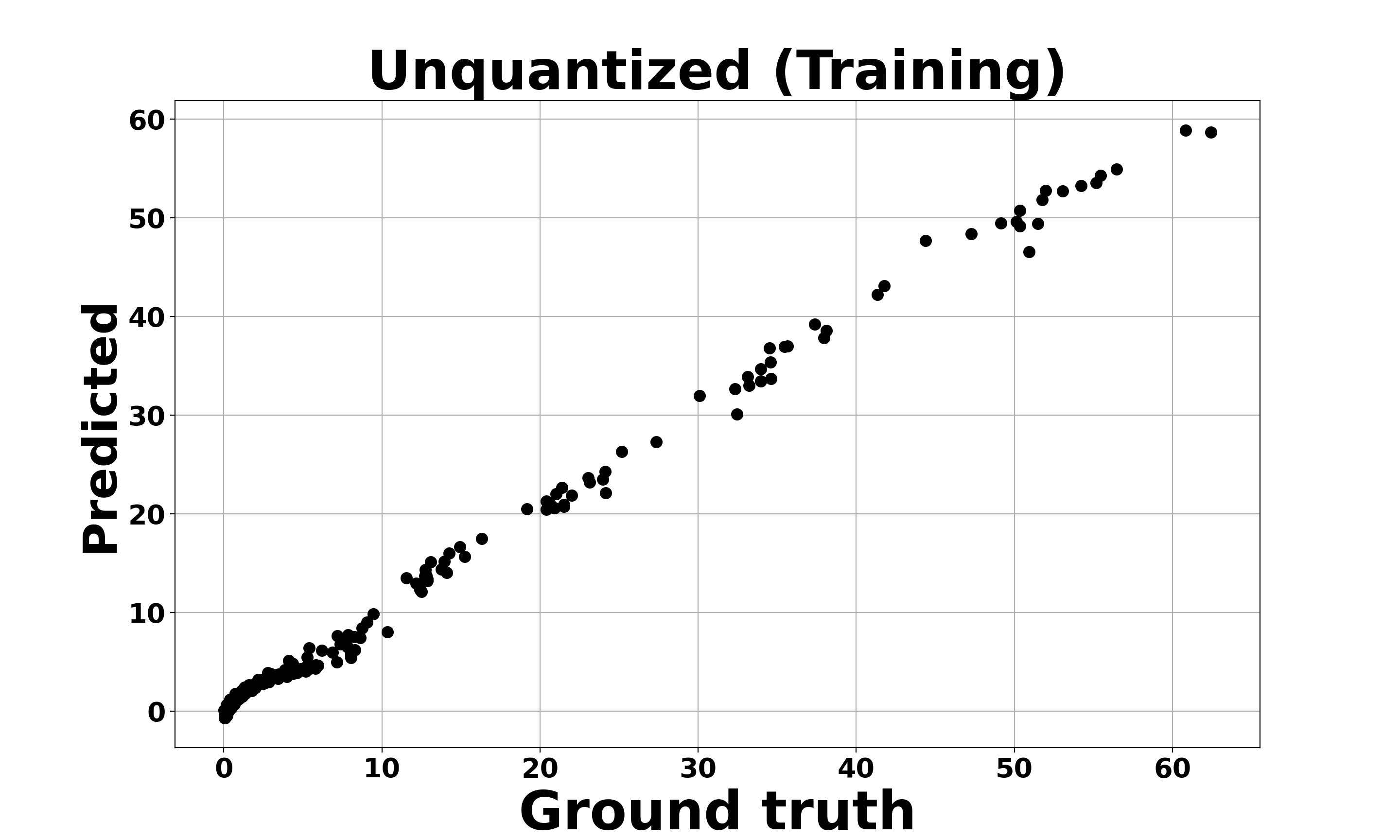} &
    \includegraphics[width=0.315\linewidth]{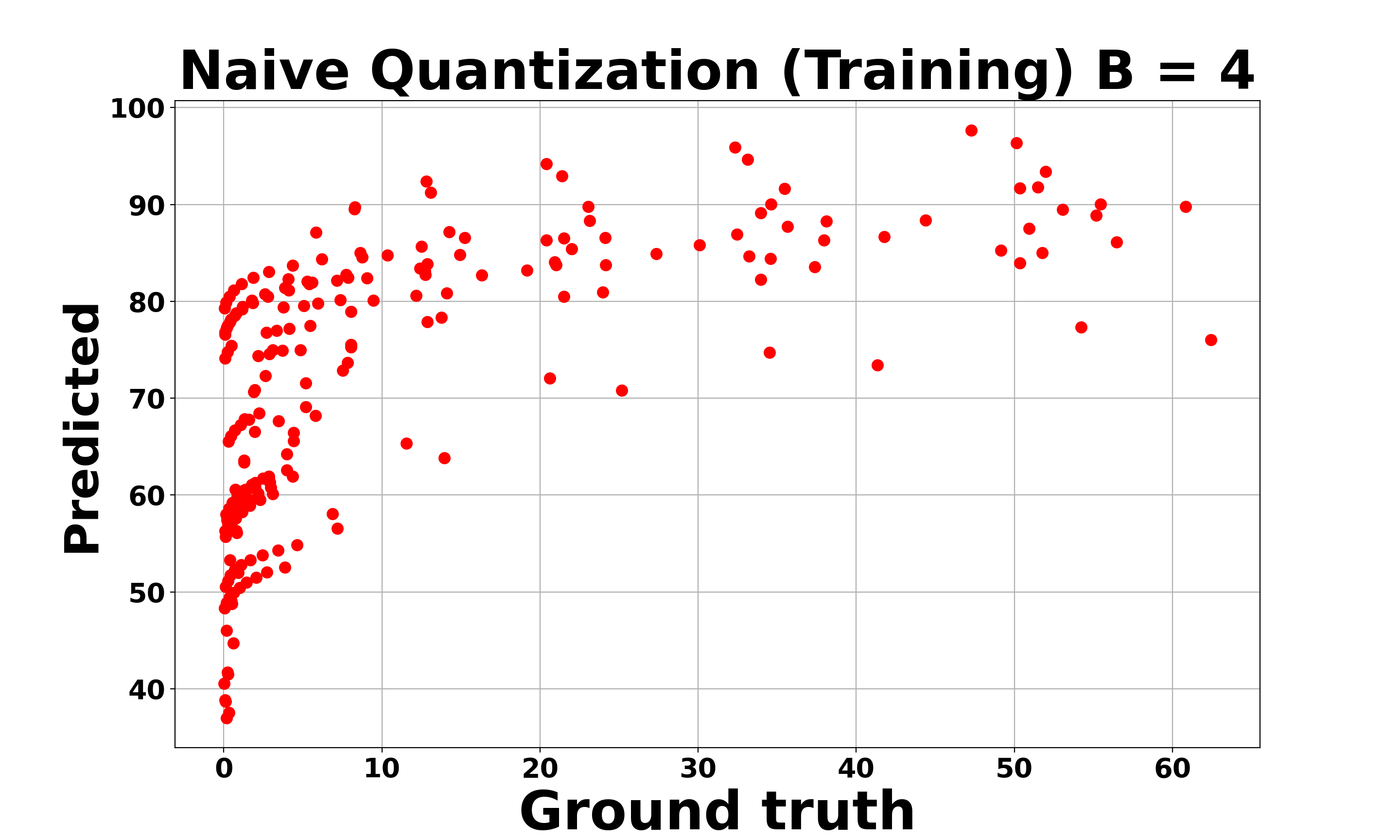} &
    \includegraphics[width=0.315\linewidth]{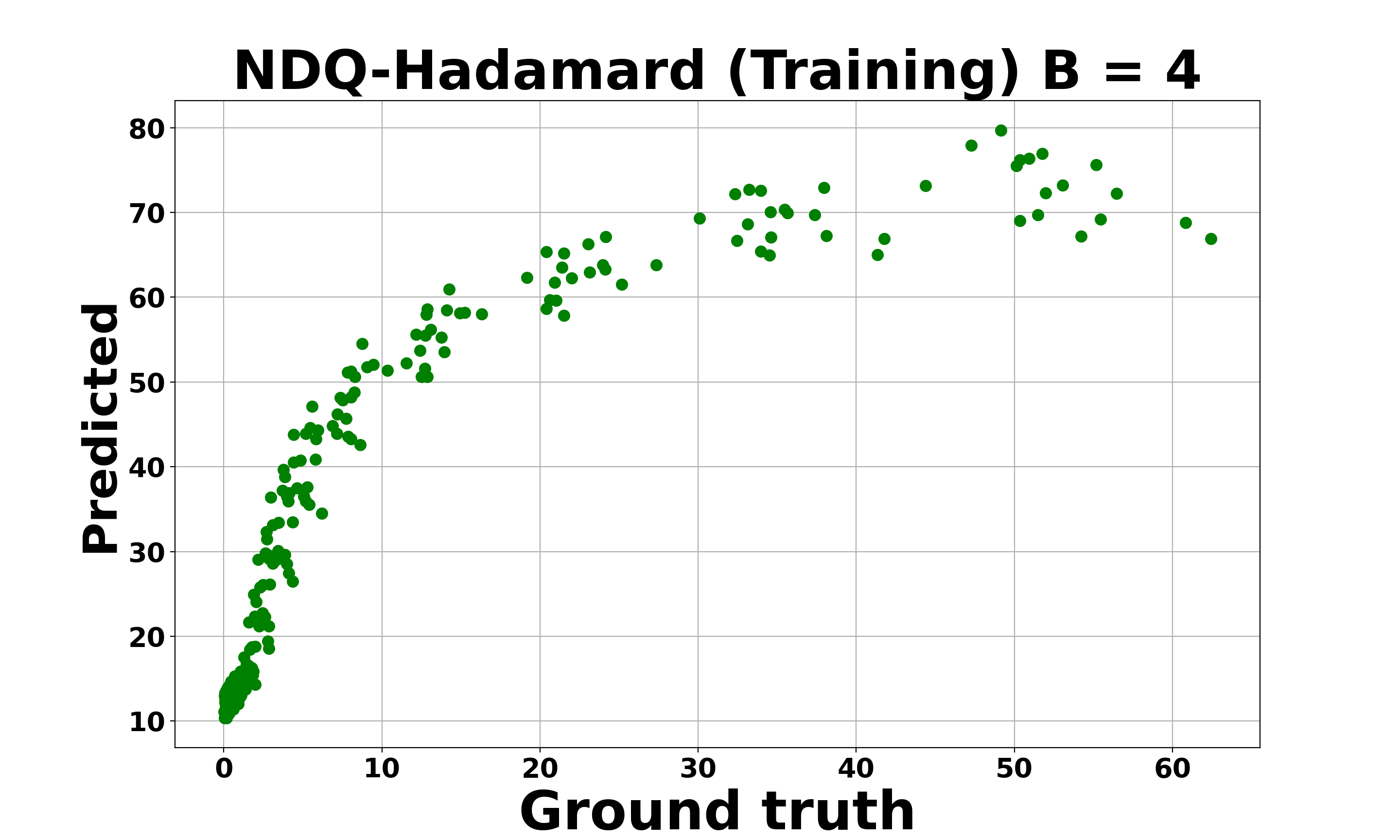}\\
    \includegraphics[width=0.315\linewidth]{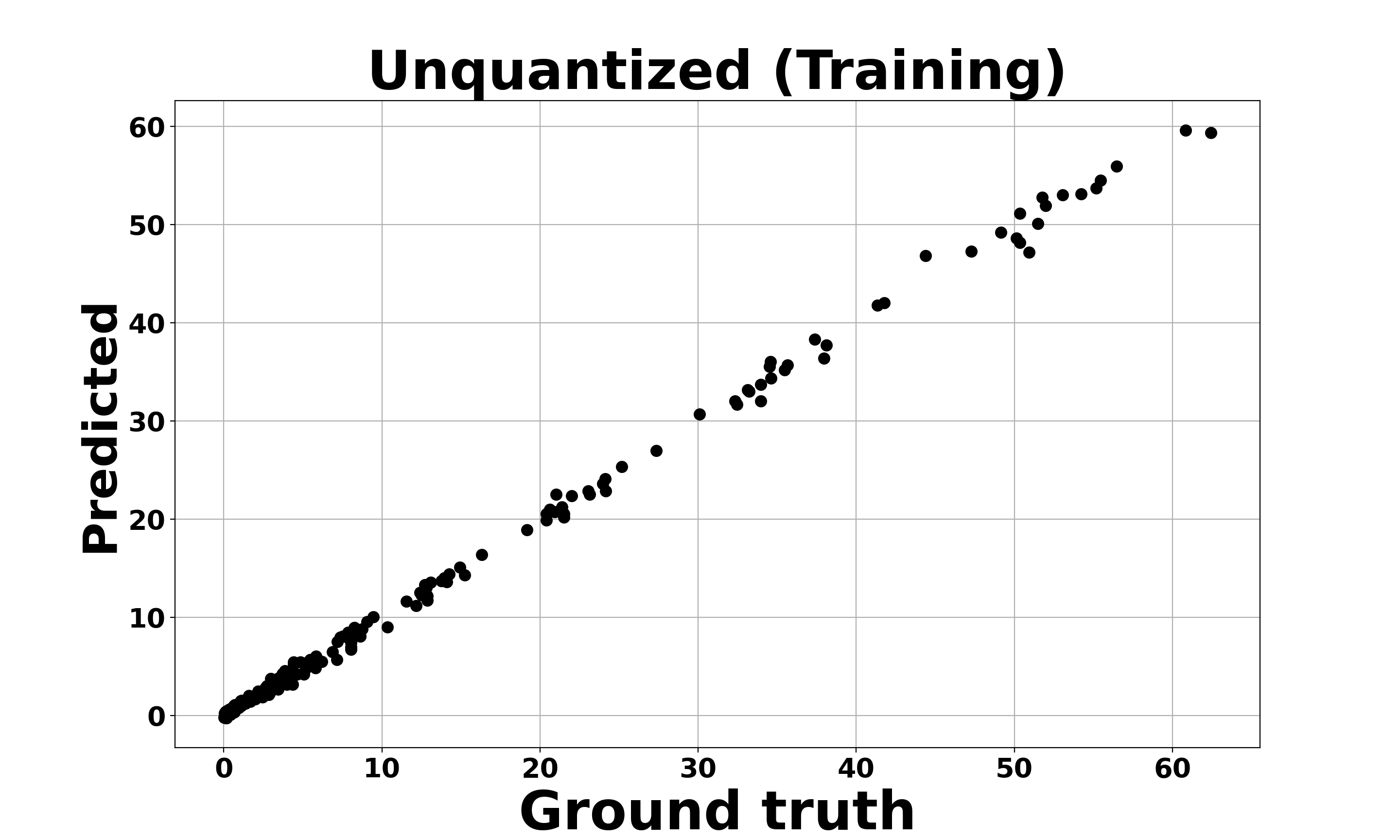} &
    \includegraphics[width=0.315\linewidth]{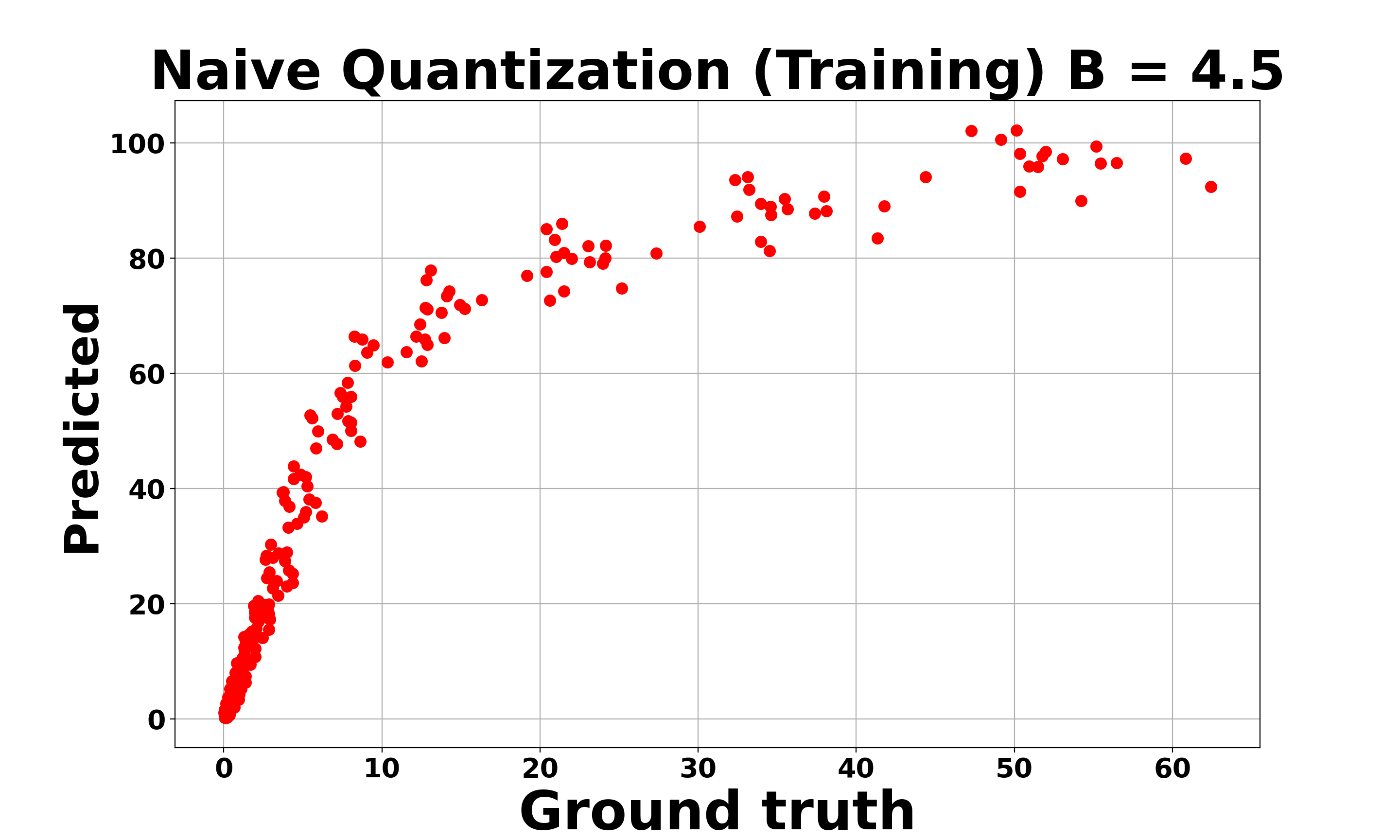} &
    \includegraphics[width=0.315\linewidth]{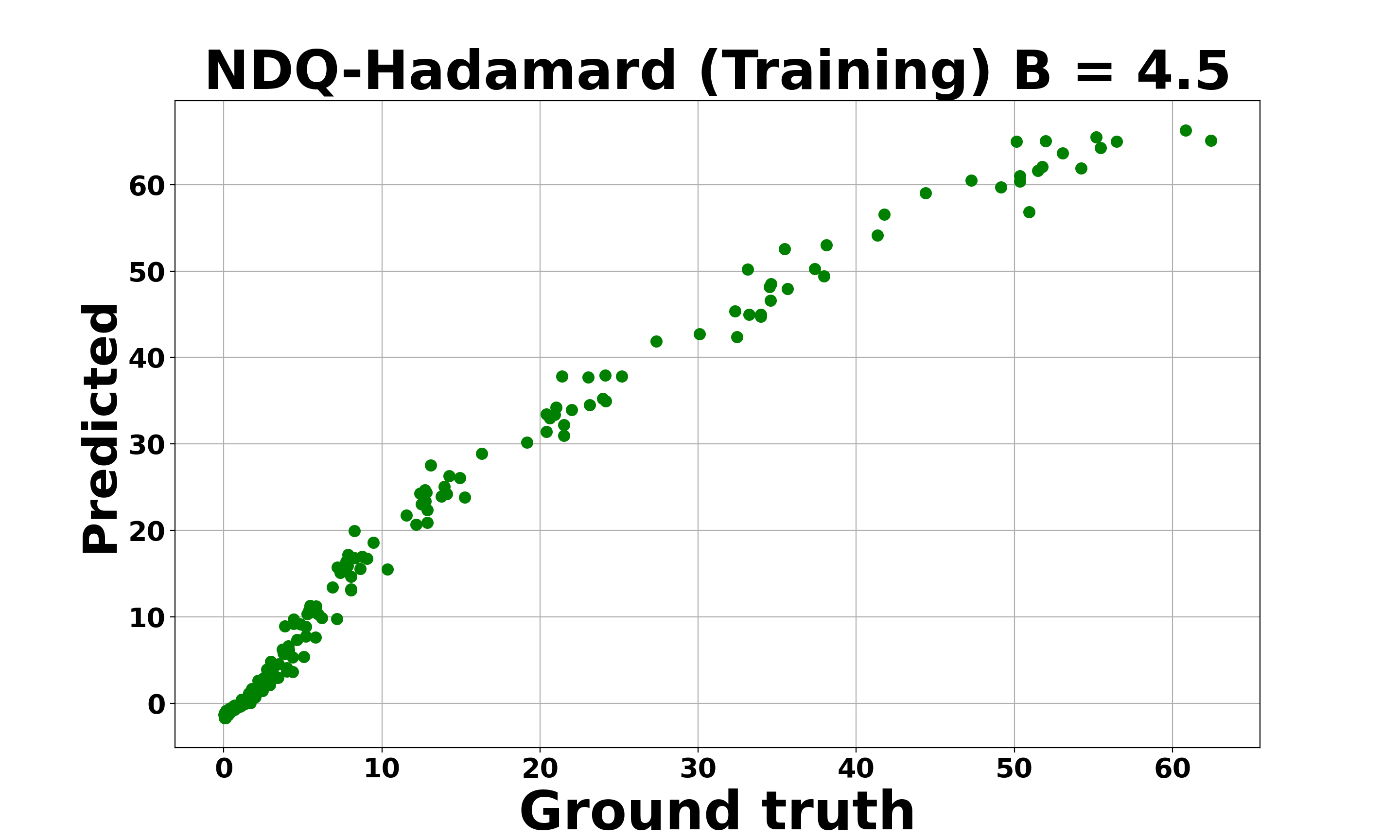} \\
    \includegraphics[width=0.315\linewidth]{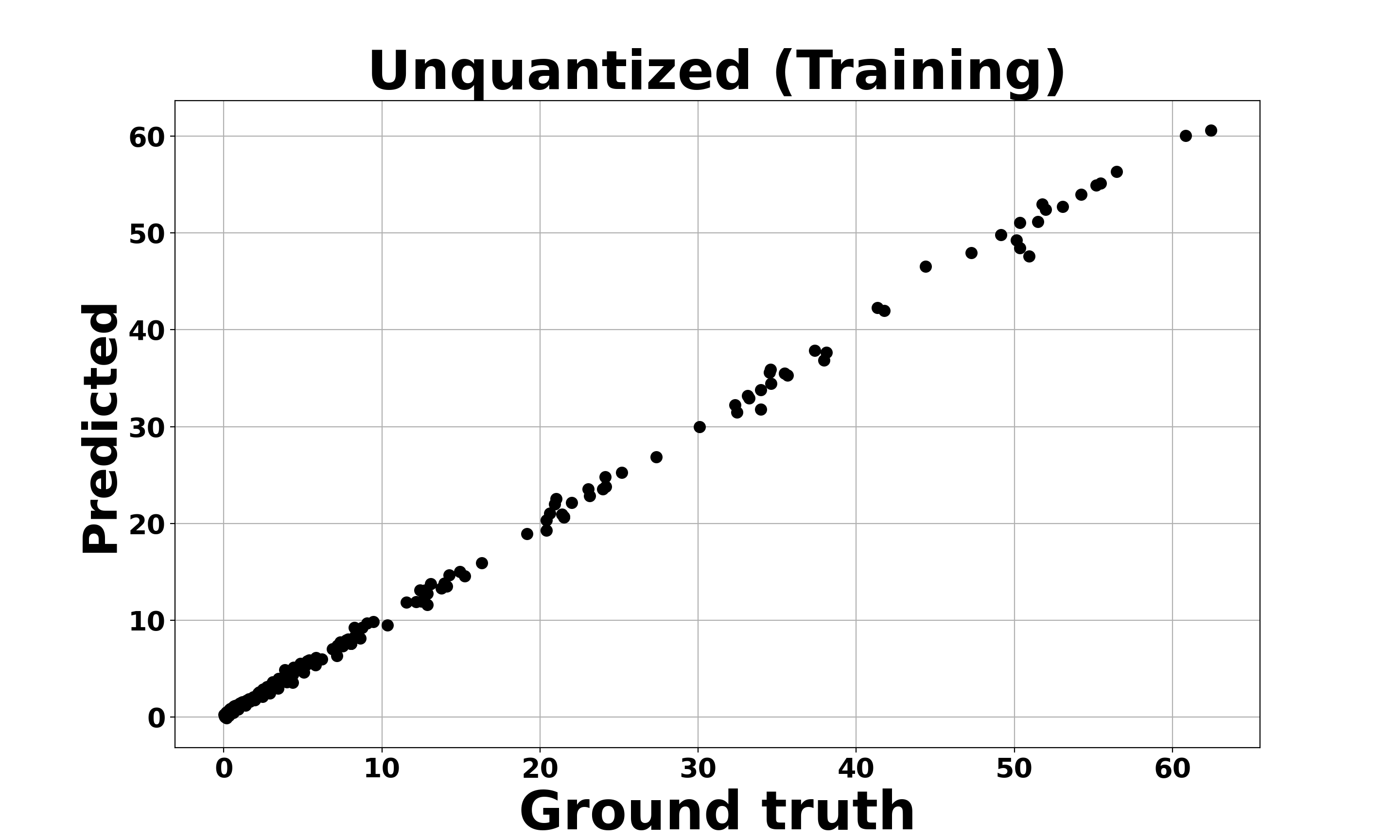} &   \includegraphics[width=0.315\linewidth]{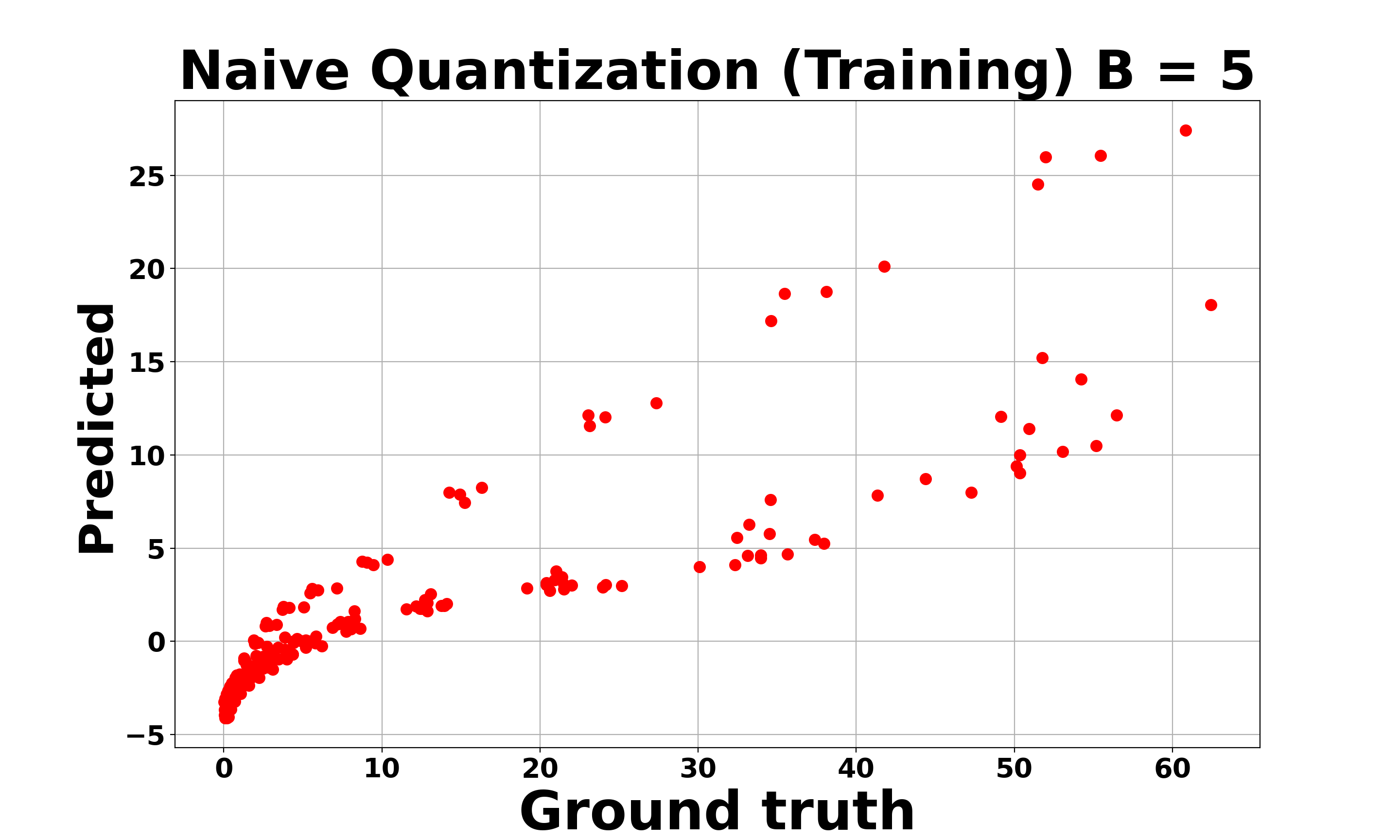} &
    \includegraphics[width=0.315\linewidth]{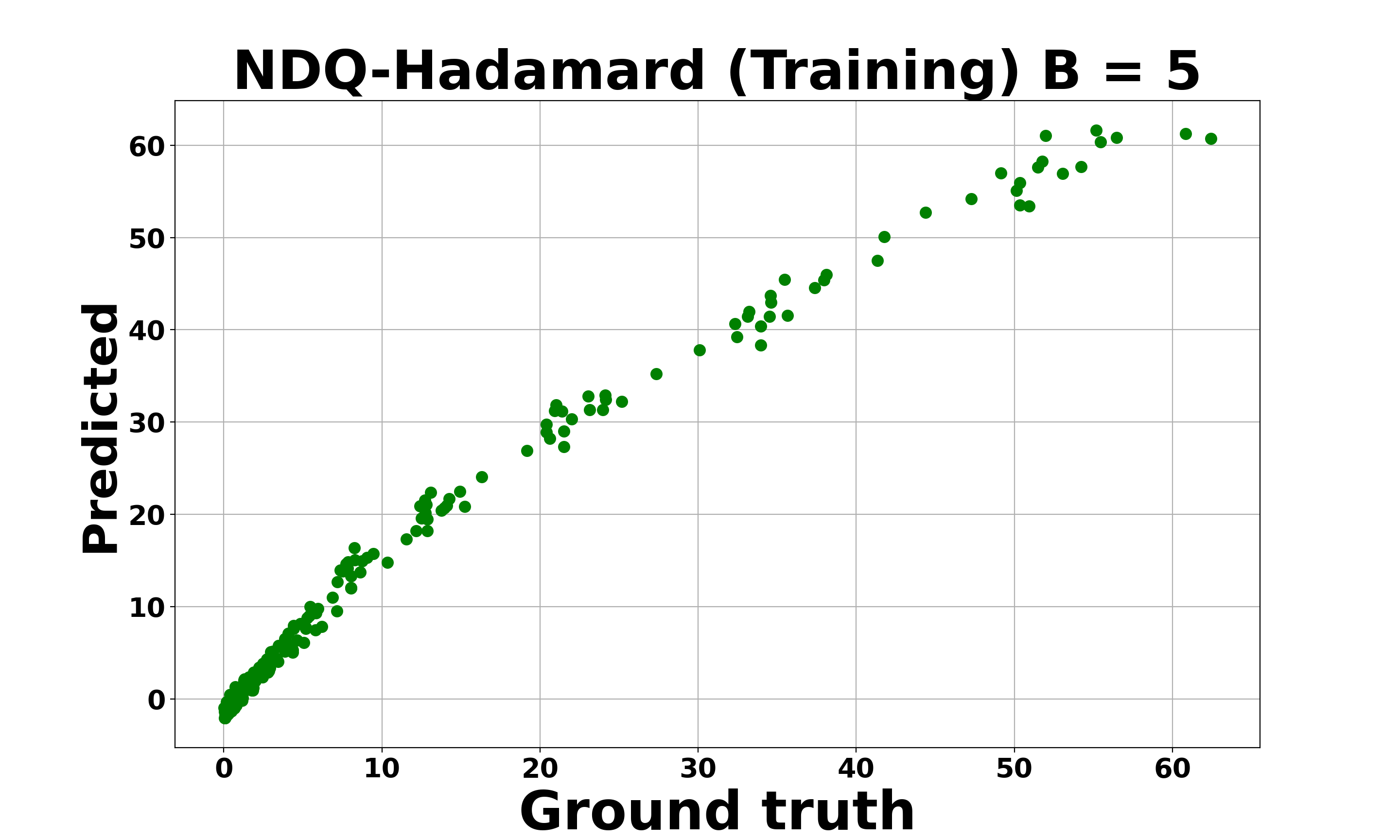}\\
\end{tabular}
\caption{Scatter plots for \textbf{training} data set}
\label{fig:hydrodynamics_scatter_plots_training}
\end{figure}

\begin{figure}[h!]
\vspace{2mm}
\begin{tabular}{ccc}
    \includegraphics[width=0.315\linewidth]{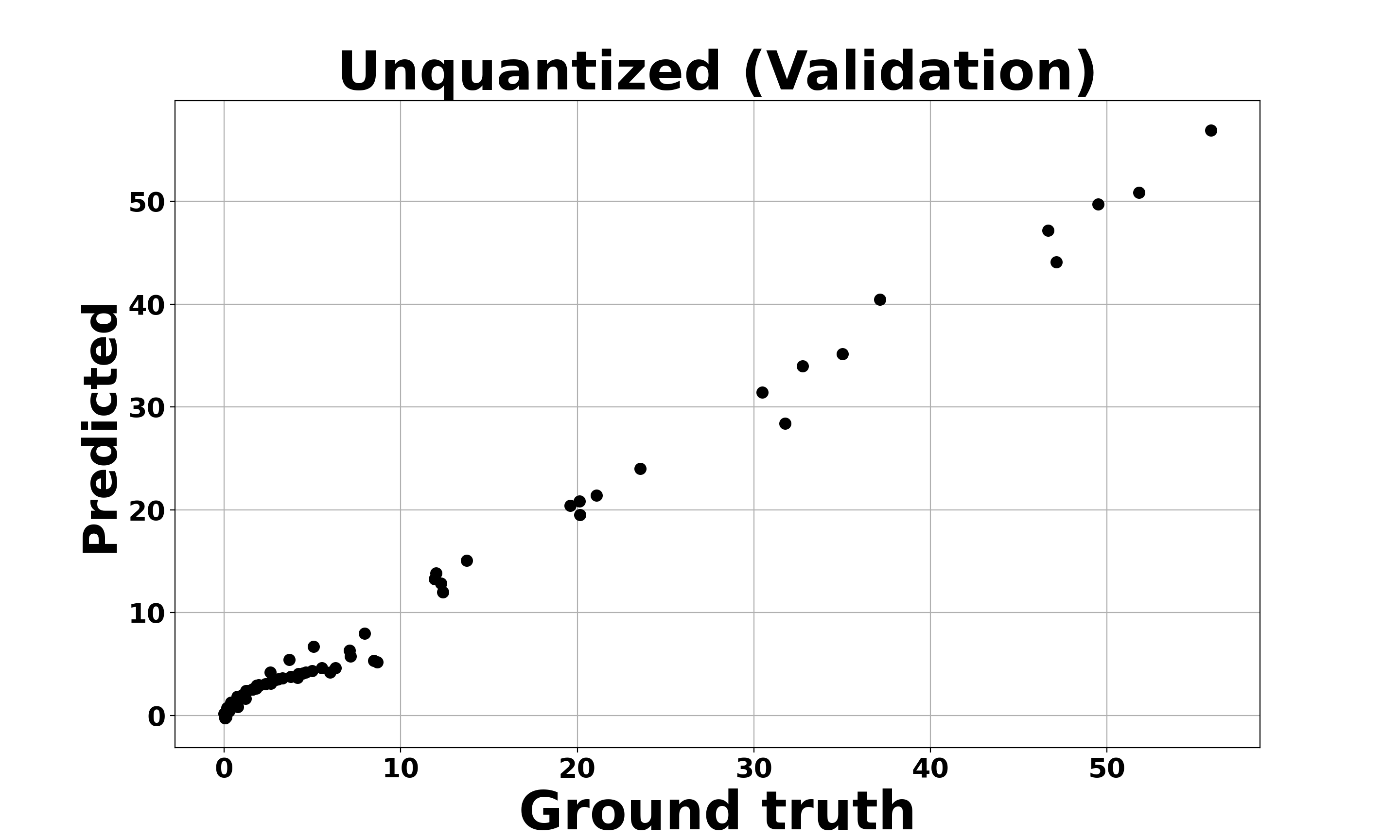} &
    \includegraphics[width=0.315\linewidth]{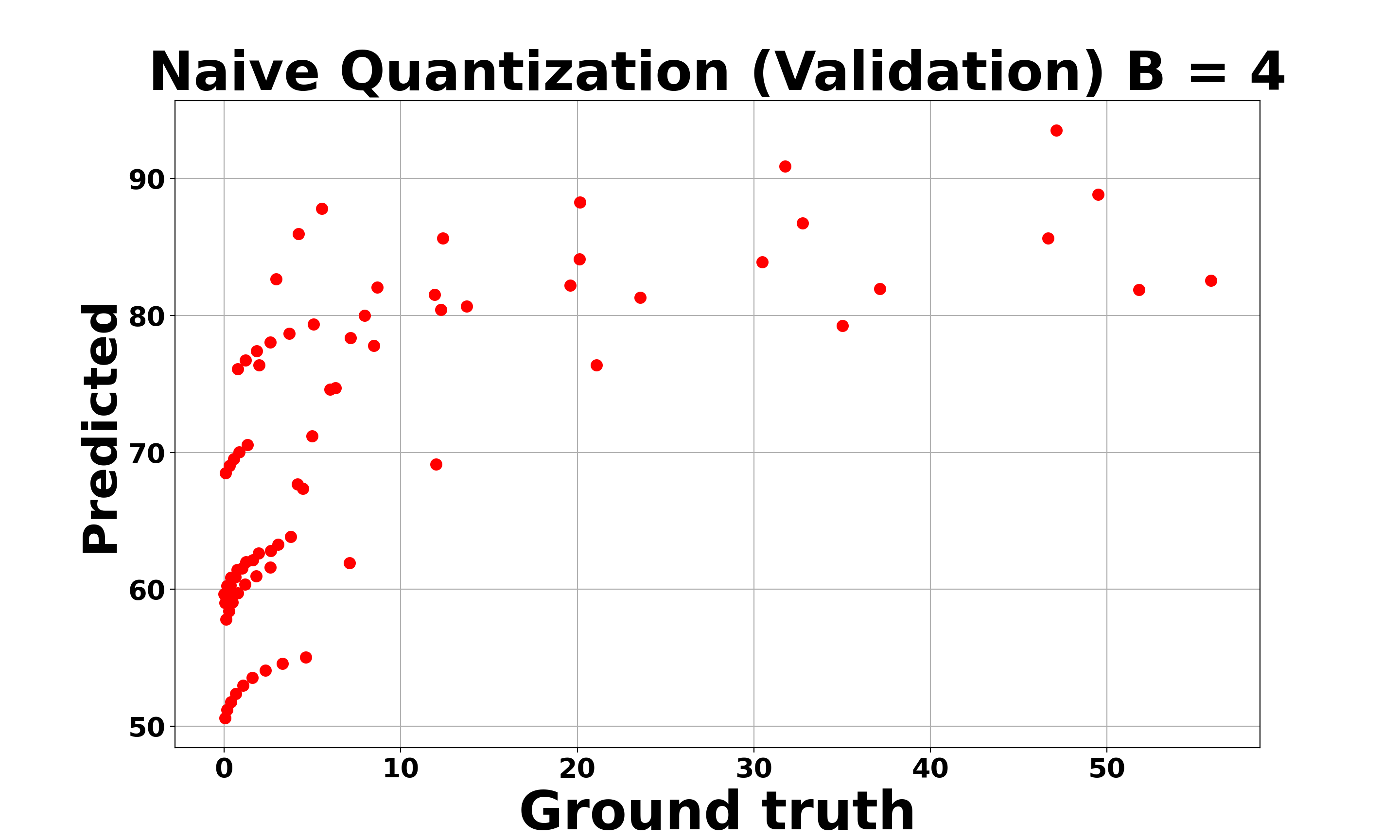} &
    \includegraphics[width=0.315\linewidth]{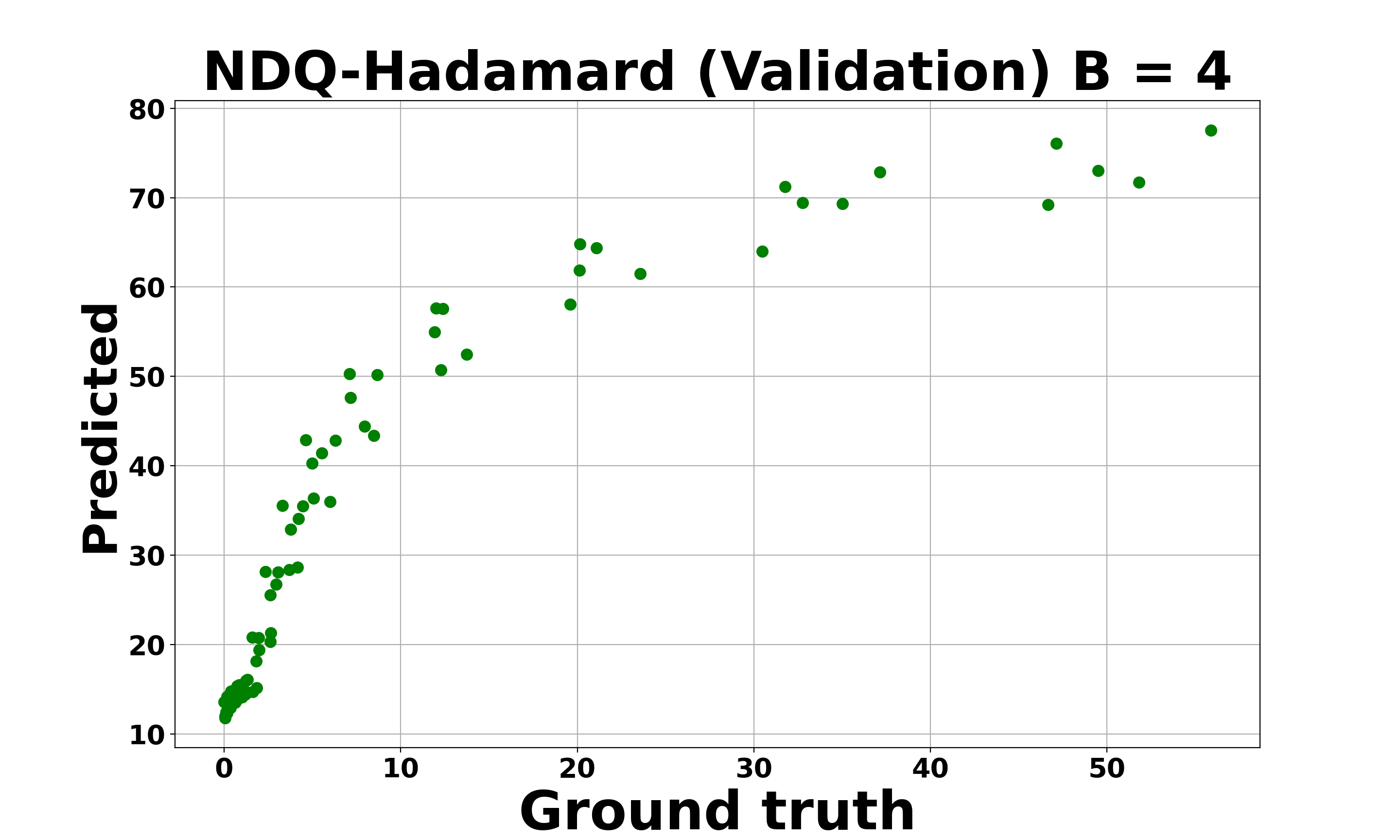}\\
    \includegraphics[width=0.315\linewidth]{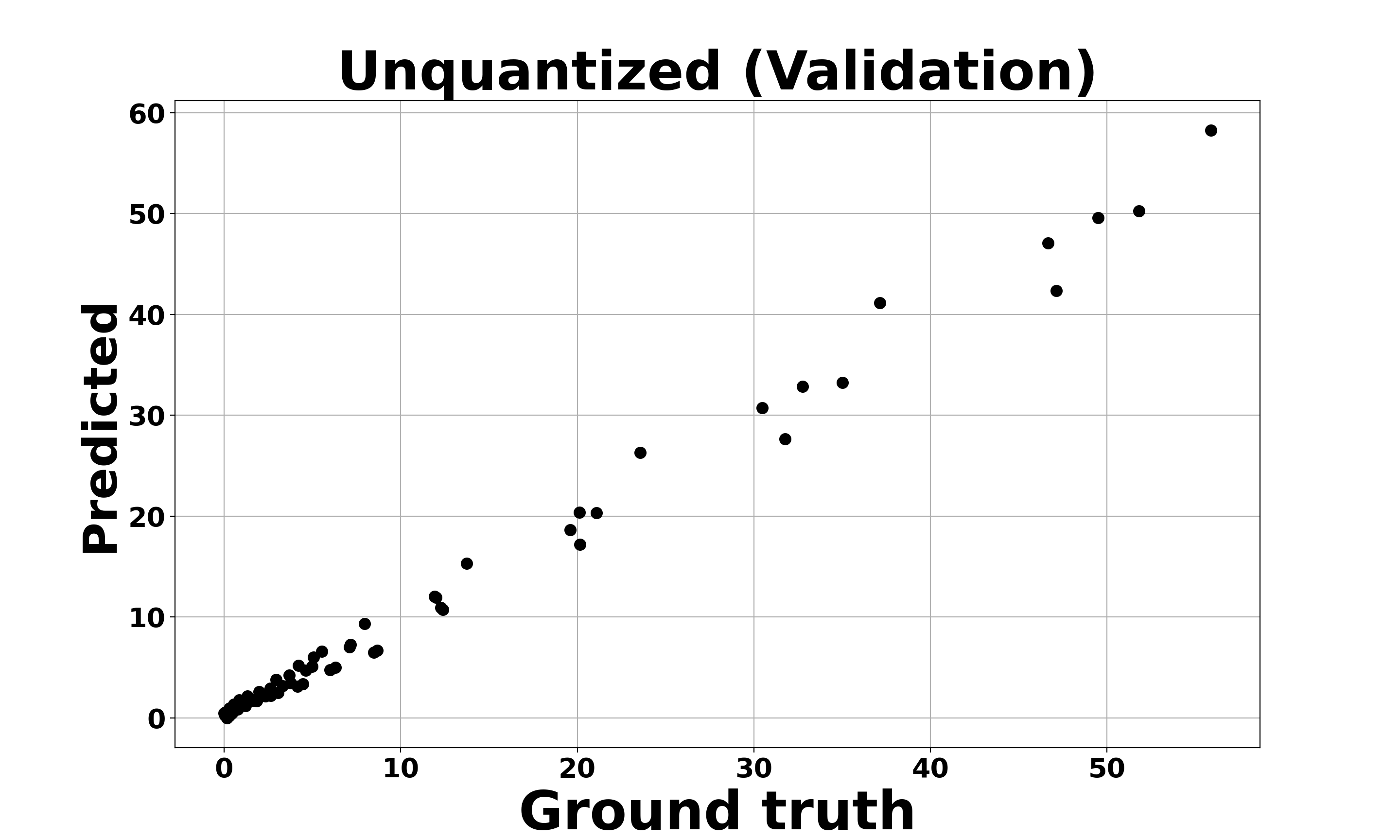} &
    \includegraphics[width=0.315\linewidth]{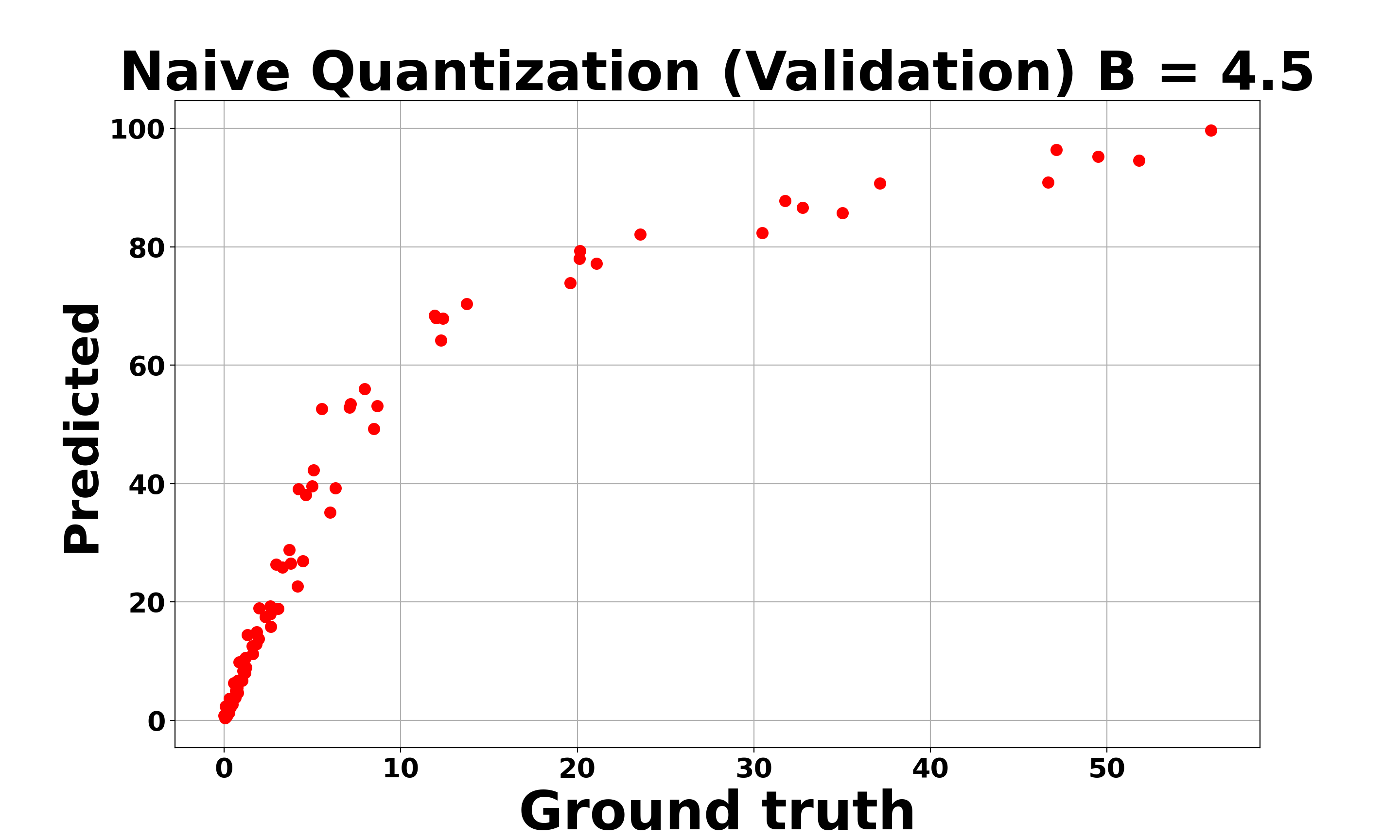} &
    \includegraphics[width=0.315\linewidth]{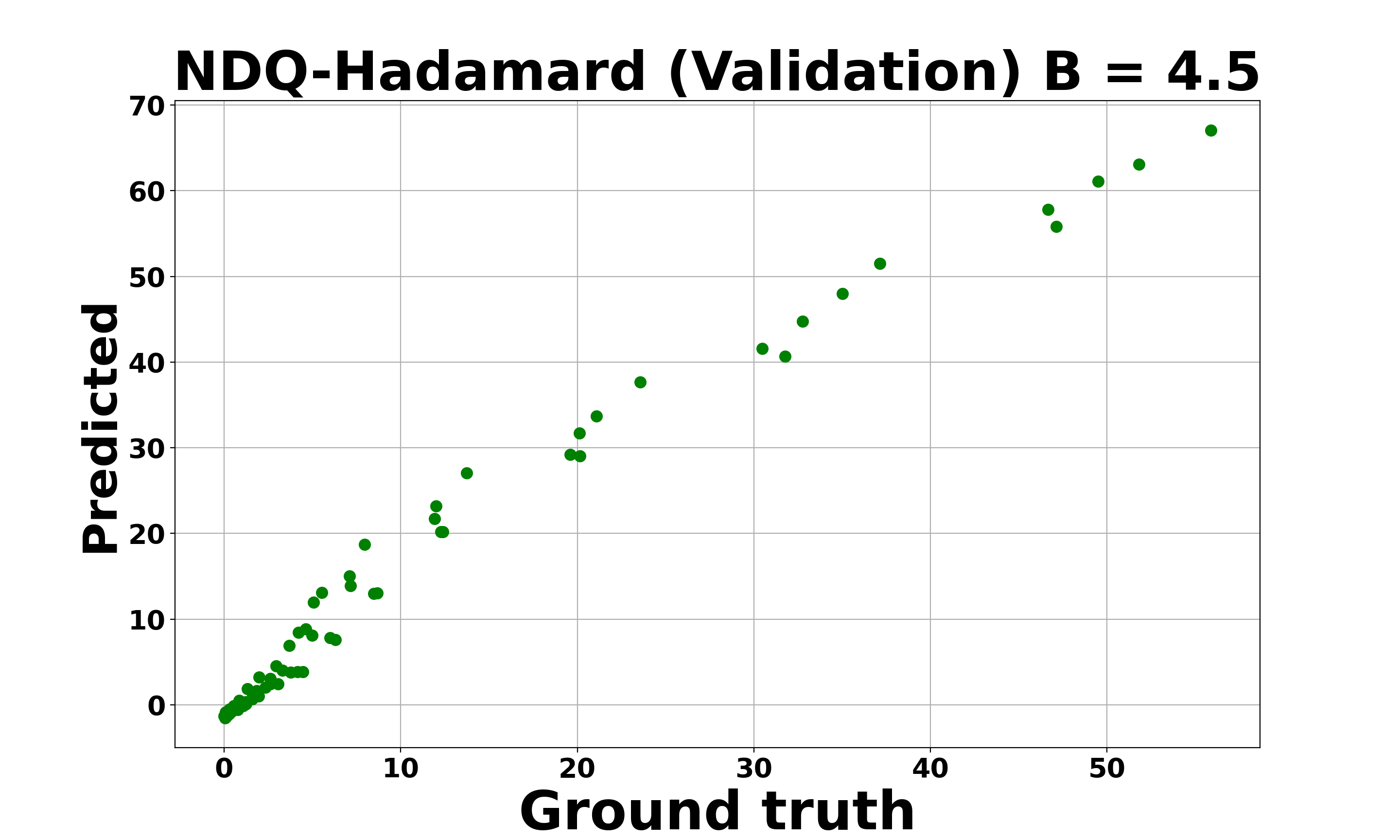} \\
    \includegraphics[width=0.315\linewidth]{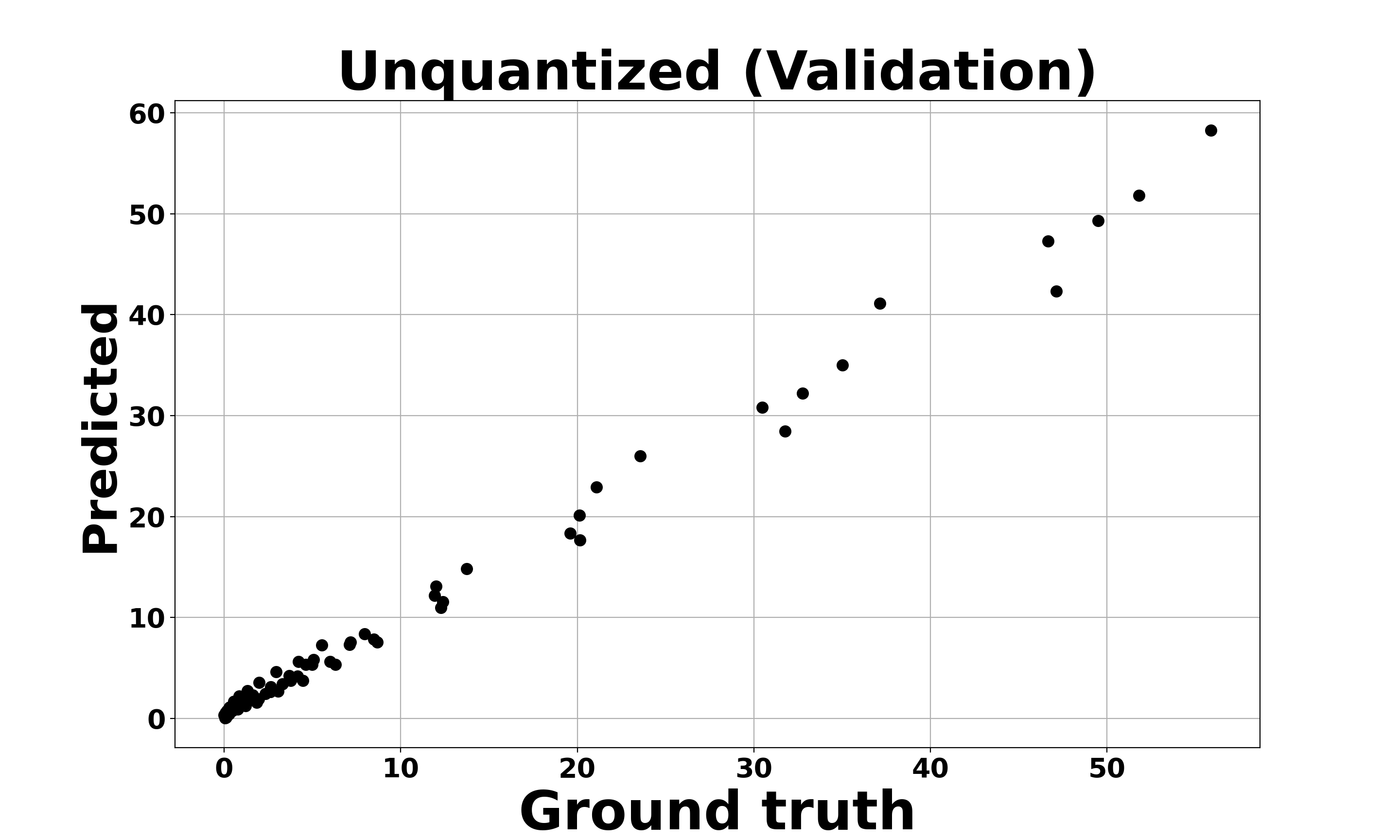} &   \includegraphics[width=0.315\linewidth]{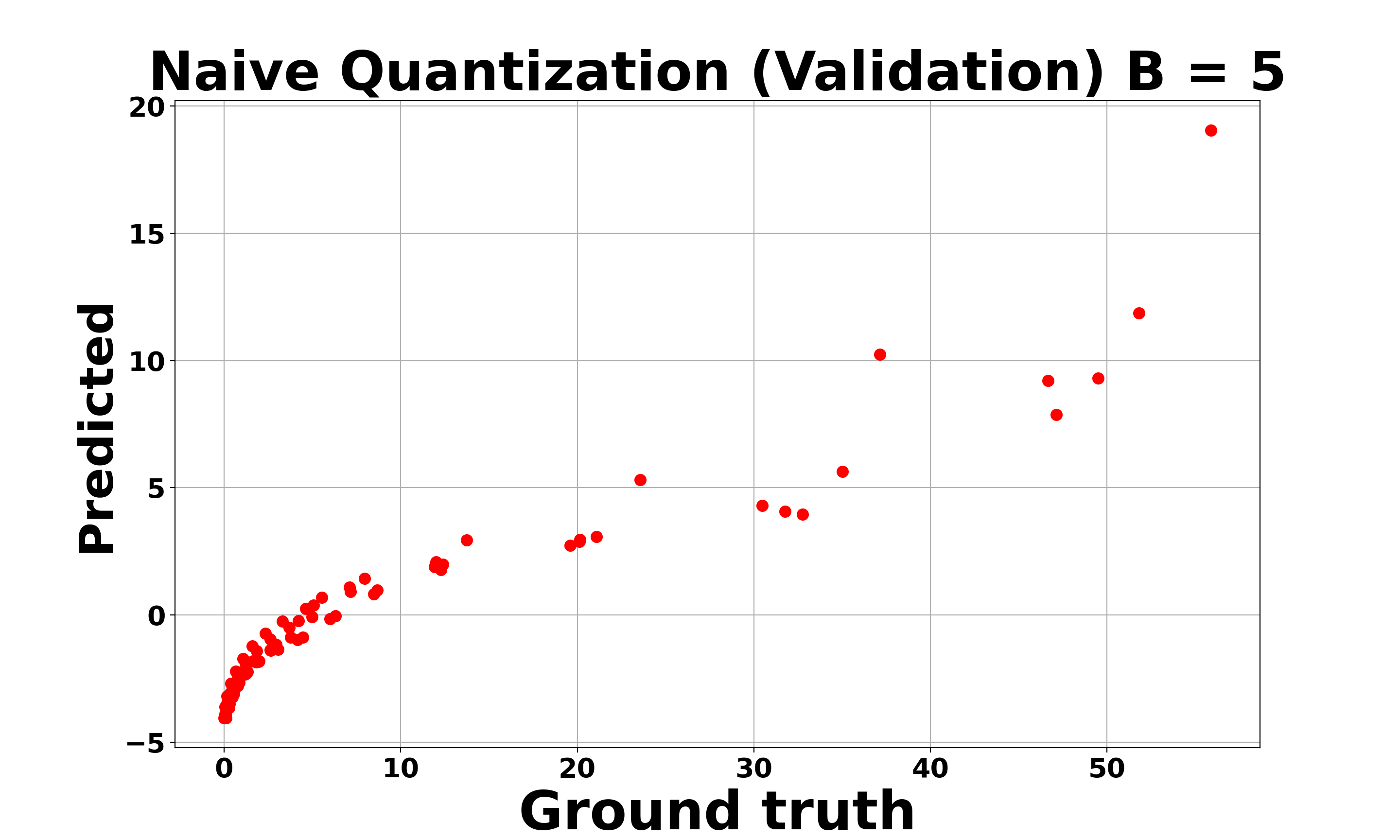} &
    \includegraphics[width=0.315\linewidth]{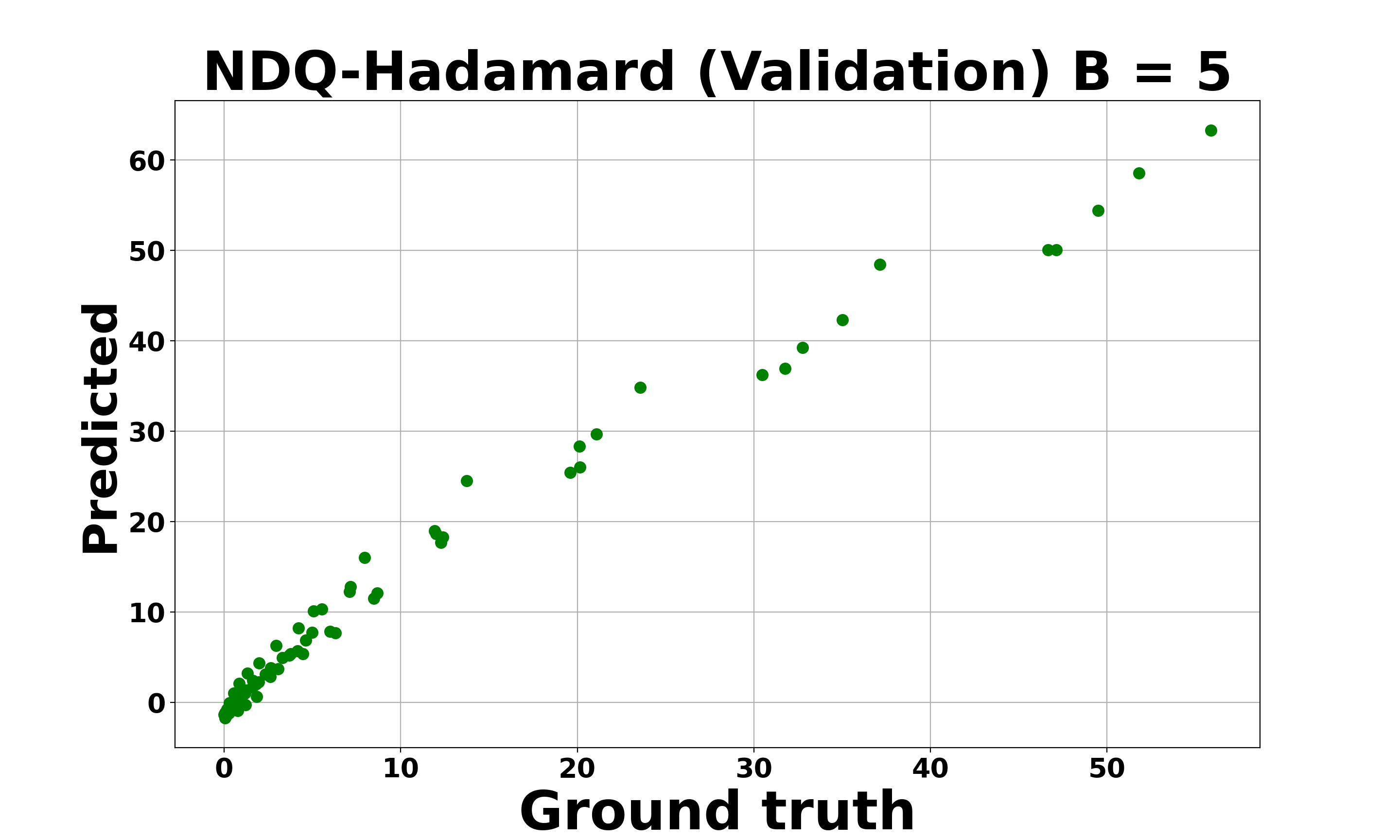}\\
\end{tabular}
\caption{Scatter plots for \textbf{validation} data set}
\label{fig:hydrodynamics_scatter_plots_validation}
\end{figure}

\section{Proof of Lemma \ref{lem:mutual_information_bit_budget_inequality}: An inequality between the mutual information \texorpdfstring{$I(\yv; \Qsf(\Xv, \yv))$}{I()} and the bit-budget \texorpdfstring{$B$}{B}}
\label{app:proof_mutual_information_bit_budget_inequality}

Given training data $(\Xv, \yv)$, $\thetatv \triangleq \Qsf(\Xv, \yv)$ is the output of the learning code $\Qsf \in \Qcal_{n,d,B}$ and serves as an estimate of the latent parameter $\thetav$ that generated the training data according to \eqref{eq:noisy_planted_model}.
For any such learning code $\Qsf$, since $\Qsf(\Xv, \yv)$ takes at most $2^{dB}$ values,
\begin{equation}
    I(\yv;\Qsf(\Xv, \yv)) = H(\Qsf(\yv)) - H(\Qsf(\Xv, \yv) \vert \yv) \leq H(\Qsf(\Xv, \yv)) \leq dB.
\end{equation}

The last inequality holds true since entropy of a discrete random variable is maximized when it is distributed uniformly over the domain.

\section{Proof of Lemma \ref{lem:inequality_supremum_bayes_error}: Worst-case risk over the parameter space \texorpdfstring{$\Thetav$}{Thetav} is asymptotically lower bounded by the Bayes risk over \texorpdfstring{$\Real^d$}{Reald} for Gaussian prior}
\label{app:proof_inequality_supremum_bayes_error}

The proof of this lemma parallels the arguments in the proof of Thm. 3.1 in \cite{zhu_neurips_2014}.
Let $\Thetaov := \Real^d\setminus\Thetav$.
We have,
\begin{align}
\label{eq:splitting_bayes error}
    \int_{\thetav \in \Real^d} R(\Qsf(\Xv, \yv), \thetav) dF(\thetav) &\leq \int_{\thetav \in \Thetav}R(\Qsf(\Xv, \yv), \thetav) dF(\thetav) + \int_{\thetav \in \Thetaov}R(\Qsf(\Xv, \yv), \thetav) dF(\thetav) \nonumber \\
    &\leq \sup_{\thetav \in \Thetav} R(\Qsf(\Xv, \yv), \thetav) + \int_{\thetav \in \Thetaov}R(\Qsf(\Xv, \yv), \thetav) dF(\thetav)
\end{align}

It then suffices to show that $\int_{\thetav \in \Thetaov}R(\Qsf(\Xv,\yv), \thetav) dF(\thetav) \to 0$ as $d \to \infty$.
For a fixed $\delta \in (0,1)$, define the distribution $F_{\delta} \equiv \Ncal(\mathbf{0}, c^2\delta^2\Iv_d)$.
Suppose $\thetav \sim F_{\delta}$.
We then have,

\begin{align}
\label{eq:upper_bound_bayesian_error_over_theta_complement}
    \int_{\thetav \in \Thetaov}R(\Qsf(\Xv,\yv), \thetav)dF_{\delta}(\thetav) &= \frac{1}{d}\int_{\thetav \in \Thetaov}\mathbb{E}_{\yv}\Br{\norm{\Qsf(\Xv,\yv) - \thetav}_2^2}dF_{\delta}(\thetav) \nonumber \\
    &\leq \frac{2}{d}\int_{\thetav \in \Thetaov}\norm{\thetav}_2^2 dF_{\delta}(\thetav) + \frac{2}{d}\int_{\thetav \in \Thetaov}\mathbb{E}_{\yv}\Br{\norm{\Qsf(\Xv,\yv)}_2^2}dF_{\delta}(\thetav) \nonumber \\
    &\leq \underbrace{\frac{2}{d}\sqrt{\int_{\thetav \in \Thetaov}dF_{\delta}(\thetav)}\sqrt{\int_{\thetav \in \Real^d}\norm{\thetav}_2^4dF_{\delta}(\thetav)}}_{T_1} + \underbrace{2c^2\int_{\thetav \in \Thetaov}dF_{\delta}(\thetav)}_{T_2},
\end{align}

where the last inequality follows from Cauchy-Schwarz inequality for probability measures and the fact that $\Qsf:\Real^{n \times d} \times \Real^n \to \Thetav$, implying $\frac{1}{d}\lVert \Qsf(\Xv, \yv)\rVert_2^2 \leq c^2$.
To upper bound $T_2$, note that,

\begin{align*}
    \int_{\thetav \in \Thetaov}dF_{\delta}(\thetav) &= \mathbb{P}_{\thetav \sim F_{\delta}}\Br{\frac{1}{d}\sum_{i=1}^{d}\theta_i^2 > c^2}\\
    &= \mathbb{P}_{\thetav \sim F_{\delta}}\Br{\frac{1}{d}\sum_{i=1}^{d}\Paren{\Paren{\frac{\theta_i}{\delta c}}^2 - 1} > \frac{1 - \delta^2}{\delta^2}} \stackrel{\mathrm{(i)}}{\leq} 2\exp\Paren{-\frac{d(1 - \delta^2)^2}{8\delta^4}}.
\end{align*}

Here, $\mathrm{(i)}$ follows from the concentration inequality: If $Z_1, \ldots, Z_d \sim \Ncal(0,1)$ and $t \in (0,1)$, then $\mathbb{P}\Paren{\left\lvert \frac{1}{d}\sum_{i=1}^{d}(Z_i^2 - 1) \right\rvert > t} \leq 2e^{-dt^2/8}$.
On the other hand, to bound $T_1$, note that,

\begin{align*}
     \int_{\thetav \in \Real^d}\norm{\thetav}_2^4dF_{\delta}(\thetav) = \mathbb{E}_{\thetav \sim F_{\delta}}\Br{\norm{\thetav}_2^4} &= \sum_{i=1}^{d}\mathbb{E}_{F_{\delta}}\Br{\theta_i^4} + \sum_{i \neq j}\mathbb{E}_{F_{\delta}}\Br{\theta_i^2} \cdot \mathbb{E}_{F_{\delta}}\Br{\theta_j^2} \\
     &= \mathbb{E}_{F_{\delta}}\Br{\theta_1^4}\cdot d + \binom{d}{2}c^4\delta^4 = O(d^2). 
\end{align*}

So, the R.H.S. of \eqref{eq:upper_bound_bayesian_error_over_theta_complement} can be upper bounded by,

\begin{equation*}
    \int_{\thetav \in \Thetaov}R(\Qsf(\Xv, \yv), \thetav) dF(\thetav) \leq \frac{2}{d}\sqrt{2}\exp\Paren{-\frac{d(1 - \delta^2)^2}{16\delta^4}}O(d) + 2c^2\exp\Paren{-\frac{d(1- \delta^2)^2}{8\delta^4}}
\end{equation*}

for any $\delta \in (0,1)$.
This implies $\int_{\thetav \in \Thetaov} R(\Qsf(\Xv, \yv), \thetav)dF(\thetav) \to 0$ as $d \to \infty$.

\section{Proof of Lemma \ref{lem:lower_bound_MI_Gaussian_prior}: Lower bound on mutual information \texorpdfstring{$I(\Qsf(\Xv, \yv), \thetav)$}{I()} for Gaussian prior over \texorpdfstring{$\thetav$}{thetav}}
\label{app:proof_lower_bound_MI_Gaussian_prior}

Given the training data $(\Xv, \yv)$, for a prior distribution $\thetav \sim F_{\delta} \equiv \Ncal(\mathbf{0}, c^2\delta^2\Iv_d)$, the Bayes optimal estimator is the maximum likelihood (ML) estimator denoted as $\thetahv$.
Since $\yv = \Xv\thetav + \vv$ where $\vv \sim \Ncal(\mathbf{0}, \sigma^2\Iv_n)$, it can be easily shown \cite{rosenberg} that the posterior distribution of $\thetav$ given $\yv$ is $\Ncal(\muov, \Sigmaov)$ where,
\begin{align}
\label{eq:conditional_mean_and_variance}
    \muov &= \Paren{\Xv^{\top}\Xv + \frac{\sigma^2}{c^2\delta^2}\Iv_d}^{-1}\Xv^{\top}\yv,\\
    \Sigmaov &= \sigma^2\Paren{\Xv^{\top}\Xv + \frac{\sigma^2}{c^2\delta^2}\Iv_d}^{-1}.
\end{align}

The ML estimator is then given by $\thetahv = \mathbb{E}\Br{\thetav \vert \yv} = \muov$.
To begin with, we state the following two lemmas that will be required for our analysis.
Their proofs are deferred to \S \ref{subsec:proof_entropy_ML_estimator} and \S \ref{subsec:proof_squared_error_ML_estimator}.

\begin{lemma_boxed}
\label{lem:entropy_ML_estimator}
For the linear model $\yv = \Xv\thetav + \vv$ and Gaussian prior $\thetav \sim F_{\delta} \equiv \Ncal(\mathbf{0}, c^2\delta^2\Iv_d)$, the differential entropy of the ML estimator $\thetahv = \mathbb{E}\Br{\thetav \vert \yv}$ is given by,
\begin{equation}
    h(\thetahv) = \frac{1}{2}\log\Paren{(2\pi e)^d \prod_{i=1}^{d}\frac{c^4\delta^4\sigma_i^2}{\sigma^2 + c^2\delta^2\sigma_i^2}},
\end{equation}

where $\{\sigma_i\}_{i \in [d]}$ denote the singular values of the feature matrix $\Xv \in \Real^{n \times d}$.
\end{lemma_boxed}

\begin{lemma_boxed}
\label{lem:squared_error_ML_estimator}
Under conditions of Lemma \ref{lem:entropy_ML_estimator}, the $\ell_2$-risk of the ML estimator $\thetahv$ is given by,
\begin{equation}
    \mathbb{E}_{\yv}\Br{\lVert\thetahv - \thetav\rVert_2^2} = \Tr{\Cov{\thetahv - \thetav}} = \sum_{i=1}^{d}\frac{c^2\delta^2\sigma^2}{\sigma^2 + c^2\delta^2\sigma_i^2},
\end{equation}

where $\{\sigma_i\}_{i \in [d]}$ denote the singular values of the feature matrix $\Xv \in \Real^{n \times d}$, $\Tr{\cdot}$ denotes the trace of a matrix, and $\Cov{\cdot}$ is the covariance.
\end{lemma_boxed}
For a prior distribution $\thetav \sim F_{\delta} \equiv \Ncal(\mathbf{0}, c^2\delta^2\Iv_d)$, the Bayes risk for any estimate $\zetav \in \Thetav$ is,
\begin{equation}
\label{eq:expression_Bayes_error_as_trace_of_covariance}
    \int_{\thetav \in \Real^d} R(\zetav, \thetav) dF_{\delta}(\thetav) = \mathbb{E}_{\zetav, \thetav \sim F_{\delta}}\Br{\frac{1}{d}\norm{\zetav - \thetav}_2^2} = \frac{1}{d}\Tr{\Cov{\zetav - \thetav}}
\end{equation}

Consider the decomposition of this error about the Bayes optimal estimator $\thetahv = \mathbb{E}\Br{\thetav \vert \yv} = \muov$.
Writing $\zetav - \thetav = (\zetav - \thetahv) + (\thetahv - \thetav)$, we have,
\begin{equation}
\label{eq:covariance_matrix_decomposition_about_ML_estimator}
    \Cov{\zetav - \thetav} = \Cov{\zetav - \thetahv} + \Cov{\thetahv - \thetav} - \mathbb{E}\Br{(\zetav - \thetahv)(\thetahv - \thetav)^{\top}} - \mathbb{E}\Br{(\thetahv - \thetav)(\zetav - \thetahv)^{\top}}
\end{equation}

Note that the first of the cross terms is
\begin{equation}
    \mathbb{E}\Br{(\zetav - \thetahv)(\thetahv - \thetav)^{\top}} = \mathbb{E}_{\yv}\Br{\mathbb{E}\Br{\zetav - \thetahv \big\vert \yv}\mathbb{E}\Br{\thetahv - \thetav \big\vert \yv}^{\top}} = \mathbf{0}
\end{equation}

This holds true because we require $p(\zetav\vert\yv, \thetav) = p(\zetav \vert \yv)$, i.e. given $\Xv$, $\thetav \to \yv \to \zetav$ forms a Markov chain.
So, $\zetav$ and $\thetav$ are conditionally independent given $\yv$, and since $\thetahv$ is a deterministic function of $(\Xv, \yv)$, consequently, $(\zetav - \thetahv)$ and $(\thetahv - \thetav)$ are also conditionally independent given $\yv$.
Moreover, from our choice of $\thetahv = \mathbb{E}\Br{\thetav \vert \yv}$, we have $\mathbf{E}\Br{\thetahv - \thetav \big\vert \yv} = \mathbf{0}$.
Similarly, the other cross-term also vanishes and we have from eqs. \eqref{eq:expression_Bayes_error_as_trace_of_covariance} and \eqref{eq:covariance_matrix_decomposition_about_ML_estimator},

\begin{equation}
    \label{eq:bayes_error_decomposition_about_ML_estimator}
    \int_{\thetav \in \Real^d}R(\zetav, \thetav)dF_{\delta}(\thetav) = \frac{1}{d}\Tr{\Cov{\zetav - \thetahv}} + \frac{1}{d}\Tr{\Cov{\thetahv - \thetav}}.
\end{equation}

We will use eq. \eqref{eq:bayes_error_decomposition_about_ML_estimator} later in the rest of the proof.
Now, to derive a lower bound on $\frac{1}{d}I(\yv; \zetav)$, note that,
\begin{align}
\label{eq:lower_bound_MI_chain_of_inequalities}
    I(\yv;\zetav) &\stackrel{\mathrm{(i)}}{\geq} I\left(\Paren{\Xv^{\top}\Xv + \frac{\sigma^2}{c^2\delta^2}\Iv_d}^{-1}\Xv^{\top}\yv; \zetav\right) \nonumber \\
    &= h(\thetahv) - h(\thetahv \vert \zetav) = h(\thetahv) - h(\thetahv - \zetav \vert \zetav) \stackrel{\mathrm{(ii)}}{\geq} h(\thetahv) - h(\thetahv - \zetav)
\end{align}

Inequality $\mathrm{(i)}$ follows from Data Processing Inequality (cf. \cite[Ch. 2]{info_theory_book}) and inequality $\mathrm{(ii)}$ holds true as conditioning reduces entropy.
Since for a given covariance matrix, the differential entropy is maximized for a Gaussian distribution, we have,

\begin{align*}
    h(\thetahv - \zetav) &\leq h\left(\Ncal\left(\mathbf{0}, \Cov{\thetahv - \zetav}\right)\right)\\
    &= \frac{1}{2}\log\Paren{(2\pi e)^d \left\vert \Cov{\thetahv - \zetav} \right\vert} \\
    &\stackrel{\mathrm{(i)}}{\leq} \frac{1}{2}\log\Paren{(2 \pi e)^d\Paren{\frac{1}{d}\Tr{\Cov{\thetahv - \zetav}}}^d} \\
    &\leq \frac{d}{2}\log(2\pi e) + \frac{d}{2}\log\Paren{\frac{1}{d}\Tr{\Cov{\thetahv - \zetav}}} \\
    &\stackrel{\mathrm{(ii)}}{=} \frac{d}{2}\log(2\pi e) + \frac{d}{2}\log\Paren{\int_{\thetav \in \Real^d}R(\zetav, \thetav)dF_{\delta}(\thetav) - \frac{1}{d}\Tr{\Cov{\thetahv - \thetav}}}.
\end{align*}

Here, $|\cdot|$ is used to denote the determinant of a matrix.
Inequality $\mathrm{(i)}$ follows from the fact that the determinant of a matrix is the product of eigenvalues and its trace is the sum of eigenvalues, along with a subsequent application of AM $\geq$ GM inequality, while $\mathrm{(ii)}$ follows from eq. \eqref{eq:bayes_error_decomposition_about_ML_estimator}.
So, \eqref{eq:lower_bound_MI_chain_of_inequalities} boils down to,
\begin{equation}
    I(\yv; \zetav) \geq h(\thetahv) - \frac{d}{2}\log(2 \pi e) - \frac{d}{2}\log\Paren{\int_{\thetav \in \Real^d}R(\zetav, \thetav)dF_{\delta}(\thetav) - \frac{1}{d}\Tr{\Cov{\thetahv - \thetav}}}.
\end{equation}
From Lemmas \ref{lem:entropy_ML_estimator} and \ref{lem:squared_error_ML_estimator}, we get,
\begin{align*}
    I(\yv; \zetav) &\geq \frac{1}{2}\log\Paren{(2\pi e)^d \prod_{i=1}^{d}\frac{c^4\delta^4\sigma_i^2}{\sigma^2 + c^2\delta^2\sigma_i^2}} - \frac{d}{2}\log(2 \pi e)\\ &\hspace{35mm}- \frac{d}{2}\log\Paren{\int_{\thetav \in \Real^d}R(\zetav, \thetav)dF_{\delta}(\thetav) - \frac{1}{d}\sum_{i=1}^{d}\frac{c^2\delta^2\sigma^2}{\sigma^2 + c^2\delta^2\sigma_i^2}} \\
    &= \frac{1}{2}\log\Paren{\prod_{i=1}^{d}\frac{c^4\delta^4\sigma_i^2}{\sigma^2 + c^2\delta^2\sigma_i^2}} - \frac{d}{2}\log\Paren{\int_{\thetav \in \Real^d}R(\zetav, \thetav)dF_{\delta}(\thetav) - \frac{1}{d}\sum_{i=1}^{d}\frac{c^2\delta^2\sigma^2}{\sigma^2 + c^2\delta^2\sigma_i^2}} \\
    &\geq \frac{1}{2}\log\Paren{\frac{c^4\delta^4\sigma_{\mathrm{min}}^2}{\sigma^2 + c^2\delta^2\sigma_{\mathrm{min}}^2}}^d - \frac{d}{2}\log\Paren{\int_{\thetav \in \Real^d}R(\zetav, \thetav)dF_{\delta}(\thetav) - \frac{c^2\delta^2\sigma^2}{\sigma^2 + c^2\delta^2\sigma_{\mathrm{min}}^2}} \\
    &= \frac{d}{2}\log\Paren{\frac{c^4\delta^4\sigma_{\mathrm{min}}^2}{\sigma^2 + c^2\delta^2\sigma_{\mathrm{min}}^2}\Paren{\int_{\thetav \in \Real^d}R(\zetav, \thetav)dF_{\delta}(\thetav) - \frac{c^2\delta^2\sigma^2}{\sigma^2 + c^2\delta^2\sigma_{\mathrm{min}}^2}}^{-1}}.
\end{align*}
This completes the proof of the Lemma \ref{lem:lower_bound_MI_Gaussian_prior}.
We complete this section by providing the proofs of Lemmas \ref{lem:entropy_ML_estimator} and \ref{lem:squared_error_ML_estimator}.

\subsection{Proof of Lemma \ref{lem:entropy_ML_estimator}: Differential entropy of ML estimator for Gaussian prior}
\label{subsec:proof_entropy_ML_estimator}

The ML estimator is given by $\thetahv = \Paren{\Xv^{\top}\Xv + \frac{\sigma^2}{c^2\delta^2}\Iv_d}^{-1}\Xv^{\top}\yv$.
Since $\yv = \Xv\thetav + \vv$ where $\vv \sim \Ncal(\mathbf{0}, \sigma^2\Iv_n)$, and the prior is $\thetav \sim \Ncal(\mathbf{0}, c^2\delta^2\Iv_d)$, this implies that $\yv$ is a Gaussian vector with zero mean and covariance given by,
\begin{align}
     \Cov{\yv} = \mathbb{E}\Br{\yv\yv^{\top}} &= \mathbb{E}\Br{(\Xv\thetav + \vv)(\Xv\thetav + \vv)^{\top}} \nonumber \\
     &= \Xv \mathbb{E}\Br{\thetav \thetav^{\top}}\Xv^{\top} + \mathbb{E}\Br{\vv\vv^{\top}} = c^2\delta^2\Xv\Xv^{\top} + \sigma^2\Iv_n
\end{align}

So, for fixed $\Xv$, $\Cov{\thetahv}$ is given by,
\begin{align}
\label{eq:covariance_matrix_ML_estimator_expression}
    \Cov{\thetahv} &= \Paren{\Xv^{\top}\Xv + \frac{\sigma^2}{c^2\delta^2}\Iv_d}^{-1}\Xv^{\top} \mathbb{E}\Br{\yv \yv^{\top}} \Xv \Paren{\Xv^{\top}\Xv + \frac{\sigma^2}{c^2\delta^2}\Iv_d}^{-1} \nonumber \\
    &= \Paren{\Xv^{\top}\Xv + \frac{\sigma^2}{c^2\delta^2}\Iv_d}^{-1}\Xv^{\top} (c^2\delta^2\Xv\Xv^{\top} + \sigma^2\Iv_n) \Xv \Paren{\Xv^{\top}\Xv + \frac{\sigma^2}{c^2\delta^2}\Iv_d}^{-1}.
\end{align}

Consider the singular value decomposition of $\Xv \in \Real^{n \times d} (n \geq d)$ to be $\Xv = \Uv \Sigmav \Vv^{\top}$, where $\Uv \in \Real^{n \times n}, \Sigmav \in \Real^{n \times d}$, and $\Vv \in \Real^{d \times d}$.
Here $\Sigmav$ is a diagonal matrix comprising of singular values $\geq \sigma_1, \ldots, \sigma_d$ such that $\sigma_1 \leq \sigma_{\mathrm{max}}$ and $\sigma_d \geq \sigma_{\mathrm{min}}$.
Then, $\Xv^{\top}\Xv = \Vv\Sigmav^{\top}\Uv^{\top}\Uv\Sigmav \Vv^{\top} = \Vv \Sigmav_1 \Vv^{\top}$, where $\Sigmav_1 = \Sigmav^{\top}\Sigmav \in \Real^{d \times d}$.
Note that without loss of generality, we can assume $\Xv$ is full column rank, i.e. all its singular values are positive.
This implies $\Sigmav_1$ is full rank since $n \geq d$, giving us,

\begin{align}
    \Paren{\Xv^{\top}\Xv + \frac{\sigma^2}{c^2\delta^2}\Iv_d}^{-1} &= \Vv\Paren{\Sigmav_1 + \frac{\sigma^2}{c^2}\Iv_d}^{-1}\Vv^{\top}, \nonumber \\ \text{ and, }\hspace{8mm} c^2\delta^2 \Xv\Xv^{\top} + \sigma^2\Iv_n &= \Uv(c^2\delta^2\Sigmav\Sigmav^{\top} + \sigma^2\Iv_n)\Uv^{\top}.
\end{align}

Substituting these expressions in eq. \eqref{eq:covariance_matrix_ML_estimator_expression} yields,

\begin{equation}
    \Cov{\thetahv} = \Vv\Br{\Paren{\Sigmav_1 + \frac{\sigma^2}{c^2\delta^2}\Iv_d}^{-1}\Paren{c^2\delta^2\Sigmav_1^2 + \sigma^2\Sigmav_1}\Paren{\Sigmav_1 + \frac{\sigma^2}{c^2\delta^2}\Iv_d}^{-1}}\Vv^{\top}
\end{equation}

Since the determinant of a matrix is the product of its eigenvalues, this gives us:

\begin{align}
\label{eq:covariance_ML_estimator_product_expression}
    \left\lvert \Cov{\thetahv} \right\rvert &= \Br{\prod_{i=1}^{d}\Paren{\sigma_i^2 + \frac{\sigma^2}{c^2\delta^2}}}^{-1} \prod_{i=1}^{d}\sigma_i^2\Paren{c^2\delta^2\sigma_i^2 + \sigma^2} \Br{\prod_{i=1}^{d}\Paren{\sigma_i^2 + \frac{\sigma^2}{c^2\delta^2}}}^{-1} \nonumber\\
    &= \prod_{i=1}^{d}\frac{c^4\delta^4\sigma_i^2}{\sigma^2 + c^2\delta^2\sigma_i^2}.
\end{align}

Since the ML estimator $\thetatv$ is a Gaussian random variable, substituting \eqref{eq:covariance_ML_estimator_product_expression} in the expression for the entropy $h(\thetahv) = \frac{1}{2}\log\Paren{(2\pi e)^d\left\lvert \Cov{\thetahv}\right\rvert}$ completes the proof.

\subsection{Proof of Lemma \ref{lem:squared_error_ML_estimator}: \texorpdfstring{$\ell_2$}{l2}-risk of ML estimator under Gaussian prior}
\label{subsec:proof_squared_error_ML_estimator}

To find a closed form expression for $\Tr{\Cov{\thetahv - \thetav}}$, where $\yv = \Xv \thetav + \vv$, $\thetav \sim \Ncal(\mathbf{0}, c^2\delta^2\Iv_d)$, $\vv \sim \Ncal(\mathbf{0}, \sigma^2\Iv_n)$, and $\thetahv = \mathbb{E}\Br{\thetav\vert\yv} = \Paren{\Xv^{\top}\Xv + \frac{\sigma^2}{c^2\delta^2}\Iv_d}^{-1}\Xv^{\top}\yv$, note that:

\begin{equation*}
    \thetahv = \Paren{\Xv^\top \Xv + \frac{\sigma^2}{c^2\delta^2}\Iv_d}^{-1}\Xv^\top\yv = \Paren{\Xv^\top \Xv + \frac{\sigma^2}{c^2\delta^2}\Iv_d}^{-1}\Xv^\top\Paren{\Xv\thetav + \vv}
\end{equation*}

So,
\begin{equation*}
\thetahv - \thetav = \underbrace{\Paren{\Paren{\Xv^\top \Xv + \frac{\sigma^2}{c^2\delta^2}\Iv_d}^{-1}\Xv^\top \Xv - \Iv_d}\thetav}_{=\uv_1} + \underbrace{\Paren{\Xv^\top \Xv + \frac{\sigma^2}{c^2\delta^2}\Iv_d}^{-1}\Xv^\top\vv}_{=\uv_2}   
\end{equation*}

Since $\thetav$ and $\vv$ are independent of each other, so are $\uv_1$ and $\uv_2$, and hence their covariances can be computed separately and added up (i.e. no cross terms).
Considering the SVD as $\Xv = \Uv \Sigmav \Vv^{\top}$, we have,
\begin{align*}
    \Cov{\uv_1} &= \mathbb{E}\Br{\uv_1 \uv_1^\top}\\
    &= c^2\delta^2\Br{\Paren{\Xv^\top \Xv + \frac{\sigma^2}{c^2\delta^2}\Iv_d}^{-1}\Xv^\top \Xv - \Iv_d}\Br{\Xv^\top \Xv \Paren{\Xv^\top \Xv + \frac{\sigma^2}{c^2\delta^2}\Iv_d}^{-1} - \Iv_d} \\
    &= \Vv \Br{c^2\delta^2\Paren{\Paren{\Sigmav_1 + \frac{\sigma^2}{c^2\delta^2}\Iv_d}^{-1}\Sigmav_1 - \Iv_d}\Paren{\Sigmav_1\Paren{\Sigmav_1 + \frac{\sigma^2}{c^2\delta^2}\Iv_d}^{-1} - \Iv_d}}\Vv^\top
\end{align*}

Similarly,
\begin{align*}
    \Cov{\uv_2} = \mathbb{E}\Br{\uv_2 \uv_2^\top}
    &= \sigma^2\Paren{\Xv^\top \Xv + \frac{\sigma^2}{c^2\delta^2}\Iv_d}^{-1}\Xv^\top \Xv \Paren{\Xv^\top \Xv + \frac{\sigma^2}{c^2\delta^2}\Iv_d}^{-1} \\
    &= \Vv\Br{\sigma^2\Paren{\Sigmav_1 + \frac{\sigma^2}{c^2\delta^2}\Iv_d}^{-1}\Sigmav_1 \Paren{\Sigmav_1 + \frac{\sigma^2}{c^2\delta^2}\Iv_d}^{-1}}\Vv^\top
\end{align*}

Since trace of a matrix is sum of its eigenvalues, this gives us,
\begin{align*}
    \Tr{\Cov{\thetav - \thetahv}} &= \sum_{i=1}^{d}\Br{c^2\delta^2\Paren{1 - \sigma_i^2\Paren{\sigma_i^2 + \frac{\sigma^2}{c^2\delta^2}}^{-1}}^2 + \sigma^2\sigma_i^2\Paren{\sigma_i^2 + \frac{\sigma^2}{c^2\delta^2}}^{-2}} \nonumber \\
    &= \sum_{i=1}^{d}\frac{c^2\delta^2\sigma^2}{\sigma^2 + c^2\delta^2\sigma_i^2}.
\end{align*}

This completes the proof of the lemma.

\section{Proof of Theorem \ref{thm:minimax_error_lower_bound}: Minimax error lower bound}
\label{app:proof_minimax_error_lower_bound}

Let $\Qsf$ be any $(n,d,B)$-learning code.
From Lemma \ref{lem:mutual_information_bit_budget_inequality}, $I(\yv;\Qsf(\Xv, \yv)) \leq dB$.
Since for a given $\Xv$, $\Qsf \equiv \Qsf(\Xv, \yv)$ solely depends on $\yv$, we also have $p(\Qsf \vert \yv, \thetav) = p(\Qsf \vert \yv)$.
Rearranging Lemma \ref{lem:lower_bound_MI_Gaussian_prior} yields the lower bound on the Bayes risk with Gaussian prior $F \equiv \Ncal(\mathbf{0}, c^2\delta^2\Iv_d)$ as,

\begin{align}
\label{eq:Bayes_error_lower_bound}
     \int_{\thetav \in \Real^d} R(\Qsf(\Xv), \thetav)dF(\thetav) &\geq \frac{c^2\delta^2\sigma^2}{\sigma^2 + c^2\delta^2\sigma_{\mathrm{min}}^2} + \frac{c^4\delta^4\sigma_{\mathrm{min}}^2}{\sigma^2 + c^2\delta^2\sigma_{\mathrm{min}}^2}\cdot 2^{-\frac{2}{d}I(\yv;\Qsf(\Xv, \yv))} \nonumber\\
    &\geq \frac{c^2\delta^2\sigma^2}{\sigma^2 + c^2\delta^2\sigma_{\mathrm{min}}^2} + \frac{c^4\delta^4\sigma_{\mathrm{min}}^2}{\sigma^2 + c^2\delta^2\sigma_{\mathrm{min}}^2}\cdot 2^{-2B}.
\end{align}

This formalizes the argument in key idea $\mathrm{(2)}$ in \S \ref{sec:lower_bound_to_the_minimax_risk} of the main paper.
We then resort to key idea $\mathrm{(1)}$, i.e. Lemma \ref{lem:inequality_supremum_bayes_error} that lower bounds the asymptotic ($\liminf_{d \to \infty}$) worst-case risk over $\Thetav$ by the Bayes risk over $\Real^d$ in \eqref{eq:Bayes_error_lower_bound}.
Since the R.H.S. is a continuous function of $\delta$, and the inequality holds true for any $\delta \in (0,1)$, letting $\delta \to 1$, and then taking an infimum over all $\Qsf \in \Qcal_{n,d,B}$, completes the proof.

\section{Tail probability bound for the estimated magnitude \texorpdfstring{$\bt^2$}{bt2}}
\label{sec:error_encoded_magnitude}

In \S \ref{subsec:naive_learning_and_quantization} of the main paper, we note that an unquatized estimate of the magnitude of $\thetav$, i.e. $b^2 \triangleq \frac{1}{d}\lVert \thetav\rVert_2^2$ can be obtained as follows.
Given that $\yv = \Xv \thetav + \vv$, we have,

\begin{align*}
    \Xinv\yv = \thetav + \Xinv\vv &\stackrel{\mathrm{(i)}}{\implies} \frac{1}{d}\lVert \thetav \rVert_2^2 =  \frac{1}{d}\lVert \Xinv\yv \rVert_2^2 - \frac{1}{d}\lVert \Xinv\vv \rVert_2^2 - \frac{2}{d}\inprod{\thetav, \Xinv\vv} \\
    &\stackrel{\mathrm{(ii)}}{\implies} \frac{1}{d}\lVert \thetav \rVert_2^2 =  \frac{1}{d}\mathbb{E}_{\vv}\lVert \Xinv\yv \rVert_2^2 - \frac{1}{d}\mathbb{E}_{\vv}\lVert \Xinv\vv \rVert_2^2 \\
    &\implies b^2 = \frac{1}{d}\mathbb{E}_{\vv}\lVert \Xinv\yv \rVert_2^2 - \frac{\sigma^2}{d}\sum_{i=1}^{d}\frac{1}{\sigma_i^2}.
\end{align*}

Since we have a single realization of the noise vector $\vv$, giving us access to just one realization $\yv$, a reasonable estimate of $b^2$ is obtained by simply using $\lVert \Xinv\yv \rVert_2^2$ as an approximation for the expected value, $\mathbb{E}_{\vv}\lVert \Xinv\yv \rVert_2^2$.
In this section, we present Lemma \ref{lem:error_encoded_magnitude} that gives a tail probability bound on the deviation of the estimated magnitude $\bhat^2$ and the true magnitude of the model, $b^2$.
Since $b^2$ is encoded using a scalar quantizer of resolution $\frac{1}{\sqrt{d}}$ to get $\bt^2$, it follows from Lemma \ref{lem:error_encoded_magnitude} that $\bt^2 - b^2 = O_P\Paren{\frac{1}{\sqrt{d}}}$, where $O_P(\frac{1}{\sqrt{d}})$ denotes a quantity which approaches $0$ with high probability as $d \to \infty$.

\begin{lemma_boxed}
\label{lem:error_encoded_magnitude}
Suppose that $\yv \sim \Ncal\Paren{\Xv \thetav, \sigma^2\Iv_n}$ for $\Xv \in \Real^{n \times d}, (n \geq d)$, and $\thetav \in \Real^d$ with $\frac{1}{d}\norm{\thetav}_2^2 = b^2$.
Then for $t \geq 0$,

\begin{align*}
    &\mathbb{P}\left(\left\lvert \frac{1}{d}\lVert\Xinv \yv\rVert_2^2 - \frac{\sigma^2}{d}\sum_{i=1}^{d}\frac{1}{\sigma_i^2} - b^2 \right\rvert \geq t \right)\\
    &\leq 2\exp\Paren{-\frac{dt^2\sigma_{\mathrm{min}}^2}{24\sigma^4}} + \frac{8\sigma b}{\sigma_{\mathrm{min}}t\sqrt{2\pi d}}\exp\Paren{-\frac{dt^2\sigma_{\mathrm{min}}^2}{32\sigma^2b^2}}.
\end{align*}

In other words, $\bhat^2 - b^2 = O_P\Paren{\frac{1}{\sqrt{d}}}$, where $\bhat^2 \triangleq \frac{1}{d}\norm{\Xinv\yv}_2^2 - \frac{\sigma^2}{d}\sum_{i=1}^{d}\frac{1}{\sigma_i^2}$, where $\sigma_{\mathrm{max}}$ and $\sigma_{\mathrm{min}}$ are upper and lower bounds to the maximum and minimum singular values of $\Xv$, respectively.
\end{lemma_boxed}
\begin{proof}
Let $\yv = \Xv\thetav + \epsb$, where $\epsb \in \Ncal(\mathbf{0}, \sigma^2\Iv_n)$.
Then, $\Xinv\yv = \Xinv\Paren{\Xv\thetav + \epsb} = \thetav + \Xinv\epsb$.
To study the concentration of $\frac{1}{d}\norm{\Xinv\yv}_2^2 - \frac{\sigma^2}{d}\sum_{i=1}^{d}\frac{1}{\sigma_i^2}$ around $b^2$, note that:

\begin{align*}
    &\frac{1}{d}\norm{\Xinv\yv}_2^2 = \frac{1}{d}\norm{\thetav + \Xinv\epsb}_2^2 = \frac{1}{d}\norm{\thetav}_2^2 + \frac{1}{d}\norm{\Xinv\epsb}_2^2 + \frac{2}{d}\inprod{\thetav, \Xinv\epsb} \\
    \implies &\frac{1}{d}\norm{\Xinv\yv}_2^2 - \frac{\sigma^2}{d}\sum_{i=1}^{d}\frac{1}{\sigma_i^2} - b^2 = \frac{1}{d}\norm{\Xinv \epsb}_2^2 + \frac{2}{d}\inprod{\thetav, \Xinv\epsb} - \frac{\sigma^2}{d}\sum_{i=1}^{d}\frac{1}{\sigma_i^2}.
\end{align*}

So, it suffices to find an upper bound on the tail probability,

\[\mathbb{P}\Paren{\left\lvert \frac{1}{d}\norm{\Xinv\epsb}_2^2 - \frac{\sigma^2}{d}\sum_{i=1}^{d}\frac{1}{\sigma_i^2} + \frac{2}{d}\inprod{\thetav, \Xinv \epsb} \right\rvert \geq t}.\]

We have,
\begin{align}
\label{eq:tail_probability_1}
    &\mathbb{P}\Paren{\left\lvert \frac{1}{d}\norm{\Xinv\epsb}_2^2 - \frac{\sigma^2}{d}\sum_{i=1}^{d}\frac{1}{\sigma_i^2} + \frac{2}{d}\inprod{\thetav, \Xinv \epsb} \right\rvert \geq t} \nonumber \\ 
    \leq \hspace{2mm}&\mathbb{P}\Paren{\left\lvert \frac{1}{d}\norm{\Xinv\epsb}_2^2 - \frac{\sigma^2}{d}\sum_{i=1}^{d}\frac{1}{\sigma_i^2} \right\rvert \geq \frac{t}{2}} + \mathbb{P}\Paren{\left\lvert \frac{2}{d}\inprod{\thetav, \Xinv \epsb} \right\rvert \geq \frac{t}{2}}
\end{align}

Since $\Xinv\epsb \sim \Ncal\Paren{\mathbf{0}, \sigma^2\Xinv{\Xinv}^\top}$, the first term in \eqref{eq:tail_probability_1} can be upper bounded using Chernoff's method (cf. \cite{tyronCMarticle}).

\fbox{
\begin{minipage}{\textwidth}
\textbf{Fact}: For a multivariate normal distribution $\zv \in \Real^d \sim \Ncal\Paren{\mathbf{0}, \Cv}$ with covariance matrix $\Cv$, the  moment generating function of $\norm{\zv}_2^2$ is given by,
\begin{equation}
    \mathbb{E}\Br{e^{\pm s\norm{\zv}_2^2}} = \frac{1}{\sqrt{\Det{\Iv \pm 2s\Cv}}}
\end{equation}
\end{minipage}
}

Moreover, $\mathbb{E}\norm{\zv}_2^2 = \Tr{\Cov{\yv}} = \Tr{\Cv}$.
Using Chernoff's bound for $\zv \triangleq \Xinv \epsb$ with $\Cv = \sigma^2\Xinv{\Xinv}^\top$ gives us,
\begin{align}
    \label{eq:tail_probability_2}
    &\mathbb{P}\Br{\norm{\zv}_2^2 \geq \Paren{1 + \varepsilon}\mathbb{E}\norm{\zv}_2^2} \leq \min_{s > 0} \frac{\mathbb{E}\Br{e^{s\norm{\zv}_2^2}}}{e^{s(1+\varepsilon)\mathbb{E}\Br{\norm{\zv}_2^2}}} \nonumber \nonumber \\
    \implies &\mathbb{P}\Br{\norm{\zv}_2^2 - \mathbb{E}\norm{\zv}_2^2 \geq \varepsilon \mathbb{E}\norm{\zv}_2^2} \leq \min_{s > 0} \frac{\Det{\Iv - 2s\Cv}^{-\frac{1}{2}}}{e^{s\Paren{1 + \varepsilon}\Tr{\Cv}}}
\end{align}

To evaluate $\Tr{\Cv}$, consider the SVD of $\Xv = \Uv \Sigmav \Vv^{\top}$, where $\Sigmav \in \Real^{n \times d}$ consists of the singular values of $\Xv$, i.e. $\sigma_1 \geq \ldots \geq \sigma_d$.
The SVD of $\Xinv \in \Real^{d \times n}$ is given by $\Xinv = \Vv \Sigmav' \Uv^{\top}$, where the diagonal entries of $\Sigmav' \in \Real^{d \times n}$ are $\sigma_1^{-1}, \sigma_2^{-1}, \ldots, \sigma_d^{-1}$.
This yields the eigenvalue decomposition of $\Xinv {\Xinv}^{\top} \in \Real^{d \times d}$ to be $\Xinv{\Xinv}^\top = \Vv \Sigmav^\top \Uv^\top \Uv \Sigmav \Vv^\top = \Vv \Sigmav_1 \Vv^{\top}$, where $\Sigmav_1 = \mathrm{diag}\{\sigma_1^{-2}, \ldots, \sigma_d^{-2}\}$.
So,

\begin{equation}
    \Tr{\Xinv{\Xinv}^{\top}} = \sum_{i=1}^{d}\frac{1}{\sigma_i^2} \implies \mathbb{E}\norm{\Xinv\epsb}_2^2 = \Tr{\Cv} = \sigma^2 \sum_{i=1}^{d}\frac{1}{\sigma_i^2} \geq \frac{\sigma^2 d}{\sigma_{\mathrm{min}}^2}.
\end{equation}

Moreover, the eigenvalues of $\Iv - 2s\Cv$ are given by $\Paren{1 - 2s\frac{\sigma^2}{\sigma_i^2}}$ for $i = 1, 2, \ldots, d$, meaning,
\begin{equation}
    \Det{\Iv - 2s\Cv} = \prod_{i=1}^{d}\Paren{1 - 2s\frac{\sigma^2}{\sigma_i^2}} \leq \Paren{1 - 2s\frac{\sigma^2}{\sigma_{\mathrm{min}}^2}}^d.
\end{equation}
From eq. \eqref{eq:tail_probability_2}, setting $s = \frac{\varepsilon \sigma_{\mathrm{min}}^2}{2\Paren{1 + \varepsilon}\sigma^2}$, we get,

\begin{align}
    \mathbb{P}\Br{\norm{\zv}_2^2 - \mathbb{E}\norm{\zv}_2^2 \geq \varepsilon \mathbb{E}\norm{\zv}_2^2} \leq \frac{\Paren{1 - 2s\frac{\sigma^2}{\sigma_{\mathrm{min}}^2}}^{-\frac{d}{2}}}{e^{s\Paren{1 + \varepsilon}\frac{\sigma^2 d}{\sigma_{\mathrm{min}}^2}}} &= \Paren{1 - 2s\frac{\sigma^2}{\sigma_{\mathrm{min}}^2}}^{-\frac{d}{2}}e^{-s\Paren{1 + \varepsilon}\frac{\sigma^2 d}{\sigma_{\mathrm{min}}^2}} \nonumber \\
    &= \Paren{\Paren{1 + \varepsilon}e^{-\varepsilon}}^{\frac{d}{2}}.
\end{align}

Utilizing the inequalities $\Paren{1 + \varepsilon}e^{-\varepsilon} = e^{-\varepsilon + \log\Paren{1 + \varepsilon}}$, and $\log\Paren{1 + \varepsilon} \leq \varepsilon - \frac{\varepsilon^2}{2} + \frac{\varepsilon^2}{3}$, and from symmetry, we have,
\begin{align}
    &\mathbb{P}\Br{\norm{\zv}_2^2 - \mathbb{E}\norm{\zv}_2^2 \geq \varepsilon \mathbb{E}\norm{\yv}_2^2} \leq \Paren{e^{-\Paren{\frac{\varepsilon^2}{2} - \frac{\varepsilon^3}{3}}}}^{\frac{d}{2}} \leq e^{-\frac{\varepsilon^2d}{6}} \nonumber \\
    \implies &\mathbb{P}\Br{\left\lvert \norm{\Xinv\epsb}_2^2 - \sigma^2\sum_{i=1}^{d}\frac{1}{\sigma_i^2} \right\rvert \geq \varepsilon \sigma^2\sum_{i=1}^{d}\frac{1}{\sigma_i^2}} \leq 2 e^{-\frac{\varepsilon^2 d}{6}} \nonumber \\
    \implies &\mathbb{P}\Br{\left\lvert \norm{\Xinv\epsb}_2^2 - \sigma^2\sum_{i=1}^{d}\frac{1}{\sigma_i^2} \right\rvert \geq \varepsilon \frac{\sigma^2 d}{\sigma_{\mathrm{min}}^2}} \leq 2 e^{-\frac{\varepsilon^2 d}{6}}.
\end{align}

The last inequality follows from the fact that $\frac{\varepsilon^2}{2} - \frac{\varepsilon^3}{3} \geq \frac{\varepsilon^2}{3}$ for $0 < \varepsilon \leq 1/2$.
Setting $\varepsilon = \frac{t}{2}\cdot \frac{\sigma_{\mathrm{min}}^2}{\sigma^2}$, the first term of \eqref{eq:tail_probability_1} can be bounded as,
\begin{equation}
    \label{eq:tail_probability_3}
    \mathbb{P}\Br{\left\lvert \frac{1}{d}\norm{\Xinv\epsb}_2^2 - \frac{\sigma^2}{d}\sum_{i=1}^{d}\frac{1}{\sigma_i^2} \right\rvert \geq \frac{t}{2}} \leq 2e^{-\frac{dt^2\sigma_{\mathrm{min}}^4}{24\sigma^4}}.
\end{equation}
To upper bound the second term in \eqref{eq:tail_probability_1}, let us consider the tail probability $\mathbb{P}\Paren{\inprod{{\Xinv}^{\top}\thetav, \epsb} \geq \frac{dt}{4}}$.
Note that $Z = \inprod{{\Xinv}^{\top}\thetav, \epsb} \sim \Ncal\Paren{0, \sigmat^2}$, where $\sigmat^2 = \Paren{{\Xinv}^\top\thetav}^\top \mathbb{E}\Br{\epsb\epsb^\top}{\Xinv}^\top\thetav = \sigma^2 \thetav^\top \Xinv {\Xinv}^{\top} \thetav$.
Recalling the SVD of $\Xinv$ to be $\Xinv = \Vv \Sigmav' \Uv^\top$, it can be shown that,
\begin{equation*}
    \sigmat^2 = \sigma^2\sum_{i=1}^{d}\sigma_i^{-2}\Paren{\vv_i^\top \thetav}^2 \leq \frac{\sigma^2}{\sigma_{\mathrm{min}}^2}\norm{\Vv^\top\thetav}_2^2 = \frac{\sigma^2}{\sigma_{\mathrm{min}}^2}\norm{\thetav}_2^2 = \frac{\sigma^2db^2}{\sigma_{\mathrm{min}}^2}.
\end{equation*}
Here, $\{\vv_i\} \in \Real^d, i = 1, \ldots, d$ denote the columns of the matrix $\Vv \in \Real^{d \times d}$.
Using the tail probability expression of a normal random variable yields an upper bound on the second term of \eqref{eq:tail_probability_1} as follows,

\begin{equation}
    \label{eq:tail_probability_4}
    \mathbb{P}\Br{\left\lvert\frac{2}{d}\inprod{\thetav, \Xinv\epsb}\right\rvert \geq \frac{t}{2}} \leq \frac{8\sigma b}{\sigma_{\mathrm{min}}t\sqrt{2 \pi d}}\exp\Paren{-\frac{dt^2\sigma_{\mathrm{min}}^2}{32\sigma^2 b^2}} \sim O_P\Paren{\frac{1}{\sqrt{d}}}.
\end{equation}

Equations \eqref{eq:tail_probability_3} and \eqref{eq:tail_probability_4} together complete the proof of Lemma \ref{lem:error_encoded_magnitude}.
\end{proof}

\section{Proof of Theorem \ref{thm:naive_quantizer_guarantee}: Na\"ive Quantizer guarantee}
\label{app:naive_quantizer_guarantee}

For the purposes of this proof, let us denote the scaling factor used in the naive learning and quantization strategy as $\gammat$, i.e.,
\begin{equation}
\label{eq:gamma_dash_scaling}
    \gammat \triangleq \sqrt{\frac{d\bt^4}{\bt^2 + \frac{\sigma^2}{d}\xi}}.
\end{equation}
Then the quantized model learned by the naive strategy, given by $\thetatv = \gammat\stv$, has an $\ell_2$-risk,
\begin{equation}
\label{eq:risk_decomposition_naive}
    R\Paren{\thetatv, \thetav} = \frac{1}{d}\lVert\thetatv - \thetav\rVert_2^2 \leq \underbrace{\frac{2}{d}\lVert\thetatv - \gammah\Xinv\yv\rVert_2^2}_{A_1} + \underbrace{\frac{2}{d}\norm{\gammah\Xinv\yv - \thetav}_2^2}_{A_2},
\end{equation}
where $\gammah \triangleq \frac{\bhat^2}{\bhat^2 + \frac{\sigma^2}{d}\xi}$.

\textbf{Analyzing the term $A_1$}.
The first term is eq. \eqref{eq:risk_decomposition_naive} can be upper bounded as,
\begin{align}
\label{eq:naive_term_A1_decomposition}
    \frac{1}{d}\norm{\thetatv - \gammah\Xinv\yv}_2^2 &= \frac{1}{d}\norm{\gammat\stv - \gammah\norm{\Xinv\yv}_2\cdot\frac{\Xinv\yv}{\norm{\Xinv\yv}_2}}_2^2 \nonumber \\
    &= \frac{1}{d}\norm{\gammah\norm{\Xinv\yv}_2\cdot\stv - \gammah\norm{\Xinv\yv}_2\cdot\frac{\Xinv\yv}{\norm{\Xinv\yv}_2} + \Paren{\gammat - \gammah\norm{\Xinv\yv}_2}\cdot\stv}_2^2 \nonumber \\
    &\leq \frac{2}{d}\gammah^2\norm{\Xinv\yv}_2^2\cdot\norm{\stv - \frac{\Xinv\yv}{\norm{\Xinv\yv}_2}}_2^2 + \frac{2}{d}\Paren{\gammat - \gammah\norm{\Xinv\yv}_2}^2\norm{\stv}_2^2
\end{align}
The second term in \eqref{eq:naive_term_A1_decomposition} is small.
Since $\stv = \text{Q}_{u,B}\Paren{\frac{\Xinv\yv}{\lVert \Xinv\yv\rVert_2}}$, this implies $\lVert \stv \rVert_2^2 \leq 1$.
Furthermore, using the facts $\gammah = \bhat^2/(\bhat^2 + \sigma^2\xi/d)$, and $\bhat^2 = \frac{1}{d}\lVert \Xinv\yv \rVert_2^2 - \frac{\sigma^2}{d}\xi$, we get,
\begin{equation*}
    \frac{1}{\sqrt{d}}\Paren{\gammat - \gammah\norm{\Xinv\yv}_2} = \sqrt{\frac{\bt^4}{\bt^2 + \frac{\sigma^2}{d}\xi}} - \sqrt{\frac{\bhat^4}{\bhat^2 + \frac{\sigma^2}{d}\xi}} = O_P\Paren{\frac{1}{\sqrt{d}}}.
\end{equation*}
Whereas, the first term in eq. \eqref{eq:naive_term_A1_decomposition} satisfies,
\begin{equation*}
    \frac{2}{d}\gammah^2\lVert \Xinv\yv \rVert_2^2 \cdot\left\lVert \stv - \frac{\Xinv\yv}{\lVert \Xinv\yv \rVert_2} \right\rVert_2^2 \leq \frac{2}{d}\cdot\frac{d\bhat^4}{\Paren{\bhat^2 + \frac{\sigma^2}{d}\xi}}\cdot2^{-2B}d = \frac{2db^4}{b^2 + \frac{\sigma^2}{d}\xi}\cdot 2^{-2B} + O_P\Paren{\frac{1}{\sqrt{d}}}.
\end{equation*}
The first inequality follows from the fact that for a uniform scalar quantizer, the worst-case quantization error is given as,
\begin{equation*}
    \sup_{\zv \in \textbf{B}_{\infty}^d(1)}\norm{\text{Q}_{u,B}(\zv) - \zv} = \frac{\Delta}{2}\sqrt{d} = \frac{\sqrt{d}}{M} = 2^{-B}\sqrt{d}.
\end{equation*}
Consequently,
\begin{equation}
\label{eq:naive_A1_simplification}
    A_1 = \frac{2db^4}{\Paren{b^2 + \frac{\sigma^2}{d}\xi}}\cdot 2^{-2B} + O_P\Paren{\frac{1}{\sqrt{d}}}.
\end{equation}

\textbf{Analyzing the term $A_2$}.
Term $A_2$ has nothing to do with the quantization scheme used for model compression and is essentially the unavoidable loss for the linear estimate $\gammah \Xinv\yv$, which is present even if there is no bit budget constraint.
Suppose $\Av$ is a $d \times d$ orthonormal matrix such that $\Av\thetav = \Paren{\sqrt{d}b, 0, \ldots, 0}^\top \triangleq \tauv$.
Let $\ytv = \Av\Xinv\yv$.
Then, $\ytv \triangleq \Av\thetav + \Av\Xinv\vv \sim \Ncal\Paren{\tauv, \sigma^2\Av\Xinv{\Xinv}^\top\Av^\top}$.
We have,
\begin{align}
    \label{eq:naive_A2_simplification}
        A_2 = \frac{1}{d}\norm{\gammah\Xinv\yv - \thetav}_2^2 = \frac{1}{d}\norm{\gammah\Av\Xinv\yv - \Av\thetav}_2^2 &= \frac{1}{d}\norm{\gammah\ytv - \tauv}_2^2 \nonumber \\
        &= \frac{1}{d}\gammah^2\norm{\ytv}_2^2 + b^2 - 2\gammah \frac{1}{d}\inprod{\ytv,\tauv} \nonumber \\
        &= \frac{b^2\cdot\frac{\sigma^2\xi}{d}}{\Paren{b^2 + \frac{\sigma^2}{d}\xi}} + O_P\Paren{\frac{1}{\sqrt{d}}}
\end{align}
The last equality follows from the fact that $\norm{\Xinv\yv}_2 = \norm{\ytv}_2$, and $\frac{1}{\sqrt{d}}\inprod{\ytv, \tauv} = \frac{\yt_1}{\sqrt{d}}$ concentrates around its mean $b$, as $\frac{\yt_1}{\sqrt{d}} = b + O_P\Paren{\frac{1}{\sqrt{d}}}$.
Here, $\yt_1$ is the first coordinate of $\ytv$.

\textbf{Completing the proof of Thm \ref{thm:naive_quantizer_guarantee}}.
From eqs. \eqref{eq:naive_A1_simplification} and \eqref{eq:naive_A2_simplification}, eq. \eqref{eq:risk_decomposition_naive} simplifies to
\begin{equation}
    R\left(\thetatv, \thetav\right) \leq \frac{2b^2\cdot\frac{\sigma^2}{d}\xi}{b^2 + \frac{\sigma^2}{d}\xi} + \frac{2db^4}{\Paren{b^2 + \frac{\sigma^2}{d}\xi}}\cdot 2^{-2B} + O_P\Paren{\frac{1}{\sqrt{d}}}.
\end{equation}

A uniform high-probability bound for all models in the euclidean ball $\frac{1}{d}\lVert \thetav \rVert_2^2 \leq c^2$ can be obtained as follows:
\begin{equation}
    \mathbb{P}\Paren{R\left(\thetatv, \thetav\right) > \frac{2c^2\sigma^2}{\sigma^2 + c^2\sigma_{\mathrm{min}}^2} + \frac{2dc^4\sigma_{\mathrm{min}}^2}{\sigma^2 + c^2\sigma_{\mathrm{min}}^2}\cdot 2^{-2B} + \frac{C_n}{\sqrt{d}}} \to 0,
\end{equation}
as $d \to \infty$, for some constant $C_n$ independent of dimension $d$.
This completes the proof.

\section{Proof of Theorem \ref{thm:RCM_guarantee}: Random Correlation Maximization (\textbf{RCM}) guarantee}
\label{app:proof_of_thm_RCM_guarantee}

Random Correlation Maximization (\textbf{RCM}) learns and encodes the magnitude and direction of the model separately and then multiplies them with appropriate scaling.
Recall Lemma \ref{lem:error_encoded_magnitude} that gives a tail bound on $\bhat^2 - b^2$.
We use this throughout the analysis of Thm. \ref{thm:RCM_guarantee}.
We also state the following two lemmas from \cite{tonycai} and \cite{zhu_neurips_2014} without proofs.
The proofs can be found in the respective references.

\begin{lemma_boxed}[\cite{tonycai}]
\label{lem:distribution_of_inner_product}
Suppose that $\zv$ is uniformly distributed on the $d$-dimensional unit sphere $\Sbb^{d-1}$.
For a fixed $\xv \in \Real^d$ such that $\norm{\xv}_2 = 1$, the inner product $\rho = \inprod{\xv, \zv}$ between $\xv$ and $\zv$ has the probability density function
\[f\Paren{\rho} = \frac{1}{\sqrt{\pi}}\frac{\Gamma\Paren{\frac{d}{2}}}{\Gamma\Paren{\frac{d-1}{2}}}\Paren{1 - \rho^2}^{\frac{d-3}{2}}I\Paren{|\rho| < 1}.\]
Here, $I(\cdot)$ denotes the indicator function.
\end{lemma_boxed}
\begin{lemma_boxed}[\cite{zhu_neurips_2014}]
\label{lem:concentration_of_inner_products}
Suppose that $p = e^{d\beta}$ and $\zv_1, \ldots, \zv_p$ are independent and identically distributed with a uniform distribution on the $d$-dimensional sphere $\Sbb^{d-1}$.
For a fixed unit vector $\xv \in \Real^d$, let $\rho_i = \inprod{\xv, \zv_i}$ and $L_d = \max_{1 \leq i \leq p}\rho_i$.
Then $L_d \to \sqrt{1 - e^{-2\beta}}$ in probability as $d \to \infty$.
Furthermore, $L_d - \sqrt{1 - e^{-2\beta}} = O_P\Paren{\frac{\log d}{d}}$.
\end{lemma_boxed}
The proof of Lemma \ref{lem:concentration_of_inner_products} utilizes Lemma \ref{lem:distribution_of_inner_product}.

To prove Thm. \ref{thm:RCM_guarantee}, we alternatively prove the following Lemma \ref{lem:RCM_fixed_model_magnitude} that provides a high probability bound for \textbf{RCM} when the magnitude of the true model $\thetav$ is known.

\begin{lemma_boxed}
\label{lem:RCM_fixed_model_magnitude}
For a true parameter $\thetav \in \Real^d$ satisfying $\frac{1}{d}\norm{\thetav}_2^2 = b^2 \leq c^2$, as $d \to \infty$, the output of $\textbf{RCM}$ satisfies,
\begin{equation}
\label{eq:RCM_fixed_model_magnitude_eqn}
    \mathbb{P}\Paren{R\Paren{\thetatv, \thetav} > \frac{b^2\cdot\frac{\sigma^2}{\sigma_{\mathrm{min}}^2}}{\Paren{b^2 + \frac{\sigma^2}{\sigma_{\mathrm{min}}^2}}} + \frac{b^4\cdot 2^{-2B}}{\Paren{b^2 + \frac{\sigma^2}{\sigma_{\mathrm{min}}^2}}} + C\sqrt{\frac{\log d}{d}}} \to 0
\end{equation}
for some constant $C$ that does not depend on $d$ (but could possibly depend on $\Xv, b, \sigma$ and $B$).
\end{lemma_boxed}

The probability measure in Lemma \ref{lem:RCM_fixed_model_magnitude} is with respect to both the randomness in the source coding scheme used to obtain the estimate $\thetatv$ as well as the randomness in $\yv \in \Real^n$ due to the additive model $\vv \in \Real^n$ as in eq. \eqref{eq:noisy_planted_model}.
Note that as a direct consequence of Lemma \ref{lem:RCM_fixed_model_magnitude}, we can obtain a uniform high probability bound on all vectors in the $\ell_2$ ball $\frac{1}{d}\lVert\thetav\rVert_2^2 \leq c^2$.
In other words, we can replace $b^2$ by $c^2$ in eq. \eqref{eq:RCM_fixed_model_magnitude_eqn} of the statement of Lemma \ref{lem:RCM_fixed_model_magnitude}, and that completes the proof of Thm. \ref{thm:RCM_guarantee}.
Hence, it suffices to prove Lemma \ref{lem:RCM_fixed_model_magnitude}, which we do next.

\subsection{Proof of Lemma \ref{lem:RCM_fixed_model_magnitude}: \textbf{RCM} guarantee when the latent parameter \texorpdfstring{$\thetav$}{thetav} lies on the boundary of Euclidean ball}
\label{app:proof_lemma_RCM_fixed_model_magnitude}

As before, let us denote $\gammah \triangleq \frac{\bhat^2}{\bhat^2 + \frac{\sigma^2}{d}\xi}$.
The $\ell_2$-risk of the output of \textbf{RCM}, $\thetatv$ can be decomposed as,
\begin{align}
    \label{eq:risk_decomposition}
    R\Paren{\thetatv, \thetav} &= \frac{1}{d}\norm{\thetatv - \thetav}_2^2 \nonumber\\
    &= \frac{1}{d}\norm{\thetatv - \gammah\Xinv\yv + \gammah\Xinv\yv - \thetav}_2^2 \nonumber \\
    &= \underbrace{\frac{1}{d}\norm{\thetatv - \gammah\Xinv\yv}_2^2}_{A_1} + \underbrace{\frac{1}{d}\norm{\gammah\Xinv\yv - \thetav}_2^2}_{A_2} + \underbrace{\frac{2}{d}\inprod{\thetatv - \gammah\Xinv\yv, \gammah\Xinv\yv - \thetav}}_{A_3}.
\end{align}
The first term in eq. \eqref{eq:risk_decomposition}, labelled $A_1$, corresponds to the model quantization error and will change accordingly depending on which source coding scheme is used.
The second term labelled $A_2$ corresponds to the unavoidable error that will be present even if we have an infinite bit budget for model quantization.

\textbf{Analyzing the term $A_1$}.
Note that $A_1$ can be further decomposed as
\begin{equation}
\label{eq:decomposition_term_A1}
    \frac{1}{d}\lVert\thetatv - \gammah\Xinv\yv\rVert_2^2 = \frac{1}{d}\lVert\thetatv\rVert_2^2 + \frac{1}{d}\gammah^2\norm{\Xinv\yv}_2^2 - \frac{2}{d}\inprod{\thetatv, \gammah\Xinv\yv}
\end{equation}

The first term in \eqref{eq:decomposition_term_A1} is:
\begin{equation*}
    \frac{1}{d}\lVert\thetatv\rVert_2^2 = \frac{\bt^4\Paren{1 - 2^{-2B}}}{\bt^2 + \frac{\sigma^2}{d}\xi}\lVert\stv\rVert_2^2 = \frac{b^4\Paren{1 - 2^{-2B}}}{b^2 + \frac{\sigma^2}{d}\xi} + O_P\Paren{\frac{1}{\sqrt{d}}}.
\end{equation*}
Here, we utilize the fact that $\lVert\stv\rVert_2 = 1$ since $\xtv \in \Sbb^{d-1}$, and the second equality follows from Lemma \ref{lem:error_encoded_magnitude} along with the fact that $\bt^2$ is obtained from $\bhat^2$ by doing nearest neighbor quantization with respect to the codebook $\Bcal$ whose elements are separated with a resolution of $\frac{1}{\sqrt{d}}$.
The second term in \eqref{eq:decomposition_term_A1} is,
\begin{equation*}
    \frac{1}{d}\gammah^2\norm{\Xinv\yv}_2^2 = \frac{\bhat^4}{\Paren{\bhat^2 + \frac{\sigma^2}{d}\xi}^2}\cdot\Paren{\bhat^2 + \frac{\sigma^2}{d}\xi} = \frac{b^4}{b^2 + \frac{\sigma^2}{d}\xi} + O_P\Paren{\frac{1}{\sqrt{d}}}.
\end{equation*}

Whereas, the third (cross) term can be seen to be,
\begin{align*}
    \frac{2}{d}\inprod{\thetatv, \gammah\Xinv\yv} &= \frac{2}{d}\frac{\bhat^2}{\Paren{\bhat^2 + \frac{\sigma^2}{d}\xi}}\sqrt{\frac{d\bt^4\Paren{1 - 2^{-2B}}}{\Paren{\bt^2 + \frac{\sigma^2}{d}\xi}}}\cdot \inprod{\Xinv\yv, \stv} \\
    &=2\frac{\bhat^2\bt^2\sqrt{1 - 2^{-2B}}}{\Paren{\bhat^2 + \frac{\sigma^2}{d}\xi}\sqrt{\bt^2 + \frac{\sigma^2}{d}\xi}}\frac{1}{\sqrt{d}}\norm{\Xinv\yv}\cdot\frac{\inprod{\Xinv\yv, \stv}}{\norm{\Xinv\yv}} \\
    &= 2\frac{\bhat^2\bt^2\sqrt{1 - 2^{-2B}}}{\sqrt{\Paren{\bhat^2 + \frac{\sigma^2}{d}\xi}\Paren{\bt^2 + \frac{\sigma^2}{d}\xi}}}\cdot\frac{\inprod{\Xinv\yv, \stv}}{\norm{\Xinv\yv}}
\end{align*}
From Lemma \ref{lem:concentration_of_inner_products}, $\frac{\inprod{\Xinv\yv,\stv}}{\norm{\Xinv\yv}} - \sqrt{1 - 2^{-2B}} = O_P\Paren{\frac{\log d}{d}}$.

So, $\frac{2}{d}\inprod{\thetatv, \gammah \Xinv\yv} = \frac{2b^4}{b^2 + \frac{\sigma^2}{d}\xi}\Paren{1 - 2^{-2B}} + O_P\Paren{\frac{1}{\sqrt{d}}}$.
Combining all the three terms together yields,
\begin{equation}
\label{eq:A1_simplification}
    A_1 = \frac{b^4\cdot 2^{-2B}}{\Paren{b^2 + \frac{\sigma^2}{d}\xi}} + O_P\Paren{\frac{1}{\sqrt{d}}}.
\end{equation}

\textbf{Analyzing the term $A_2$}.
Term $A_2$ has nothing to do with the quantization scheme used for model compression and can be analyzed similar to eq. \eqref{eq:naive_A2_simplification} in App. \ref{app:naive_quantizer_guarantee}, i.e.
\begin{equation*}
    A_2 = \frac{b^2\cdot\frac{\sigma^2\xi}{d}}{\Paren{b^2 + \frac{\sigma^2}{d}\xi}} + O_P\Paren{\frac{1}{\sqrt{d}}}.
\end{equation*}

\textbf{Analyzing the term $A_3$}.
Let us define,
\begin{equation}
\label{eq:auxiliary_variable_1_RCM}
    \thetaov \triangleq \sqrt{\frac{d\bhat^4\Paren{1 - 2^{-2B}}}{\bhat^2 + \frac{\sigma^2}{d}\xi}}\cdot\stv.
\end{equation}
The term $A_3$ can be written as,

\begin{align}
\label{eq:term_A3_decomposition}
    \frac{2}{d}\inprod{\thetatv - \gammah\Xinv\yv, \gammah\Xinv\yv - \thetav} = \frac{2}{d}\inprod{\thetatv - \thetaov, \gammah\Xinv\yv - \thetav} &+ \frac{2}{d}\inprod{\thetaov, \gammah\Xinv\yv - \thetav}\\
    &- \frac{2}{d}\inprod{\gammah\Xinv\yv, \gammah\Xinv\yv - \thetav}.
\end{align}

The first term in eq. \eqref{eq:term_A3_decomposition} is $O_P\Paren{\frac{1}{\sqrt{d}}}$, since $\frac{1}{\sqrt{d}}\lVert\thetatv - \thetaov\rVert_2 = O_P\Paren{\frac{1}{\sqrt{d}}}$, and $\frac{1}{\sqrt{d}}\norm{\gammah\Xinv\yv - \thetav}_2$ has bounded second moment.
The third term in eq. \eqref{eq:term_A3_decomposition} is,

\begin{align*}
    \frac{2}{d}\inprod{\gammah\Xinv\yv, \gammah\Xinv\yv - \thetav} = \frac{2}{d}\inprod{\gammah\ytv, \gammah\ytv - \tauv} = \frac{2}{d}\gammah^2\norm{\ytv}_2^2 - 2\gammah b\frac{\yt_1}{\sqrt{d}} = O_P\Paren{\frac{1}{\sqrt{d}}}.
\end{align*}

We now consider the second term,

\begin{equation}
\label{eq:projection_error_1}
    \frac{2}{d}\inprod{\thetaov, \gammah\Xinv\yv - \thetav} = \frac{2}{d}\inprod{\thetaov, \gamma'\Xinv\yv - \thetav} + \frac{2}{d}\Paren{\gammah - \gamma'}\inprod{\thetaov, \Xinv\yv},
\end{equation}

where $\gamma' \triangleq \frac{\inprod{\thetav, \Xinv\yv}}{\norm{\Xinv\yv}_2^2}$.
Note that this choice of $\gamma'$ ensures that $\gamma'\Xinv\yv$ is the projection of the true model $\thetav$ on to $\Xinv\yv$, since the error is orthogonal, i.e. $\inprod{\Xinv\yv, \gamma'\Xinv\yv - \thetav} = 0$.
Furthermore, since $\gammah\Xinv\yv$ can be interpreted as the linear estimator of $\thetav$ from $\yv$ given $\Xv$, which uses the unquantized magnitude $\gammah$ for scaling, the second term in eq. \eqref{eq:projection_error_1} quantifies the distance between this vector and the projection of $\thetav$ on $\Xinv\yv$.
We now show that $\gammah\Xinv\yv$ is not too far from this projection because,

\begin{align*}
    \gamma' = \frac{\sum_{i=1}^{d}\theta_i\Paren{\Xinv\yv}_i}{\norm{\Xinv\yv}_2^2} = \frac{\norm{\thetav}_2^2 + \sum_{i=1}^{d}\theta_i\Paren{\Xinv\vv}_i}{d\Paren{\bhat^2 + \frac{\sigma^2}{d}\xi}} = \frac{b^2}{\bhat^2 + \frac{\sigma^2}{d}\xi} + \frac{1}{\bhat^2 + \frac{\sigma^2}{d}\xi}\cdot\frac{1}{d}\sum_{i=1}^{d}\theta_i\Paren{\Xinv\vv}_i
\end{align*}

Moreover, since we considered the SVD of $\Xv = \Uv\Sigmav\Vv^{\top}$, the SVD of $\Xinv$ is, $\Xinv\vv = \Vv\Sigmav'\Uv^{\top}\vv$, where $\Sigmav'$ consists of singular values $\sigma_1^{-1}, \ldots, \sigma_d^{-1}$.
Since Gaussian distribution is rotation invariant, and $\Vv^{\top}$ is a unitary transform, $\vv \sim \Ncal(\mathbf{0}, \sigma^2\Iv_n)$ implies that $\Vv^{\top}\vv \sim \Ncal(\mathbf{0}, \sigma^2\Iv_n)$.
Furthermore, since this is followed by multiplication by a diagonal matrix, which  simply scales the variance of each coordinate differently, and another unitary transform $\Uv$, it holds that, $\Paren{\Xinv\vv}_i \sim \Ncal\Paren{0, \frac{\sigma^2}{\sigma_i^2}}$.
Consequently, $\sum_{i=1}^{d}\theta_i\Paren{\Xinv\vv}_i \sim \Ncal\Paren{0, \sum_{i=1}^{d}\theta_i^2\frac{\sigma^2}{\sigma_i^2}} \implies$ The variance of $\frac{1}{d}\sum_{i=1}^{d}\theta_i\Paren{\Xinv\vv}_i$ can be upper bounded by $\frac{\sigma^2b^2}{d\sigma_{\mathrm{min}}^2}$.
So, $\frac{1}{d}\sum_{i=1}^{d}\theta_i\Paren{\Xinv\vv}_i$ concentrates around its mean $0$ as $O_P\Paren{\frac{1}{\sqrt{d}}}$, meaning, $\gammah - \gamma' = \frac{\bhat^2 - b^2}{\bhat^2 + \frac{\sigma^2}{d}\xi} = O_P\Paren{\frac{1}{\sqrt{d}}}$.
Moreover, the coefficient of $(\gammah - \gamma')$,
\begin{equation*}
    \frac{2}{d}\inprod{\thetaov, \Xinv\yv} \leq \frac{2}{d}\lVert\thetaov\rVert_2\norm{\Xinv\yv}_2 = \frac{2}{d} \sqrt{\frac{d\bhat^4\Paren{1 - 2^{-2B}}}{\bhat^2 + \frac{\sigma^2}{d}\xi}}\sqrt{d\Paren{\bhat^2 + \frac{\sigma^2}{d}\xi}} = 2\bhat^2\Paren{1 - 2^{-2B}}
\end{equation*}
is finite.
Therefore, the second term in eq. \eqref{eq:projection_error_1} is $O_P\Paren{\frac{1}{\sqrt{d}}}$.
The final term left to analyze is the first term in eq. \eqref{eq:projection_error_1},
\begin{equation*}
    \frac{2}{d}\inprod{\thetaov, \gamma'\Xinv\yv - \thetav} = \sqrt{\frac{4\bhat^4\Paren{1 - 2^{-2B}}}{\bhat^2 + \frac{\sigma^2}{d}\xi}}\inprod{\stv, \frac{1}{\sqrt{d}}\Paren{\gamma'\Xinv\yv - \thetav}}
\end{equation*}
Note that the scaling factor $\sqrt{\frac{4\bhat^4\Paren{1 - 2^{-2B}}}{\bhat^2 + \frac{\sigma^2}{d}\xi}} = \sqrt{\frac{4b^4\Paren{1 - 2^{-2B}}}{b^2 + \frac{\sigma^2}{d}\xi}} + O_P\Paren{\frac{1}{\sqrt{d}}}$.
To analyze the term $\inprod{\stv, \frac{1}{\sqrt{d}}\Paren{\gamma'\Xinv\yv - \thetav}}$, we make use of the following result from \cite{cai_fan_jiang_JMLR_2013}.

\begin{lemma_boxed}[]Near-orthogonality of vectors in high dimensions; \cite[Prop. 1]{cai_fan_jiang_JMLR_2013}, \cite[Lemma A.4]{zhu_neurips_2014}]
\label{lem:angle_distribution_unit_sphere}
Let $\uv$ be uniformly distributed on the unit sphere $\Sbb^{d-1}$.
Let $\xv \in \Real^d$ be a fixed unit vector.
Then,
\begin{equation*}
    \mathbb{P}\Paren{\left\lvert\inprod{\uv, \xv}\right\rvert > \epsilon} \leq K\sqrt{d}\Paren{1 - \epsilon^2}^{\frac{d-2}{2}}
\end{equation*}
for all $d \geq 2$ and $0 < \epsilon < 1$, where $K$ is a universal constant.
Therefore,
\begin{equation*}
    \inprod{\uv,\xv} = O_P\Paren{\frac{\log d}{d}}.
\end{equation*}
\end{lemma_boxed}

From the construction of our codebook $\Ycal$, the vector $\stv$ is uniformly distributed on the unit sphere $\Sbb^{d-1}$.
Consequently, for any vector $\yv \in \Real^d$, from symmetry, the projection of $\stv$ on to subspace orthogonal $\Xinv\yv$, i.e. $\mathrm{Proj}_{\Xinv\yv^\perp}\Paren{\stv}$ has a spherical distribution in $\Real^{d-1}$.
Furthermore, its length can be calculated as follows,
\begin{align*}
    \mathrm{Proj}_{\Xinv\yv^\perp}\Paren{\stv} &= \stv - \frac{\inprod{\stv,\Xinv\yv}}{\norm{\Xinv\yv}_2^2}\Xinv\yv \\
    \implies \norm{\mathrm{Proj}_{\Xinv\yv^\perp}\Paren{\stv}}_2^2 &= \norm{\stv}_2^2 + \frac{\inprod{\stv, \Xinv\yv}^2}{\norm{\Xinv\yv}_2^2} - 2\cdot\frac{\inprod{\stv, \Xinv\yv}^2}{\norm{\Xinv\yv}_2^2} \\
    &= 1 - \frac{\inprod{\stv, \Xinv\yv}^2}{\norm{\Xinv\yv}_2^2} \\
    &= 1 - \inprod{\stv, \frac{\Xinv\yv}{\norm{\Xinv\yv}_2}}^2 = 1 - \sqrt{1 - 2^{-2B}} + O_P\Paren{\frac{\log d}{d}}.
\end{align*}
The final equality is due to Lemma \ref{lem:concentration_of_inner_products}.
So, we can write, with high probability, $ \mathrm{Proj}_{\Xinv\yv^\perp}\Paren{\stv} = L\cdot \uv$ where, $\uv$ follows the uniform distribution on sphere $\Sbb^{d-1}$ and $L = 1 - \sqrt{1 - 2^{-2B}} + O_P\Paren{\frac{\log d}{d}}$.
Note that the projection of $\stv$ onto the orthogonal space of $\Xinv\yv$, i.e. $\mathrm{Proj}_{\Xinv\yv^\perp}\Paren{\stv}$, is independent of $\Xinv\yv$.
So, conditioning on a particular value of $\Xinv\yv = \wv \in \Real^d$, and recalling that $\inprod{\Xinv\yv, \gamma'\Xinv\yv - \thetav} = 0$, we have,

\begin{align*}
    &\mathbb{P}\Paren{\inprod{\stv, \frac{1}{\sqrt{d}}\Paren{\gamma'\Xinv\yv - \thetav}} > t \hspace{1mm} \Bigg\vert \hspace{1mm} \Xinv\yv = \wv} \\
    &= \mathbb{P}\Paren{\inprod{\mathrm{Proj}_{\Xinv\yv^\perp}\Paren{\stv}, \frac{1}{\sqrt{d}}\mathrm{Proj}_{\Xinv\yv^\perp}\Paren{\gamma'\Xinv\yv - \thetav}} > t \hspace{1mm} \Bigg\vert \hspace{1mm} \Xinv\yv = \wv} \\
    &= \mathbb{P}\Paren{L\norm{\frac{1}{\sqrt{d}}\Paren{\gamma'\wv - \thetav}}_2\inprod{\uv, \ev} > t} \\
    &\leq K_1\sqrt{d}\Paren{1 - \frac{t^2}{K_2\norm{\frac{1}{\sqrt{d}}\Paren{\gamma'\wv- \thetav}}_2^2}}^{\frac{d-2}{2}},
\end{align*}

where, $K_1$ and $K_2$ are positive constants.
Here, $\ev$ is the unit vector in the direction of $(\gamma'\wv -\thetav)$, and the last inequality follows from Lemma \ref{lem:angle_distribution_unit_sphere}.
We then have,

\begin{align*}
    &\mathbb{P}\Paren{\inprod{\stv, \frac{1}{\sqrt{d}}\Paren{\gammah\Xinv\yv - \thetav}} > t}\\
    &= \bigintsss\mathbb{P}\Paren{\inprod{\stv, \frac{1}{\sqrt{d}}\Paren{\gamma'\Xinv\yv - \thetav}} > t \hspace{1mm} \Bigg\vert \hspace{1mm} \Xinv\yv = \wv}p_{\Xinv\yv}\Paren{\Xinv\yv = \wv}d\wv \\
    &\leq K_1\sqrt{d}\bigintsss\Paren{1-\frac{t^2}{K_2\norm{\frac{1}{\sqrt{d}}\Paren{\gamma'\Xinv\yv - \thetav}}_2^2}}^{\frac{d-2}{2}}p_{\Xinv\yv}\Paren{\Xinv\yv = \wv}d\wv \\
    &\leq K_1\sqrt{d}\Paren{1 - \frac{t^2}{K_2'}}^{\frac{d-2}{2}} + \mathbb{P}\Paren{\frac{1}{d}\norm{\Xinv\yv}_2^2 > b^2 + \frac{\sigma^2}{d}\xi + K_3},
\end{align*}

for positive constants $K_1, K_2$ and $K_3$.
This implies $\inprod{\stv, \frac{1}{\sqrt{d}}\Paren{\gamma'\Xinv\yv - \thetav}} = O_P\Paren{\sqrt{\frac{\log d}{d}}}$ and thus $A_3 = O_P\Paren{\sqrt{\frac{\log d}{d}}}$.

\textbf{Completing the proof of Lemma \ref{lem:RCM_fixed_model_magnitude}}.
From eq. \eqref{eq:risk_decomposition}, we have that $R\Paren{\thetatv, \thetav} = A_1 + A_2 + A_3$.
From eq. \eqref{eq:A1_simplification}, $A_1 = \frac{b^4\cdot 2^{-2B}}{\Paren{b^2 + \frac{\sigma^2}{d}\xi}} + O_P\Paren{\frac{1}{\sqrt{d}}}$.
From eq. \eqref{eq:naive_A2_simplification}, $A_2 = \frac{b^2\cdot\frac{\sigma^2\xi}{d}}{\Paren{b^2 + \frac{\sigma^2}{d}\xi}} + O_P\Paren{\frac{1}{\sqrt{d}}}$.
Finally, as just seen above, $A_3 = O_P\Paren{\sqrt{\frac{\log d}{d}}}$.
Therefore, for Random Correlation Maximization (\textbf{RCM}), the risk is given by,
\begin{equation}
    R\Paren{\thetatv, \thetav} = \frac{b^2\cdot\frac{\sigma^2\xi}{d}}{\Paren{b^2 + \frac{\sigma^2}{d}\xi}} + \frac{b^4\cdot 2^{-2B}}{\Paren{b^2 + \frac{\sigma^2}{d}\xi}} + O_P\Paren{\sqrt{\frac{\log d}{d}}}
\end{equation}
In other words,

\begin{equation}
    \mathbb{P}\Paren{R\Paren{\thetatv, \thetav} > \frac{b^2\cdot\frac{\sigma^2\xi}{d}}{\Paren{b^2 + \frac{\sigma^2}{d}\xi}} + \frac{b^4\cdot 2^{-2B}}{\Paren{b^2 + \frac{\sigma^2}{d}\xi}} + C\sqrt{\frac{\log d}{d}}} \to 0
\end{equation}

as $d \to \infty$, for some constant $C$ independent of $d$.
This also implies,

\begin{equation}
    \mathbb{P}\Paren{{R\Paren{\thetatv, \thetav} > \frac{b^2\cdot\frac{\sigma^2}{\sigma_{\mathrm{min}}^2}}{\Paren{b^2 + \frac{\sigma^2}{\sigma_{\mathrm{min}}^2}}} + \frac{b^4\cdot 2^{-2B}}{\Paren{b^2 + \frac{\sigma^2}{\sigma_{\mathrm{max}}^2}}} + C\sqrt{\frac{\log d}{d}}}} \to 0
\end{equation}
which completes the proof.

\section{Democratic Embeddings and Quantization}
\label{app:democratic_embeddings}

For a given wide matrix $\Sv \in \Real^{d \times D}$ where $d \leq D$, and for a given vector $\zv \in \Real^n$, the system of equations $\zv = \Sv \xv$ is under-determined for $\xv \in \Real^D$, with the set $\Xcal = \{\xv \in \Real^D \hspace{1mm} \vert \hspace{1mm} \zv = \Sv\xv\}$ as the solution space.
The \textbf{democratic embedding} of $\zv$ with respect to $\Sv$ is defined to be the solution of the following linear program,
\begin{equation}
\label{eq:democratic_embedding_definition}
    \xv_d = \argminimize_{\xv \in \Real^D}\lVert \xv \rVert_{\infty} \text{ subject to } \zv = \Sv \xv.
\end{equation}
In order to characterize the solution of the optimization problem \eqref{eq:democratic_embedding_definition} (done in Lemma \ref{lem:democratic_embedding_property}), we review certain definitions from \cite{lyubarskii, studer2012}.

\begin{definition_boxed} \textbf{(Frame)}
\label{def:frame}
A matrix $\Sv \in \Real^{d \times D}$ with $d \leq D$ is called a \textbf{frame} if there exist $A,B$ with $0 < A \leq B < \infty$, such that for any vector $\zv \in \Real^d$, we have $A\norm{\zv}_2^2 \leq \norm{\Sv^\top\zv}_2^2 \leq B\norm{\zv}_2^2$, where $A$ and $B$ are called \textbf{lower} and \textbf{upper frame bounds} respectively.
\end{definition_boxed}

Note that the definition of a \textit{frame} includes all full-rank matrices \cite{frames_book, wikipedia_frame}.
However, for obtaining democratic embeddings with desirable properties (as specified in Lemma \ref{lem:democratic_embedding_property} below), we are interested in frames that have a small value of the upper frame bound $B$.
Explicit expressions for $B$ that hold true with high probability, can be obtained for special classes of frames that satisfy the uncertainty principle defined below in Def. \ref{def:uncertainty_principle}.

\begin{definition_boxed}\textbf{(Uncertainty principle (UP))}
\label{def:uncertainty_principle}
A frame $\Sv \in \Real^{d \times D}$ satisfies the Uncertainty Principle with parameters $\eta, \delta$, with $\eta > 0, \delta \in (0,1)$ if $\norm{\Sv\xv}_2 \leq \eta\norm{\xv}_2$ holds for all (sparse) vectors $\xv \in \Real^D$ satisfying $\norm{\xv}_0 \leq \delta D$, where $\norm{\xv}_0$ denotes the number of non-zero elements in $\xv$.
\end{definition_boxed}

\begin{lemma_boxed}[\textbf{Democratic Embeddings}; \cite{studer2012}]
\label{lem:democratic_embedding_property}
Let $\Sv \in \Real^{d \times D}$ be a frame with bounds $A,B$ (cf. Def. \ref{def:frame}) that satisfies  the uncertainty principle (UP) (cf. Def. \ref{def:uncertainty_principle}) with parameters $\eta, \delta$ such that $A > \eta\sqrt{B}$.
Then for any vector $\zv \in \Real^d$, the solution $\xv_d$ of \eqref{eq:l_inf_minimization_problem} satisfies
\begin{equation}
    \frac{K_l}{\sqrt{D}}\norm{\zv}_2 \leq \norm{\xv_d}_{\infty} \leq \frac{K_u}{\sqrt{D}}\norm{\zv}_2,
\end{equation}
where $ K_l = \frac{1}{\sqrt{B}}$ and $K_u = \frac{\eta}{\left(A-\eta\sqrt{B}\right)\sqrt{\delta}}$ are called \textbf{lower} and \textbf{upper Kashin constants} respectively.
\end{lemma_boxed}

\textbf{Remark}. The linear program \eqref{eq:democratic_embedding_definition} can be solved with $O(d^3)$ multiplications using simplex or Newton's method.
Alternatively, \cite{lyubarskii} propose an iterative projected descent type algorithm with $O(d^2)$ complexity that makes explicit use of the values of $\eta$ and $\delta$.

In this work, we consider Parseval frames $(A = B = 1)$, i.e. they satisfy $\Sv\Sv^{\top} = \Iv_n$, implying $K_l = 1$ and $K_u = \eta(1-\eta)^{-1}\delta^{-1/2}$.
Lemma \ref{lem:democratic_embedding_property} shows that none of the coordinates of the democratic embedding is too large, and the information content of $\zv$ is distributed evenly across the coordinates of its embedding.
The value of upper Kashin constant $K_u$ depends on the choice of frame construction $\Sv$, as well as its aspect ratio $\lambda = D/d$.
\cite{lyubarskii, studer2012} show that if $\Sv$ is a \textit{random orthonormal matrix}, then $K = K\left(\lambda\right)$, where $\lambda > 1$ can be arbitrarily close to $1$.
This is precisely characterized below.

\begin{theorem_boxed}[\textbf{UP for Random Orthonormal Matrices}; \cite{lyubarskii}(Thm. 4.1)]
\label{thm:uncertainty_principle_random_orthonormal}
Let $\mu > 0$ and $D = (1 + \mu)d$.
Then with probability at least $1 - 2\exp\left(-c\mu^2d\right)$, a random orthonormal $d \times D$ matrix $\Sv$ satisfies the uncertainty principle with parameters
\begin{equation}
    \eta = 1 - \frac{\mu}{4} \hspace{2mm} \text{and}, \hspace{2mm} \delta = \frac{c\mu^2}{\log(1/\mu)}
\end{equation}
where $c > 0$ is an absolute constant.
\end{theorem_boxed}

Since choosing $\lambda$ is up to us, Lemma \ref{lem:democratic_embedding_property} implies that for random orthonormal frames, democratic embeddings satisfy $\norm{\xv_d}_{\infty} = \Theta(1/\sqrt{D})$.
Lemma \ref{lem:democratic_quantizer_worst_case_error} utilizes this fact to give the worst-case quantization error of \textbf{Democratic Quantization} as described in \eqref{eq:democratic_quantization_shape}.

\clearpage
\begin{lemma_boxed}[\textbf{Worst-case $\ell_2$ error of democratic quantization}]
\label{lem:democratic_quantizer_worst_case_error}
Given a frame $\Sv \in \Real^{d \times D}$ with upper Kashin constant $K_u$ and a uniform scalar quantizer $\text{Q}_{u,B}$, for any $\zv \in \Real^d$ such that $\norm{\zv}_2 \leq 1$, let $\ztv = \text{Q}_{d,B}(\zv) = \Sv\cdot\text{Q}_{u,B}(\xv_d)$.
Then, for $\lambda = D/d$,
\begin{equation}
    \sup_{\lVert \zv \rVert_2 \leq 1} \lVert \ztv - \zv \rVert_2 \leq 2^{\Paren{1 - \frac{B}{\lambda}}}K_u.
\end{equation}
\end{lemma_boxed}
\begin{proof}
From Lemma \ref{lem:democratic_embedding_property}, we know that the democratic representation of $\zv$ with respect to $\Sv$, i.e. $\xv_d$, satisfies, $\norm{\xv_d}_{\infty} \leq \frac{K_u}{\sqrt{D}}$.
Let us denote $\xtv = \text{Q}_{u,B}(\xv_d)$.
Then, since $\xv_d, \xtv \in \Real^D$, we have $\norm{\xtv - \xv_d}_2 \leq \frac{2K_u/\sqrt{D}}{2^{dB/D}}\sqrt{D} = 2^{\Paren{1 - \frac{dB}{D}}}K_u$.
This inequality follows since we can exploit Lemma \ref{lem:democratic_embedding_property} and choose the dynamic range of the uniform quantizer $\text{Q}_{u,B}$ to be $\left[-\frac{K_u}{\sqrt{D}}, +\frac{K_u}{\sqrt{D}}\right]$.
Moreover, since we originally had a bit-budget of $dB$-bits but we are now using it to quantize a vector in $\Real^D$, the effective average number of bits per dimension is $dB/D$, implying $2^{dB/D}$ quantization points per dimension.

Moreover, $\norm{\ztv - \zv}_2 \leq \norm{\Sv\Paren{\xtv - \xv_d}}_2 \leq \norm{\Sv}_2\norm{\xtv - \xv_d}_2 \leq \norm{\xtv - \xv_d}_2 \leq 2^{\Paren{1 - \frac{dB}{D}}}K_u$.
The third inequality holds true because we consider Parseval frames and $\Sv\Sv^{\top} = 1 \implies \norm{\Sv}_2 \leq 1$.
This completes the proof.
\end{proof}

\section{Proof of Theorem \ref{thm:DQ_guarantee}: Democratic Quantization (DQ) guarantee}
\label{app:DQ_guarantee_proof}

The proof of this is similar to the proof of Thm. \ref{thm:naive_quantizer_guarantee}.
Recall that we defined the scaling factor as $\gammat \triangleq \sqrt{\frac{d\bt^4}{\bt^2 + \frac{\sigma^2}{d}\xi}}$ in eq. \eqref{eq:gamma_dash_scaling}.
Then, the output of \text{Democratic Model Quantization} is given by $\thetatv = \gammat\stv$, where $\stv$ is the the democratically quantized direction $Q_{d,B}\Paren{\frac{\Xinv\yv}{\lVert \Xinv\yv \rVert_2}}$.
We refer to Lemma \ref{lem:democratic_quantizer_worst_case_error} from App. \ref{app:democratic_embeddings} that gives the worst-case quantization error for democratic quantizer whose input is a vector on the unit sphere $\Sbb^{d-1}$.
Since we compute the democratic embedding of $\sv = \Xinv\yv/\lVert \Xinv\yv \rVert_2$ that satisfies $\lVert \sv \rVert_2 = 1$, Lemma \ref{lem:democratic_quantizer_worst_case_error} and Thm. \ref{thm:uncertainty_principle_random_orthonormal} together imply that when $\Sv \in \Real^{d \times D}$ is taken to be a \textit{random orthonormal frame}, with probability at least $1 - 2\exp\left(-c\mu^2d\right)$, where $\mu = \frac{D}{d} - 1$, we have,
\begin{equation}
    \left\lVert \stv - \text{Q}_{d,B}\Paren{\frac{\Xinv\yv}{\lVert \Xinv\yv\rVert_2}} \right\rVert_2 \leq 2^{\Paren{1 - \frac{B}{\lambda}}}K_u,
\end{equation}
where $\lambda = D/d$ is the aspect ratio of $\Sv$.
Similar to App. \ref{app:DQ_guarantee_proof}, the $\ell_2$-risk can be decomposed as,
\begin{equation}
\label{eq:risk_decomposition_DQ}
    R\Paren{\thetatv, \thetav} = \frac{1}{d}\lVert\thetatv - \thetav\rVert_2^2 \leq \underbrace{\frac{2}{d}\lVert\thetatv - \gammah\Xinv\yv\rVert_2^2}_{A_1} + \underbrace{\frac{2}{d}\norm{\gammah\Xinv\yv - \thetav}_2^2}_{A_2},
\end{equation}
where $\gammah = \frac{\bhat^2}{\bhat^2 + \frac{\sigma^2}{d}\xi}$.
The term $A_2$ is the inherent learning risk and remains unchanged.

To analyze term $A_1$, we decompose it in the same way as before in eq. \eqref{eq:naive_term_A1_decomposition}.
Similar as before, the term $\frac{2}{d}\gammah^2\norm{\Xinv\yv}_2^2\cdot\norm{\stv - \frac{\Xinv\yv}{\norm{\Xinv\yv}_2}}_2^2 + \frac{2}{d}\Paren{\gammat - \gammah\norm{\Xinv\yv}_2}^2\norm{\stv}_2^2$ can be shown to be $O_P\Paren{\frac{1}{\sqrt{d}}}$.
However, we analyze the first term in eq. \eqref{eq:naive_term_A1_decomposition} as follows,
\begin{align*}
    \frac{2}{d}\gammah^2\lVert \Xinv\yv \rVert_2^2 \cdot\left\lVert \stv - \frac{\Xinv\yv}{\lVert \Xinv\yv \rVert_2} \right\rVert_2^2 &\leq \frac{2}{d}\cdot \frac{d\bhat^4}{\Paren{\bhat^2 + \frac{\sigma^2}{d}\xi}}\cdot 2^{2\Paren{1 - \frac{B}{\lambda}}}K_u^2 + O_P\Paren{e^{-d}} \\
    &\leq \frac{8K_u^2b^4}{b^2 + \frac{\sigma^2}{d}\xi}\cdot 2^{-\frac{2B}{\lambda}} + O_P\Paren{\frac{1}{\sqrt{d}}}.
\end{align*}
Proceeding similarly as before, we can get a uniform upper bound over all models in $\frac{1}{d}\lVert \thetav \rVert_2^2 \leq c^2$ as,

\begin{equation}
    \mathbb{P}\Paren{R(\thetatv, \thetav) > \frac{2c^2\sigma^2}{\sigma^2 + c^2\sigma_{\mathrm{min}}^2} + \frac{16K_u^2c^4\sigma_{\mathrm{min}}^2}{\sigma^2 + c^2\sigma_{\mathrm{min}}^2}\cdot 2^{-\frac{2B}{\lambda}} + \frac{C_{dq}}{\sqrt{d}}} \to 0,
\end{equation}

as $d \to \infty$, for some constant $C_{dq}$ independent of the dimension $d$.

\section{Near-Democratic Embeddings and Quantization}
\label{app:near_democratic_embeddings}

To reduce the computational complexity of obtaining the democratic embeddings, an alternative would be to solve a relaxation of \eqref{eq:l_inf_minimization_problem}.
The \textbf{near-democratic embedding} of $\zv \in \Real^d$ with respect to $\Sv \in \Real^{d \times D}$ is defined to be the solution of the following convex program,

\begin{equation}
    \label{eq:near_democratic_embedding_definition}
    \xv_{nd} = \argminimize_{\xv \in \Real^D}\lVert \xv \rVert_2 \text{ subject to } \zv = \Sv \xv.
\end{equation}

To solve this in closed form, the Lagrangian is $L(\xv, \nuv) = \xv^{\top}\xv + \nuv^{\top}\left(\Sv \xv - \yv\right)$ for $\nuv \in \Real^d$.
Computing the derivative and setting it to zero, $\nabla_{\xv}L(\xv, \nuv) = 2\xv + \Sv^{\top}\nuv = 0 \implies \xv_{nd} = -\frac{1}{2}\Sv^{\top}\nuv$.
So, $\nuv = -2\left(\Sv\Sv^{\top}\right)^{-1}\zv$ and $\xv_{nd} = -\frac{1}{2}\Sv^{\top}\left(-2\left(\Sv\Sv^{\top}\right)^{-1}\zv\right) = \Sv^{\top}\left(\Sv\Sv^{\top}\right)^{-1}\zv = \Sv^{\top}\zv$.
The last equality follows from the fact that we choose our frames $\Sv \in \Real^{d \times D}$ to be Parseval, i.e. they satisfy $\Sv\Sv^{\top} = \Iv_d$.

\begin{remark}
The frame $\Sv$ can be chosen to be either a random Haar orthonormal matrix as in Thm. \ref{thm:uncertainty_principle_random_orthonormal}, or the Randomized Hadamard frame $\Sv = \Pv \Dv \Hv$ as described in \S \ref{subsec:near_democratic_quantization}.
Choosing $\Sv$ to be random orthonormal still requires $O(d^2)$ multiplications to find the near-democratic embedding but does not require the explicit knowledge of the UP parameters $\eta$ and $\delta$.
On the other hand, when $\Sv$ is taken to be randomized Hadamard, $\xv_{nd}$ can be computed using only $O(d \log d)$ additions and sign flips using the fast Walsh-Hadamard transform.
\end{remark}

The following lemmas \ref{lemma:near_democratic_dynamic_range_random_orthonormal} and \ref{lemma:near_democratic_dynamic_range_randomized_Hadamard} provide an upper bound to the $\ell_{\infty}$-norm of the solution of \eqref{eq:near_democratic_embedding_definition} for orthonormal and Hadamard frames respectively.
Their proofs are reproduced here from \cite{DSC_early_access} for completeness.

\begin{lemma_boxed}
\label{lemma:near_democratic_dynamic_range_random_orthonormal}
For a random orthonormal frame $\Sv \in \Real^{d \times D}$, with probability as least $1 - \frac{1}{2D}$, the near-democratic embedding of $\zv$ with respect to $\Sv$ satisfies,
\begin{equation}
    \norm{\xv_{nd}}_{\infty} \leq 2\sqrt{\frac{\lambda \log(2D)}{D}}\norm{\zv}_2,
\end{equation}
where $\lambda = D/d$.
\end{lemma_boxed}
\begin{proof}
A random orthonormal matrix is obtained by first generating a random Gaussian matrix $\Gv$, each of whose entries are drawn i.i.d. from $\Ncal(0,1)$, computing its singular value decomposition $\Gv = \Uvu \Sigmauv \Vvu^{\top}$, letting $\Stv = \Uvu \Vvu^{\top}$, and randomly selecting $d$ rows of $\Stv$, i.e. $\Sv = \Pv\Stv$.
Let $\{\sv_i\}_{i=1}^{D} \in \Real^d$ and $\{\stv_i\}_{i=1}^{D} \in \Real^D$ denote the columns of $\Sv$ and $\Stv$ respectively.
We have, $\norm{\sv_i}_2 = \norm{\Pv\stv_i}_2 \leq \norm{\stv_i}_2 = 1$.
For any $\zv \in \Real^d$ and $\zhv = \zv / \norm{\zv}_2$,

\[\norm{\Sv^{\top}\zv}_{\infty} = \max_{i \in [D]}|\sv_i^{\top}\zv| = \norm{\zv}_2\max_{i \in [D]}\norm{\sv_i}_2|\shv_i^{\top}\zhv| \leq \norm{\zv}_2 \max_{i \in [D]}|\shv_i^{\top}\zhv|,\]
where $\shv_i = \sv_i/\norm{\sv_i}_2$.
Note that $\shv_i \in \Real^d$ is uniformly random on the unit sphere in $\Real^d$, i.e. $\shv_i$ has identical distribution as $\gv/\norm{\gv}_2$ where $\gv \sim \Ncal(\mathbf{0}, \Iv_d)$.
Due to rotational invariance of Gaussian distribution, for any fixed $\zhv \in \Real^d$ such that $\norm{\zhv}_2 = 1$, $\shv_i^{\top}\zhv$ has identical distribution as $\shv_i^{\top}\ev_1$, where $\ev_1$ is $[1, 0, \ldots, 0]^{\top} \in \Real^d$.
Concentration of measure for uniform distribution on the unit sphere implies,

\[\mathbb{P}\left[|\shv_i^{\top}\zhv| \geq t\right] = \mathbb{P}\left[|s_1| \geq t\right] \leq 2e^{-dt^2/2}.\]

A subsequent union bound gives us,

\[\mathbb{P}\left[\max_{i \in [D]}|\shv_i^{\top}\zhv| \geq t\right]\leq 2De^{-dt^2/2}.\]
Choosing $t = 2\sqrt{\frac{\log(2D)}{d}}$ yields,

\[\mathbb{P}\left[\norm{\Sv^{\top}\zv}_{\infty} \geq 2\sqrt{\frac{\lambda \log(2D)}{D}}\norm{\zv}_2\right] \leq \frac{1}{2D}.\]
This completes the proof.
\end{proof}

\begin{lemma_boxed}
\label{lemma:near_democratic_dynamic_range_randomized_Hadamard}
For $\Sv = \Pv \Dv \Hv \in \Real^{d \times D}$, with probability at least $1 - \frac{1}{2D}$, the near-democratic embedding of $\zv$ with respect to $\Sv$ satisfies,
\begin{equation}
    \norm{\xv_{nd}}_{\infty} \leq 2\sqrt{\frac{\log\left(2D\right)}{D}}\norm{\zv}_2.
\end{equation}
\end{lemma_boxed}
\begin{proof}
For the purposes of this proof, let $\wv = \Pv^{\top}\zv = [w_1, \ldots, w_D]^{\top} \in \Real^D$ and $\uv = \Hv\Dv\wv = [u_1, \ldots,u_N]^{\top}$.
Note that $u_j$ is of the form $\sum_{i=1}^{D}a_iw_i$ with each $a_i = \pm\frac{1}{\sqrt{D}}$ chosen i.i.d.
A Chernoff-type argument gives,
\[\mathbb{P}[u_j > t] = \mathbb{P}\left[e^{\lambda u_j} > e^{\lambda t}\right] \leq e^{-\lambda t}\prod_{i=1}^{D}\mathbb{E}\left[e^{\lambda a_i w_i}\right].\]
Moreover,
\[\mathbb{E}\left[e^{\lambda a_i w_i}\right] = \frac{1}{2}e^{\frac{\lambda}{\sqrt{D}}w_i} + \frac{1}{2}e^{-\frac{\lambda}{\sqrt{D}}w_i} = \cosh\left(\frac{\lambda w_i}{\sqrt{D}}\right) \leq e^{\lambda^2w_i^2/(2D)},\] 
where the last inequality follows from a bound on hyperbolic cosine.
This gives us,
\[\mathbb{P}\left[u_j > t\right] \leq e^{\frac{\lambda^2}{2D}\norm{\wv}_2^2 - \lambda t}.\]
Letting $\lambda = tD/\norm{\wv}_2$ yields $\mathbb{P}\left[u_j > t\right] \leq e^{-t^2D/(2\norm{\wv}_2^2)}$.
Similarly, $\mathbb{P}\left[u_j < -t\right] \leq e^{-t^2D/(2\norm{\wv}_2^2)}$.
Since $\norm{\uv}_{\infty} = \max_{j \in [D], s \in \{\pm 1\}}s u_j$, union bound yields,
\[\mathbb{P}\left[\norm{\Hv\Dv\wv}_{\infty} > t\right] \leq e^{-\frac{t^2D}{2\norm{\wv}_2^2} + \log(2D)}.\]
Setting $t = 2\norm{\wv}_2\sqrt{\frac{\log(2D)}{D}}$ gives us,
\[\mathbb{P}\left[\norm{\Hv\Dv\wv}_{\infty} \leq 2\norm{\wv}_2\sqrt{\frac{\log(2D)}{D}}\right] \geq 1 - \frac{1}{2D}.\]
As $\wv = \Pv^{\top}\zv \implies \norm{\wv}_2 = \norm{\zv}_2$ and this completes the proof.
\end{proof}

The above lemmas imply that with high probability, near-democratic embeddings satisfy $\norm{\xv_{nd}}_{\infty} \leq O\left(\frac{\log D}{D}\right)$.
Lemma \ref{lem:near_democratic_quantizer_worst_case_error} utilizes this fact to give the worst-case quantization error of \textbf{Near-Democratic Quantization} as described in eq. \eqref{eq:near_democratic_quantization_shape} in the main paper.

\begin{lemma_boxed}[\textbf{Worst-case $\ell_2$-error of near-democratic quantization}]
\label{lem:near_democratic_quantizer_worst_case_error}
Given the frame $\Sv = \Pv \Dv \Hv \in \Real^{d \times D}$ and a uniform scalar quantizer $\text{Q}_{u,B}$, for any $\zv \in \Real^d$ such that $\norm{\zv}_2 \leq 1$, let $\ztv = \text{Q}_{nd,D}(\zv) = \Sv \cdot \text{Q}_{u,B}(\xv_{nd})$.
Then, for $\lambda = D/d$,
\begin{equation}
    \sup_{\norm{\zv}_2 \leq 1} \norm{\ztv - \zv}_2 \leq 2^{\Paren{2 - \frac{B}{\lambda}}}\sqrt{\log(2D)}.
\end{equation}
\end{lemma_boxed}
\begin{proof}
The proof is similar to that of Lemma \ref{lem:democratic_quantizer_worst_case_error} and makes use of Lemma \ref{lemma:near_democratic_dynamic_range_randomized_Hadamard}.
\end{proof}

\section{Proof of Theorem \ref{thm:NDQ_guarantee}: Near-Democratic Quantization (NDQ) guarantee}
\label{app:NDQ_guarantee_proof}

We refer to Lemma \ref{lem:near_democratic_quantizer_worst_case_error} from App. \ref{app:near_democratic_embeddings} that gives the worst-case quantization error for near-democratic quantizer whose input is a vector on the unit sphere $\Sbb^{d-1}$.
We consider $\Sv$ to be the randomized Hadamard construction $\Pv \Dv \Hv$.
Since we compute the near-democratic embedding of $\sv_{nd}$ that satisfies $\lVert \sv_{nd} \rVert_2 = 1$, Lemma \ref{lem:near_democratic_quantizer_worst_case_error} states that with probability at least $1 - \frac{1}{2D}$,
\begin{equation}
    \left\lVert \stv - \text{Q}_{nd,B}\Paren{\frac{\Xinv \yv}{\lVert \Xinv\yv \rVert_2}} \right\rVert_2 \leq 2^{\Paren{2 - \frac{B}{\lambda}}}\sqrt{\log (2D)},
\end{equation}
where $\lambda = D/d$ is the aspect ratio of $\Sv$, usually taken to be as close to $1$ as possible so that a Hadamard matrix of size $D \times D$ can be constructed.
The proof of Thm. \ref{thm:NDQ_guarantee} is similar to the proof of Thm. \ref{thm:DQ_guarantee} provided in App. \ref{app:DQ_guarantee_proof}, except for the analysis of the quantization error term $A_1$ in eq. \eqref{eq:risk_decomposition_DQ}.
Term $A_1$ is again decomposed as in eq. \eqref{eq:naive_term_A1_decomposition} for naive quantizers to get the term $\frac{2}{d}\gammah^2\lVert \Xinv\yv \rVert_2^2 \cdot \left\lVert \stv - \frac{\Xinv\yv}{\lVert \Xinv \yv \rVert_2} \right\rVert_2^2$, which is subsequently upper bounded as,
\begin{align*}
    \frac{2}{d}\gammah^2\lVert \Xinv\yv \rVert_2^2 \cdot \left\lVert \stv - \frac{\Xinv\yv}{\lVert \Xinv \yv \rVert_2} \right\rVert_2^2 &\leq \frac{2}{d}\cdot\frac{d\bhat^4}{\Paren{\bhat^2 + \frac{\sigma^2}{d}\xi}}\cdot 2^{2\Paren{2 - \frac{B}{\lambda}}}\log(2D) + O_P\Paren{\frac{1}{D}} \\
    &\leq \frac{32b^4\log (2\lambda d)}{b^2 + \frac{\sigma^2}{d}\xi}\cdot 2^{-\frac{2B}{\lambda}} + O_P\Paren{\frac{1}{\sqrt{d}}}
\end{align*}
Proceeding as before, a uniform upper bound over all models in $\frac{1}{d}\lVert \thetav \rVert_2^2 \leq c^2$ can be obtained as,

\begin{equation}
    \mathbb{P}\Paren{R(\thetatv, \thetav) > \frac{2c^2\sigma^2}{\sigma^2 + c^2\sigma_{\mathrm{min}}^2} + \frac{64c^4\sigma_{\mathrm{min}}^2\log(2\lambda d)}{\sigma^2 + c^2\sigma_{\mathrm{min}}^2}\cdot 2^{-\frac{2B}{\lambda}} + \frac{C_{nd}}{\sqrt{d}}} \to 0
\end{equation}

as $d \to \infty$ for some constant $C_{nd}$ independent of $d$.

\section{Non-asymptotic lower bound for finite dimension \texorpdfstring{$d$}{d}}
\label{app:finite_dimension_lower_bound}

To remove the asymptoticity in our lower bound result, we note that the only step that makes use of the fact $d \to \infty$ is the proof of Lemma \ref{lem:inequality_supremum_bayes_error} in App. \ref{app:proof_inequality_supremum_bayes_error}.
More specifically, we consider \eqref{eq:splitting_bayes error} that splits the Bayes' error over the whole of $\Real^d$ as the sum of integrals over the parameter space $\Thetav$ and $\Real^d\setminus\Thetav$ as,
\begin{equation}
    \int_{\thetav \in \Real^d} R(\Qsf(\Xv,\yv), \thetav) dF(\thetav) \leq \sup_{\thetav \in \Thetav} R(\Qsf(\Xv,\yv), \thetav) + \int_{\thetav \in \Thetaov}R(\Qsf(\Xv,\yv), \thetav) dF(\thetav)
\end{equation}
Recalling eq. \eqref{eq:upper_bound_bayesian_error_over_theta_complement} from App. \ref{app:proof_inequality_supremum_bayes_error}, when $d \centernot\rightarrow \infty$, this term can be upper bounded as before, i.e. 
\begin{align*}
    \int_{\thetav \in \Thetaov}R(\Qsf(\Xv,\yv), \thetav) dF_{\delta}(\thetav) &\leq \frac{2}{d}\sqrt{2}\exp\Paren{-\frac{d(1 - \delta^2)^2}{16\delta^4}}Cd + 2c^2\exp\Paren{-\frac{d(1- \delta^2)^2}{8\delta^4}},
\end{align*}
for some constant, independent of $d$ that can be evaluated.
This gives us, for a fixed $\delta \in (0,1)$,
\begin{align}
\label{eq:lower_bound_worst_case_error_finite_d}
   \sup_{\thetav \in \Thetav} R(\Qsf(\Xv,\yv), \thetav) \geq \int_{\thetav \in \Real^d} R(\Qsf(\Xv,\yv), \thetav) dF_{\delta}(\thetav) &- \frac{2}{d}\sqrt{2}\exp\Paren{-\frac{d(1 - \delta^2)^2}{16\delta^4}}Cd \nonumber \\
   &- 2c^2\exp\Paren{-\frac{d(1- \delta^2)^2}{8\delta^4}}
\end{align}
Moreover, for the choice of prior $\thetav \sim \Ncal\Paren{\mathbf{0}, c^2\delta^2\Iv_d}$, Lemma \ref{lem:lower_bound_MI_Gaussian_prior} can be modified to lower bound the first term in eq. \eqref{eq:lower_bound_worst_case_error_finite_d} as,
\begin{align}
\label{eq:lower_bound_worst_case_error_finite_d_2}
    \int_{\thetav \in \Real^d} R(\Qsf(\Xv,\yv), \thetav) dF_{\delta}(\thetav) \geq &\frac{c^2\delta^2\sigma^2}{\sigma^2 + c^2\delta^2\sigma_{\mathrm{min}}^2} + \frac{c^4\sigma_{\mathrm{min}}^2}{\sigma^2 + c^2\delta^2\sigma_{\mathrm{min}}^2}\cdot 2^{-2B}\nonumber \\ 
    &- \frac{2}{d}\sqrt{2}\exp\Paren{-\frac{d(1 - \delta^2)^2}{16\delta^4}}Cd - 2c^2\exp\Paren{-\frac{d(1- \delta^2)^2}{8\delta^4}}
\end{align}
Note that inequality \eqref{eq:lower_bound_worst_case_error_finite_d_2} holds true for any $\delta \in (0,1)$.
The particular choice of $\delta \in (0,1)$ can be optimized to maximize the right hand side expression so as to ensure that the lower bound is tight.
It is worthwhile to note that since the last two terms are exponentially decaying in $d$ while the first two terms are independent of $d$, for large enough (but finite) $d$, the lower bound is non-negative.
Moreover, for the optimal choice of $\delta$, it is difficult to obtain a neat closed form expression for the R.H.S. of eq. \eqref{eq:lower_bound_worst_case_error_finite_d_2}.

\section{Extension to general noise models}
\label{app:general_noise_models}

Although we assume model \ref{eq:noisy_planted_model} for deriving tight lower and upper bounds, our proposed model quantization algorithms are not limited to Gaussian noise.
For any arbitrary distribution of the noise $\vv$, suppose our quantized model is given by $Q(\Xinv)\yv$, where $Q$ can be $Q_{u,B}$, or $Q_{d,B}$, or $Q_{nd,B}$.
The $\ell_2$-risk can be upper bounded by,
\begin{equation}
    \frac{1}{d}\mathbb{E}\lVert Q(\Xinv\yv) - \thetav \rVert_2^2 \leq \frac{2}{d}\mathbb{E}\lVert \Xinv\yv - \thetav \rVert_2^2 + \frac{2}{d}\mathbb{E}\lVert Q(\Xinv\yv) - \Xinv\yv \rVert_2^2.
\end{equation}
The first term, $\mathbb{E}\lVert \Xinv\yv - \thetav \rVert_2^2 = \mathbb{E}\lVert \Xinv\vv \rVert_2^2 = \Tr{{\Xinv}^{\top}\Xinv \Cov{\vv}}$ contributes to the inherent learning risk, and is independent of the quantization scheme used.
For \textit{na\"ive quantization}, in the worst-case the second term is,
\begin{equation}
    \sup_{\Xv, \yv}\mathbb{E}\lVert Q_{u,B}(\Xinv\yv) - \Xinv\yv \rVert_2^2 = 2^{-2B}d.
\end{equation}
For \textit{democratic quantization}, from Lemma \ref{lem:democratic_quantizer_worst_case_error}, in the worst case, we have,
\begin{equation}
    \sup_{\Xv, \yv}\mathbb{E}\lVert Q_{u,B}(\Xinv\yv) - \Xinv\yv \rVert_2^2 \leq 2^{2\Paren{1 - \frac{B}{\lambda}}}K_u,
\end{equation}
and for \textit{near-democratic quantization}, from Lemma \ref{lem:near_democratic_quantizer_worst_case_error},
\begin{equation}
\sup_{\Xv, \yv}\mathbb{E}\lVert Q_{u,B}(\Xinv\yv) - \Xinv\yv \rVert_2^2 \leq 2^{2\Paren{2 - \frac{B}{\lambda}}}\log(2\lambda d).
\end{equation}
These results do not require the Gaussianity of noise $\vv$, and clearly, the scaling of the $\ell_2$-risk with dimension for \textbf{DQ} and \textbf{NDQ} are better than that of the na\"ive strategy.

\section{Extension to \texorpdfstring{$2$}{2}-Layer scalar output ReLU Neural Networks}
\label{app:NN_analysis}

A scalar output $2$-layer neural network can be represented as:
\begin{equation}
\label{eq:scalar_output_NN}
    f(\xv) = \wv^{\top}\Paren{\Wv\xv}_+ = \sum_{j=1}^{m}\Paren{\wv_j^\top\xv}_+w_j \in \Real.
\end{equation}
Here, $f(\xv)$ denotes the scalar output of the neural network, $(\cdot)_+ = \max\{\cdot, 0\}$ is the ReLU activation, $\Wv \in \Real^{m \times d}$ denotes the weight matrix for the first layer, $\xv \in \Real^d$ is the input, $m$ is the number of neurons in the first (hidden) layer.
$\{\wv_1^\top, \ldots, \wv_m^\top\} \in \Real^d$ denote the rows of the weight matrix $\Wv$, whereas $\wv \triangleq [w_1, \ldots, w_m]^{\top} \in \Real^m$ denote the weights of the second dense layer.
Note that for every $j$, we have $w_j = \lvert w_j \rvert \cdot \text{Sign}(w_j)$.
Consequently, eq. \eqref{eq:scalar_output_NN} can be re-parametrized by absorbing $\lvert w_j \rvert$ within the $(\cdot)_+$ function.
Setting $\whv_j = \wv_j\lvert w_j \rvert$, we get,
\begin{align}
\label{eq:scalar_output_NN_reparametrized}
    f(\xv) = \sum_{j=1}^{m}\Paren{\whv_j^\top\xv}_+\cdot \text{Sign}(w_j) &= \sum_{j : w_j \geq 0}\Paren{\whv_j^\top\xv}_+ - \sum_{j : w_j < 0}\Paren{\whv_j^\top\xv}_+ 
\end{align}
Let us denote the set $\{j : w_j \geq 0\}$ as $\mathrm{M}_+$ and $\{j : w_j < 0\}$ as $\mathrm{M}_-$.
Suppose we quantize the weights of each layer to get $\Wtv$ and $\wtv$ instead of the unquantized weights, $\Whv$ and $\whv$.
Here, $\Whv$ is the matrix whose rows are $\{\whv_1^{\top}, \ldots, \whv_m^{\top}\}$.
The output of the quantized neural network can be written as $\ft(\xv) = \sum_{j=1}^{m}\Paren{\wtv_j^\top\xv}_+\cdot \text{Sign}(\wt_j)$, where $\wtv_1^{\top}, \ldots, \wtv_m^{\top}$ are the rows of $\Wtv$, and $\wtv \triangleq [\wt_1, \ldots, \wt_m]$.
We then have the error in the scalar output of the quantized neural network as,
\begin{align*}
    \left\lvert \ft(\xv) - f(\xv)\right\rvert &= \left\lvert \sum_{j \in M_+}\left[\Paren{\whv_j^\top\xv}_+ - \Paren{\wtv_j^{\top}\xv}_+ \right] - \sum_{j \in M_-}\left[\Paren{\whv_j^\top\xv}_+ - \Paren{\wtv_j^{\top}\xv}_+ \right] \right\rvert \\
    &\leq \sum_{j \in M_+}\left\lvert \Paren{\whv_j^\top\xv}_+ - \Paren{\wtv_j^{\top}\xv}_+ \right\rvert + \sum_{j \in M_-}\left\lvert \Paren{\whv_j^\top\xv}_+ - \Paren{\wtv_j^{\top}\xv}_+ \right\rvert \\
    &= \sum_{j=1}^{m}\left\lvert \Paren{\whv_j^\top\xv}_+ - \Paren{\wtv_j^{\top}\xv}_+ \right\rvert \\
    &\leq \sum_{j=1}^{m}\left\lvert \Paren{\Paren{\whv_j - \wtv_j}^\top\xv}_+ \right\rvert \leq \sum_{j=1}^{m}\norm{\whv_j - \wtv_j}_2\cdot \norm{\xv}_2 \\
    &\leq \sqrt{m}\sqrt{\sum_{j=1}^{m}\norm{\whv_j - \wtv_j}_2^2}\cdot\norm{\xv}_2 = \sqrt{m}\norm{\Whv - \Wtv}_{\rm F}\norm{\xv}_2.
\end{align*}

We essentially vectorize $\Whv \in \Real^{m \times d}$ and apply \textbf{DQ} or \textbf{NDQ} for quantizing the resultant vector.
Since $\lVert \Whv - \Wtv \rVert_{\rm F} = \lVert \mathrm{Vec}(\Whv) - \mathrm{Vec}(\Wtv) \rVert_2$, this suggests us that vectorizing the weight matrices and subsequently using \textit{democratic} or \textit{near-democratic} quantization strategies should prove beneficial over the na\"ive quantizaition strategy.
For simplicity of analysis, we assume $\norm{\xv}_2 \leq 1$ and $\lVert\mathrm{Vec}(\Whv)\rVert_2 \leq 1$.
This assumption can be relaxed by taking in account appropriate scaling that depends on the Frobenius norm of the weight matrix.

From Lemma \ref{lem:near_democratic_quantizer_worst_case_error}, which is a guarantee on the quantization error of \textbf{NDQ} for any general vector input, if we use \textbf{NDQ} with a randomized Hadamard frame to quantize $\mathrm{Vec}(\Whv) \in \Real^{md}$ with a bit-budget of $B$-bits per weight (which we do in our numerical simulations), we get,
\begin{align*}
    \lVert \Wtv - \Whv \rVert_{\rm F} &= \lVert \mathrm{Vec}(\Wtv) - \mathrm{Vec}(\Whv) \rVert_2 \\
    &= \lVert Q_{nd,B}(\mathrm{Vec}(\Whv)) - \mathrm{Vec}(\Whv) \rVert_2 \leq 2^{(2- B)}\sqrt{\log(2md)}.
\end{align*}
This implies, to ensure $\left\lvert \fhat(\xv) - f(\xv) \right\rvert \leq \epsilon$ for some error tolerance $\epsilon > 0$, it suffices to ensure that
\begin{align}
\label{eq:threshold_bit_budget_NN}
     &4\sqrt{m}\cdot2^{-B}\sqrt{\log\Paren{2md}} \leq \epsilon \nonumber \\
     \implies &B \geq \log_2\Paren{\frac{4}{\epsilon}} + \frac{1}{2}\log_2 m + \frac{1}{2}\log_2\log_2 md
\end{align}
Let us define $B^*_{nd}(\epsilon)$ to be the minimum bit budget required to achieve a modeling error of $\epsilon$, i.e. to ensure $\lvert \ft(\xv) - f(\xv) \rvert \leq \epsilon$, where $\ft(\cdot)$ and $f(\cdot)$ denote the neural network functions with quantized and unquantized weights, respectively. 
Then, from inequality \eqref{eq:threshold_bit_budget_NN}, it follows that,
\begin{equation}
    B^*_{nd}(\epsilon) \leq \log_2\Paren{\frac{4}{\epsilon}} + \frac{1}{2}\log_2 m + \frac{1}{2}\log_2\log_2 md
\end{equation}

Inequality \eqref{eq:threshold_bit_budget_NN} gives us an upper bound on the minimum number of bits required to achieve a given error $\epsilon$.
A similar upper bound could be obtained for na\"ive quantization by noting that the worst-case na\"ive quantizer error is upper bounded by
\[\lVert Q_{n,B}(\mathrm{Vec}(\Whv)) - \mathrm{Vec}(\Whv) \rVert_2 \leq 2^{-B}\sqrt{md}.\]
This gives us,
\begin{equation}
    \sqrt{m}\cdot 2^{-B}\sqrt{d} \leq \epsilon \implies B^*_{n}(\epsilon) \leq \log_2\Paren{\frac{1}{\epsilon}} + \frac{1}{2}\log_2(md)  
\end{equation}
Similarly, for democratic quantizer, Lemma \ref{lem:democratic_quantizer_worst_case_error} tells us,
\begin{equation*}
    \lVert Q_{d,B}(\mathrm{Vec}(\Whv)) - \mathrm{Vec}(\Whv) \rVert_2 \leq 2^{\Paren{1 - B}}K_u,
\end{equation*}
giving us,
\begin{equation*}
    2\sqrt{m}\cdot 2^{-B}K_u \leq \epsilon \implies B_d^*(\epsilon) \leq \log_2\Paren{\frac{2K_u}{\epsilon}} + \frac{1}{2}\log_2m. 
\end{equation*}
The following Table \ref{tab:NN_compression_learning_code_comparison} summarizes the results of this section.
Note that the additive $\log_2(md)$ term for naive quantizer can be significantly larger than \textbf{DQ} or \textbf{NDQ}.

\begin{table*}[h!]
\begin{center}
\begin{small}
\begin{sc}
\renewcommand{\arraystretch}{1.8}
\begin{tabular}{lcccr}
\toprule

\textsc{\textbf{Quantization algorithm}} & \textsc{\textbf{Bit budget for $\epsilon$-modeling error}} \\
\midrule
Na\"ive & $B^*_{n}(\epsilon) \leq \log_2\Paren{\frac{1}{\epsilon}} + \frac{1}{2}\log_2(md)$ \\
Democratic & $B_d^*(\epsilon) \leq \log_2\Paren{\frac{2K_u}{\epsilon}} + \frac{1}{2}\log_2m$ \\
Near-Democratic & $B^*_{nd}(\epsilon) \leq \log_2\Paren{\frac{4}{\epsilon}} + \frac{1}{2}\log_2 m + \frac{1}{2}\log_2\log_2 (md)$\\
\bottomrule
\end{tabular}
\end{sc}
\end{small}
\end{center}
\vspace{-2mm}
\caption{\textsc{Comparison of various learning codes for quantizing neural network weights}}
\label{tab:NN_compression_learning_code_comparison}
\vspace{-4mm}
\end{table*}

\end{document}